\documentclass[a4paper,oneside,12pt]{article}

\usepackage{jheppub}
\usepackage{amsfonts}
\usepackage{amsmath}
\usepackage{amssymb}
\usepackage{amsthm}
\usepackage[mathscr]{eucal}
\usepackage{bbold}
\usepackage{braket}
\usepackage{color}
\usepackage{dsfont}
\usepackage{framed}
\usepackage{mathtools}
\usepackage{physics}
\usepackage[normalem]{ulem}
\usepackage{tensor}
\usepackage{thmtools}
\usepackage{thm-restate}
\usepackage{ulem}
\usepackage{tikz}
\usetikzlibrary{arrows.meta,arrows} 
\usepackage{centernot}
\usepackage{enumitem} 
\usepackage[a]{esvect}  
\usepackage{slashed}
\usepackage[only,llbracket,rrbracket]{stmaryrd} 

\usepackage{yfonts}
\usepackage{float}
\usetikzlibrary{math} 
\usetikzlibrary{calc}		

\usepackage{subcaption}
\usepackage{float}
\usepackage{afterpage}
\renewcommand{\thefigure}{\arabic{figure}}

\usetikzlibrary{arrows}
\tikzset{elabelcolor/.style={color=blue} 
    vertex/.style={circle,draw,minimum size=1.5em},
    edge/.style={->,> = latex'}
}  

\usepackage{cleveref}

\usepackage{longtable}

\usepackage[ruled,vlined]{algorithm2e}

\DeclareFontFamily{U}{jkpmia}{}
\DeclareFontShape{U}{jkpmia}{m}{it}{<->s*jkpmia}{}
\DeclareFontShape{U}{jkpmia}{bx}{it}{<->s*jkpbmia}{}
\DeclareMathAlphabet{\mathfrak}{U}{jkpmia}{m}{it}
\SetMathAlphabet{\mathfrak}{bold}{U}{jkpmia}{bx}{it}

\definecolor{VHcolor}{rgb}{0.7,0.3,0.9}

\definecolor{MRocolor}{rgb}{0.1,0,1}

\newtheorem{cor}{Corollary}
\newtheorem{defi}{Definition}
\newtheorem{lemma}{Lemma}

\newtheorem{thm}{Theorem}

\newcommand{\comp}[1]{#1^\complement}
\newcommand{\A}{\mathcal{A}} 
\newcommand{\B}{\mathcal{B}}  
\newcommand{\C}{\mathcal{C}}  
\newcommand{\F}{\mathcal{F}}  

\newcommand{\N}{{\sf N}}
\newcommand{\D}{\sf{D}}

\newcommand{\I}{I}
\newcommand{\J}{J}
\newcommand{\K}{K}
\newcommand{\X}{X}
\newcommand{\Y}{Y}
\newcommand{\Z}{Z}
\newcommand{\W}{W}

\newcommand{\nsp}{\llbracket\N\rrbracket} 
\newcommand{\nspp}{\llbracket\N'\rrbracket} 

\newcommand{\mgs}{\mathcal{E}} 
\newcommand{\ds}{\mathcal{D}} 
\newcommand{\cds}{\comp{\mathcal{D}}} 
\newcommand{\ac}{\Check{\mathcal{A}}}  

\newcommand{\bset}{\beta\text{-set}} 
\newcommand{\bsets}{\beta\text{-sets}}
\newcommand{\bs}[1]{\beta(#1)}
\newcommand{\ess}{\mathfrak{B}^{^{\!\vee}}}

\newcommand{\pos}{\mathfrak{B}^{^{\!+}}}

\newcommand{\bsi}[2]{\beta_{#2}(#1)} 
\newcommand{\mipos}[1]{\beta_{_{\!+}}(#1)}
\newcommand{\mivan}[1]{\beta_{_{\!0}}(#1)}

\newcommand{\mmC}{\text{\textit{\u{C}}}} 

\newcommand{\cl}{\text{cl}} 
\newcommand{\clw}{\text{cl}_{\cup}} 

\newcommand{\ent}{{\sf S}}  
\newcommand{\mi}{{\sf I}}  
\newcommand{\pmi}{\mathcal{P}}   
\newcommand{\cpmi}{\comp{\mathcal{P}}} 
\newcommand{\pmimap}{\pi^{_{\bullet}}} 
\newcommand{\face}{\mathscr{F}}   

\newcommand{\Svec}{\vec{\ent}}  
\newcommand{\Svecp}{\vec{\ent}'}  

\newcommand{\cgmap}{\Theta}
\newcommand{\cginv}{\cgmap^{-1}}
\newcommand{\cgpart}{\cgmap^\circ}
\newcommand{\cgpartrest}[1]{\cgpart_{|#1}}
\newcommand{\cgpmi}{\Lambda_\cgmap}
\newcommand{\cgsvec}{\Lambda^{\!^{\!\bullet}}_{\cgmap}}
\newcommand{\cgmapeq}{\sim_{^\cgmap}}
\newcommand{\cgpmiinv}{\cgpmi^{-1}}
\newcommand{\cgmeet}{\widehat{\Phi}} 
\newcommand{\cgmeetpmi}{\Lambda_{\widehat{\Phi}}} 
\newcommand{\cgindex}{\chi_{_\cgmap}}
\newcommand{\eqcl}[1]{[#1]^{^{\sim}}}

\newcommand{\gf}[1]{{\mathrm{#1}}}
\newcommand{\hp}{\gf{H}_{\mathcal{P}}}  
\newcommand{\hpp}{\gf{H}_{\mathcal{P}'}}  
\newcommand{\lhp}{\gf{L}_{\mathcal{P}}}  
\newcommand{\lhpp}{\gf{L}_{\mathcal{P}'}}

\newcommand{\lat}[1]{\mathcal{L}_{\text{#1}}}

\SetKwComment{Comment}{** }{ **} 
\SetKwInOut{Input}{input}\SetKwInOut{Output}{output}

\definecolor{PminusEcol}{rgb}{1,0.8,0.8}
\definecolor{EminusMcol}{rgb}{1,0.85,0.65}
\definecolor{Mcol}{rgb}{1,0.97,0.5}
\definecolor{Qcol}{rgb}{0.7,1,0.9}
\definecolor{Vcol}{rgb}{0.83,0.83,1}
\definecolor{EminusPcol}{rgb}{0.8,1,0.6}

\definecolor{Ccol}{rgb}{0.9,0.9,0.9}

\definecolor{bulkcol}{rgb}{0.8,0.8,0.8}

\title{On the construction of graph models realizing given entropy
vectors}

\author[a]{Veronika E. Hubeny,}
\emailAdd{veronika@physics.ucdavis.edu}

\author[b]{Massimiliano Rota}
\emailAdd{max.rota@bristol.ac.uk}

\affiliation[a]{Center for Quantum Mathematics and Physics (QMAP)\\ 
Department of Physics \& Astronomy, University of California, Davis, CA 95616 USA}
\affiliation[b]{School of Mathematics, University of Bristol,
Woodland Road, Bristol, BS8 1UG UK}

\abstract{We present an efficient algorithm for the construction of a holographic simple tree graph model that realizes a given entropy vector, subject to a specific ``chordality'' condition first introduced in arXiv:2412.18018. We further develop the toolkit of the correlation hypergraph, particularly in relation to coarse-graining and fine-graining of subsystems. We then use these techniques to take the first steps towards the generalization of this new algorithm to arbitrary (not necessarily simple) holographic tree graph models, and the ``detection'' of unrealizability of an entropy vector independently from the knowledge of holographic entropy inequalities. 
}

\begin{document}
 

\maketitle

\section{Introduction}

In gauge-gravity duality \cite{Maldacena:1997re,Gubser:1998bc,Witten:1998qj}, a long-standing problem is the characterization of states that correspond to bulk classical geometries, and in particular, the understanding of how exactly the bulk geometry is encoded in the boundary theory. The HHRT prescription \cite{Ryu:2006bv,Hubeny:2007xt} for the computation of the von Neumann entropy of a boundary spatial region via the area of bulk extremal surfaces provides deep insights into this encoding. For bulk static spacetimes, it was first found in \cite{Hayden:2011ag} that the RT formula implies that the entropy of boundary spatial subsystems satisfies an inequality, called the \textit{monogamy of mutual information} (MMI), which does not hold in general for arbitrary quantum systems, thereby characterizing ``geometric'' states.\footnote{\,The proof of MMI was generalized to dynamical spacetimes in \cite{Wall:2012uf}. For further details on holographic entropy inequalities in the dynamical case, see e.g.\ \cite{Grado-White:2025jci} and references therein.} New inequalities of this type were then found in \cite{Bao:2015bfa}, which introduced the framework of the \textit{holographic entropy cone} (HEC), and more recently in \cite{Czech:2022fzb,Hernandez-Cuenca:2023iqh,Czech:2024rco}.

As the number of boundary regions (or parties) $\N$ increases, the number of \textit{holographic entropy inequalities} (HEIs) grows substantially (for $\N=6$ there are at least a thousand orbits of inequalities \cite{Hernandez-Cuenca:2023iqh}), although it remains finite \cite{Bao:2015bfa}. For sufficiently large $\N$, however, not only does the complete list of inequalities become practically impossible to compute, but it is also unclear how much insight one could realistically extract from such a list. This difficulty is further increased by the fact that even the sole $\N=3$ inequality (MMI) currently lacks a clear interpretation, even from an information-theoretic perspective.\footnote{\,See \cite{Czech:2025jnw} for first steps in the direction of interpreting HEIs.} Motivated by these considerations and building on the ideas of \cite{Hubeny:2018trv,Hubeny:2018ijt}, a search for a ``fundamental principle'' underlying the structure of the HEC for arbitrary $\N$ was initiated in \cite{Hernandez-Cuenca:2022pst}. It was argued that for any $\N$, the structure of the HEC is intimately connected to the holographically permissible ``patterns'' of simultaneous saturation of multiple instances of the subadditivity inequality, which encode independence among collections of subsystems.

One of the key difficulties, not only in proving, but even already
in testing, the conjecture from \cite{Hernandez-Cuenca:2022pst}, is closely related to the complications encountered when attempting to establish the completeness of a set of known inequalities and to interpret these inequalities: the absence of a systematic method to construct a geometry, or a \textit{holographic graph model},\footnote{\,See \cite{Bao:2015bfa} for a detailed discussion about the relationship between geometric set-ups in holography and holographic graph models.} that realizes a given list of candidate areas of RT surfaces. More precisely, suppose that for a fixed number of regions $\N$ we are given a list of $\D=2^\N-1$ positive numbers. How can we determine whether there exists a geometry and a choice of boundary regions such that these numbers correspond to the areas of the RT surfaces for all (non-empty) collections of these
regions? Of course, this is precisely the purpose of HEIs: if all HEIs for that value of $\N$ are known, one can straightforwardly check whether the given list of numbers—referred to as an ``entropy vector''—satisfies all the inequalities. To advance our understanding of the physical significance of these inequalities, and consequently to extract the general principle underlying the structure of the HEC, we seek alternative methods for addressing this question—namely, a way of ``detecting unrealizability'' independently of the prior knowledge of HEIs. Furthermore, it is important to note that even when all inequalities are known and an entropy vector satisfies them, the information obtained remains limited: one can only conclude the existence of a holographic graph model realizing that vector, but a systematic procedure to explicitly construct such a model is currently lacking. The absence of such a method becomes a critical obstacle in attempts to prove the completeness of a set of HEIs for a given $\N$, as the strategy to answer this question relies on the construction of graph models realizing the \textit{extreme rays} (ERs) of the cone specified by the inequalities.

Building on this perspective, the main contribution of this work is to take a first step towards a systematic construction of holographic graph models for a given entropy vector by proposing an algorithm in a certain restricted context, namely when the entropy vector can be realized by a ``simple tree'' holographic graph model—that is, a graph model with tree topology and distinct boundary vertices. Three particularly notable features of this algorithm are that the condition for its applicability can be verified efficiently (and it is not itself an entropy inequality), it is completely general for any $\N$, and it enables an efficient construction of a candidate graph model. While in this work we do not prove that the algorithm always succeeds, i.e., that the condition for its applicability is both necessary and sufficient, we conjecture that this is the case, and we hope to report on a general proof in the near future \cite{sufficiency}.

The second goal of this work is to take initial steps toward generalizing the algorithm introduced above, extending it from entropy vectors realizable by simple tree graph models to those that can be realized by arbitrary graph models with tree topology—that is, graph models without cycles involving only ``bulk'' vertices. In the course of discussing these strategies, we also comment on potential methods to detect the unrealizability of an entropy vector by one of these graph models, or possibly \textit{any} graph model, independently of prior knowledge of HEIs. As we will remark in the discussion, this possibility is closely connected to the conjecture of \cite{Hernandez-Cuenca:2022pst} regarding the structure of the HEC.

Lastly, both the necessary condition for the applicability of the algorithm presented here, and the algorithm itself, rely heavily on the new tools introduced in \cite{Hubeny:2024fjn} for describing ``patterns of marginal independence'' in the framework of the ``correlation hypergraph''. This object was introduced in \cite{Hubeny:2024fjn} as a more convenient and conceptually suggestive way to represent the set of vanishing instances of the mutual information (equivalently, saturated instances of subadditivity) for a given entropy vector. While \cite{Hubeny:2024fjn} focused on the properties of the correlation hypergraph for a fixed number of parties $\N$, the generalization of our new algorithm from simple trees to arbitrary trees requires the ``refinement'' of individual parties into multiple components—or, in a holographic set-up, the replacement regions with collections of regions—which may then be relabeled as new parties. The third contribution of this work is a discussion of the properties of the correlation hypergraph under such fine-grainings, and the development of the toolkit to describe the (simpler) opposite transformation, a coarse-graining, from this perspective.

\paragraph{Structure of the paper:} All essential notions from previous works are reviewed in \S\ref{sec:toolkit}, in particular the concepts of a PMI, KC, and MI-poset in \S\ref{subsec:review}, and $\bsets$ and the correlation hypergraph from \cite{Hubeny:2024fjn} in \S\ref{subsec:correlation-hypergraph}. On the other hand, the content of \S\ref{subsec:kc-lattice} and \S\ref{subsec:cg-and-fg} is largely new. Specifically, in \S\ref{subsec:kc-lattice} we examine the KC-lattice, first introduced in \cite{He:2022bmi}, from the perspective of the correlation hypergraph, while \S\ref{subsec:cg-and-fg} focuses on transformations of PMIs under coarse-graining and fine-graining maps, where the number of parties is no longer fixed. The basics of holographic graph models are reviewed in \S\ref{subsec:graph-review}. The algorithm for the construction of holographic graph models is the focus of \S\ref{sec:strategies}. Given an arbitrary entropy vector for any number of parties, \S\ref{subsec:preliminary-steps} discusses preliminary steps to verify the necessary condition for realizability by simple tree graph models \cite{Hubeny:2024fjn}, as well as additional steps that, for certain entropy vectors, allow a reduction of the problem to fewer parties. The algorithm is then presented in \S\ref{subsec:chordal-case}, while \S\ref{subsec:non-chordal-case} comments on the generalization to arbitrary tree graph models and on methods for detecting unrealizability. Finally, \S\ref{sec:discussion} concludes with a discussion of future directions for this program and its connection to the conjecture of \cite{Hernandez-Cuenca:2022pst}. Throughout this work, the general theory developed in the main text is complemented by figures with detailed captions, which illustrate and further elucidate particular constructions and serve as a concrete counterpart to the more high-level presentation.

\paragraph{Guide for the reader:} Readers primarily interested in the new algorithm for the construction of simple tree graph models can skip subsections \S\ref{subsec:kc-lattice} and \S\ref{subsec:cg-and-fg}, and, if already familiar with the basics of holographic graph models and the correlation hypergraph from \cite{Hubeny:2024fjn}, may proceed directly to \S\ref{sec:strategies}, in particular \S\ref{subsec:preliminary-steps} and \S\ref{subsec:chordal-case}. The discussion in \S\ref{subsec:non-chordal-case} relies on a few notions related to coarse-grainings and fine-grainings of subsystems; however, on a first reading, the basic definitions at the beginning of \S\ref{subsec:cg-and-fg}, and the overview of these transformations for holographic graph models in \S\ref{subsec:graph-review}, may be sufficient. Subsection \S\ref{subsec:cg-and-fg} examines these transformations in greater depth and is particularly relevant for readers focused on the general toolkit of the correlation hypergraph, rather than on specific holographic applications.

\paragraph{Conventions:} For an arbitrary positive integer $k$, we denote the set $\{1,\ldots,k\}$ by $[k]$. To simplify the presentation, we typically use the same notation for an element $a$ of a set $A$, and the subset of $A$ whose only element is $a$. For example, we write $A\setminus a$ instead of $A\setminus \{a\}$. The complement of a subset $B$ of a given set (which will be clear from context and hence left implicit) is denoted by $\comp{B}$.

\section{A correlation hypergraph toolkit}
\label{sec:toolkit}

In this section we present the toolkit of the correlation hypergraph in detail. We emphasize that while \S\ref{subsec:review} and \S\ref{subsec:correlation-hypergraph} review earlier work, \S\ref{subsec:kc-lattice} and \S\ref{subsec:cg-and-fg} are mostly new. Specifically, in \S\ref{subsec:review} we summarize the basic definitions concerning PMIs, the MI-poset, and Klein's condition, which have been discussed in several previous works. The notion of $\bsets$ and the correlation hypergraph introduced in \cite{Hubeny:2024fjn}, together with its main properties, are reviewed in \S\ref{subsec:correlation-hypergraph}. As first shown in \cite{He:2022bmi}, the set of PMIs that obey Klein's condition at any given $\N$ has the structure of a lattice. This result is reviewed in \S\ref{subsec:kc-lattice}, where we examine this lattice from the perspective of the correlation hypergraph. Finally, in \S\ref{subsec:cg-and-fg}, we formalize the definitions of coarse-grainings of entropy vectors and PMIs,\footnote{\,Some of these notions were already discussed in \cite{Hernandez-Cuenca:2022pst}, although mainly from the perspective of holographic graph models rather than PMIs.} as well as the closely related notion of fine-grainings, and explain how these transformations can be more conveniently—and suggestively—described in terms of the correlation hypergraph.

\subsection{Review of basic concepts and definitions}
\label{subsec:review}

For a fixed number of parties $\N$, and quantum density matrix $\rho_\N$, the \textit{entropy vector} $\Svec(\rho_\N)$ is the collection
(in some conventional order, which we choose to be lexicographic) of the von Neumann entropy of its reduced density matrices, for all possible non-empty collections $\I$ of the $\N$ parties ($\varnothing\neq\I\subseteq [\N]$). The space $\mathbb{R}^{\D}$, with $\D=2^\N-1$, where these entropy vectors live is called \textit{entropy space}. Not all vectors in entropy space are vectors of entropies of a density matrix; we will call an arbitrary vector $\Svec$ in entropy space an entropy vector, and when there is a density matrix whose entropy vector is $\Svec$, we will say that $\Svec$ is \textit{realizable}. We will typically implicitly consider a purification of a density matrix $\rho_\N$ by an additional party, called the \textit{purifier}, and labeled conventionally by $0$. The set of $\N$ parties and purifier is denoted by $\nsp$. To distinguish the subsets of $[\N]$, labeled by $\I,\J,\K$ (and assumed to be non-empty), from the (possibly empty) subsets of $\nsp$, we will denote the latter by $\X,\Y,\Z$.\footnote{\,This notation is the same as the one adopted in \cite{Hubeny:2024fjn}, but differs from the one used in several earlier works (e.g., \cite{He:2022bmi}), where collections that could include $0$ were denoted by $\underline{\mathscr{I}},\underline{\mathscr{J}},\underline{\mathscr{K}}$.} 

Entropy inequalities, like \textit{subadditivity} (SA) and \textit{strong subadditivity} (SSA), bound the region of entropy space that contains realizable entropy vectors. The $\N$-party \textit{subadditivity cone} (SAC$_\N$) is an obvious outer bound to this region, and will play a key role in this work.\footnote{\,Of course one can obtain a more stringent bound by also imposing SSA, but this inequality will not play any role here. Intuitively, one reason for this is that in the holographic setting, because of MMI, SSA cannot be saturated without also saturating SA. In other words, SSA is a redundant inequality for the description of the HEC.\label{ft:SSA}} 
The SAC$_\N$ is defined as the pointed polyhedral cone specified by all instances of SA for $\N$ parties
\begin{equation}
\label{eq:sac-def}
    \mi(\X:\Y)\geq0,\ \ \ \forall\,\X,\Y\;\; \text{with}\;\; \X\cap\Y=\varnothing,\; \ \X,\Y\neq\varnothing,
\end{equation}
where we have written an instance of SA in terms of the corresponding instance of the \textit{mutual information} (MI). In \eqref{eq:sac-def}, notice that when one of the arguments of an MI instance contains the purifier, $\mi(\X:\Y)\geq0$ is not written as an inequality in entropy space. In these cases, we always implicitly assume to have replaced each term $\ent_\Z$ in the inequality where $0\in\Z$ with $\ent_{\comp{\Z}}$, effectively mapping an instance of SA to one of the \textit{Araki-Lieb} inequality.

Even if the SAC is a rather weak outer bound on the region of entropy space accessible to quantum states, this of course does not imply that some its faces cannot be reached by realizable entropy vectors. If an entropy vector $\Svec$ belongs to a face of the SAC with dimension less than $\D$, this means that at least one instance of the mutual information $\mi(\X:\Y)$ vanishes for $\Svec$, and if $\Svec$ is realizable by a density matrix, it must be that the corresponding marginal factorizes $\rho_{\X\Y}=\rho_\X\otimes\rho_\Y$, i.e., that the subsystems $\X$ and $\Y$ are uncorrelated \cite{Nielsen_Chuang_2010}. For this reason, the collection of all vanishing MI instances for an entropy vector $\Svec$ in the \textit{interior} of a face $\face$ of the SAC is called a \textit{pattern of marginal independence} (PMI). Notice that for any given face, all vectors in its interior share the same PMI. The importance of $\Svec$ being in the interior of $\face$ is that if $\Svec$ lies instead on the boundary of $\face$, the set of vanishing MI instances is larger, corresponding to another face and therefore another PMI. For a given value of $\N$ and entropy vector $\Svec$ which obeys all instances of SA, the PMI $\pmi$ \textit{of} $\Svec$ is then defined as the subset of the MI-poset given by
\begin{equation}
\label{eq:pmi-map-def}
    \pmi=\pmimap(\Svec) \coloneq \{\mi(\X:\Y)\ |\ \; \mi(\X:\Y)(\Svec)=0\},
\end{equation}
where $\mi(\X:\Y)(\Svec)$ denotes the value of $\mi(\X:\Y)$ for $\Svec$, and we have introduced a map $\pmimap$ which associates to an entropy vector the corresponding PMI. The dot in the notation $\pmimap$ is meant to indicate that the map acts on entropy vectors, rather than on extended objects like faces or subspaces (we will use a similar notation for other maps below). We will return to this correspondence between PMIs and the interior of a face of the SAC in \S\ref{subsec:cg-and-fg}.

A natural question, intimately connected to the understanding of fundamental constraints on the entropy of subsystems in quantum mechanics, is which faces of the SAC can be reached by realizable entropy vectors, or equivalently, which PMIs are \textit{realizable}. This question was termed the \textit{quantum marginal independence problem} (QMIP) in \cite{Hernandez-Cuenca:2019jpv}. Naturally, as we shall explore in later sections, the same question can also be posed for restricted classes of states, such as classical probability distributions, stabilizer states, or geometric states in holography (equivalently, holographic graph models). Note that the realizability of a PMI is weaker than the realizability of an entropy vector, since the fact that a PMI is realizable neither implies nor requires that all entropy vectors in the interior of the corresponding face of the SAC are themselves realizable.

Already for small values of $\N$, it is evident that the vast majority of PMIs are not realizable by quantum states due to SSA, and one should therefore restrict attention to PMIs that are SSA-\textit{compatible}.\footnote{\,A PMI is \textit{compatible} with SSA if in the interior of the corresponding face of the SAC there exists at least one entropy vector which satisfies all instances of SSA.} One could, of course, address this issue by considering a smaller cone, defined not only by SA but also by SSA. However, this approach would introduce many additional faces, which are at least partially associated with the saturation of SSA, rather than with SA alone. We will instead concentrate on the faces of the SAC, motivated by the fact that the conjecture in \cite{Hernandez-Cuenca:2022pst} connects the structure of the HEC to marginal independence, rather than \textit{conditional} independence (see also \Cref{ft:SSA}).

To restrict the set of faces of the SAC that merit consideration, it is convenient, following \cite{He:2022bmi}, to impose an additional condition on PMIs which, while slightly weaker than SSA, has the advantage of being expressible entirely as a combinatorial constraint. This condition, referred to as \textit{Klein's condition} (KC) in \cite{He:2022bmi}, simply requires that for each MI instance $\mi(\X:\Y)$ in a PMI $\pmi$, all MI instances of the form $\mi(\X':\Y')$, where $\X'\subseteq\X$ and $\Y'\subseteq\Y$ (or $\X'\subseteq\Y$ and $\Y'\subseteq\X$), are also included in $\pmi$. A PMI satisfying this requirement is called a KC-PMI, and we leave it as a straightforward exercise for the reader to verify that KC is implied by SSA.\footnote{\,As shown in \cite{He:2022bmi}, at $\N=3$ all KC-PMIs are SSA-compatible. However, this is no longer true for $\N\geq 4$, demonstrating that, as a constraint on the set of realizable PMIs, KC is in general strictly weaker than SSA (for entropy vectors, this is immediate).}

\begin{figure}[tb]    
    \centering
        \begin{tikzpicture}
    \node (0a123) at (2.2,3) {{\scriptsize $\mi(0\!:\!123)$}};
     \node (1a023) at (3.7,3) {{\scriptsize $\mi(1\!:\!023)$}};
     \node (2a013) at (5.2,3) {{\scriptsize $\mi(2\!:\!013)$}};
     \node (3a012) at (6.7,3) {{\scriptsize $\mi(3\!:\!012)$}};
     \node (01a23) at (8.2,3) {{\scriptsize $\mi(01\!:\!23)$}};
     \node (02a13) at (9.7,3) {{\scriptsize $\mi(02\!:\!13)$}};
     \node (03a12) at (11.2,3) {{\scriptsize $\mi(03\!:\!12)$}};
     
    \node (01a2) at (0.1,1) {{\scriptsize $\mi(01\!:\!2)$}};
    \node (02a1) at (1.2,1) {{\scriptsize $\mi(02\!:\!1)$}};
    \node (12a0) at (2.3,1) {{\scriptsize $\mi(12\!:\!0)$}};
    \node (01a3) at (3.7,1) {{\scriptsize $\mi(01\!:\!3)$}};
    \node (03a1) at (4.8,1) {{\scriptsize $\mi(03\!:\!1)$}};
    \node (13a0) at (5.9,1) {{\scriptsize $\mi(13\!:\!0)$}};
    \node (02a3) at (7.3,1) {{\scriptsize $\mi(02\!:\!3)$}};
    \node (03a2) at (8.4,1) {{\scriptsize $\mi(03\!:\!2)$}};
    \node (23a0) at (9.5,1) {{\scriptsize $\mi(23\!:\!0)$}};
    \node (12a3) at (10.9,1) {{\scriptsize $\mi(12\!:\!3)$}};
    \node (13a2) at (12,1) {{\scriptsize $\mi(13\!:\!2)$}};
    \node (23a1) at (13.1,1) {{\scriptsize $\mi(23\!:\!1)$}};

    \node (0a1) at (1.8,-1) {{\scriptsize $\mi(0\!:\!1)$}};
     \node (0a2) at (3.8,-1) {{\scriptsize $\mi(0\!:\!2)$}};
     \node (1a2) at (5.8,-1) {{\scriptsize $\mi(1\!:\!2)$}};
     \node (0a3) at (7.8,-1) {{\scriptsize $\mi(0\!:\!3)$}};
     \node (1a3) at (9.8,-1) {{\scriptsize $\mi(1\!:\!3)$}};
     \node (2a3) at (11.8,-1) {{\scriptsize $\mi(2\!:\!3)$}};

    \draw[-,gray,very thin] (0a123.south) -- (12a0.north);
    \draw[-,gray,very thin] (0a123.south) -- (13a0.north);
    \draw[-,gray,very thin] (0a123.south) -- (23a0.north);
    \draw[-,gray,very thin] (1a023.south) -- (02a1.north);
    \draw[-,gray,very thin] (1a023.south) -- (03a1.north);
    \draw[-,gray,very thin] (1a023.south) -- (23a1.north);
    \draw[-,gray,very thin] (2a013.south) -- (01a2.north);
    \draw[-,gray,very thin] (2a013.south) -- (03a2.north);
    \draw[-,gray,very thin] (2a013.south) -- (13a2.north);
    \draw[-,gray,very thin] (3a012.south) -- (01a3.north);
    \draw[-,gray,very thin] (3a012.south) -- (02a3.north);
    \draw[-,gray,very thin] (3a012.south) -- (12a3.north);
    \draw[-,gray,very thin] (01a23.south) -- (01a2.north);
    \draw[-,gray,very thin] (01a23.south) -- (01a3.north);
    \draw[-,gray,very thin] (01a23.south) -- (23a0.north);
    \draw[-,gray,very thin] (01a23.south) -- (23a1.north);
    \draw[-,gray,very thin] (02a13.south) -- (02a1.north);
    \draw[-,gray,very thin] (02a13.south) -- (13a0.north);
    \draw[-,gray,very thin] (02a13.south) -- (02a3.north);
    \draw[-,gray,very thin] (02a13.south) -- (13a2.north);
    \draw[-,gray,very thin] (03a12.south) -- (12a0.north);
    \draw[-,gray,very thin] (03a12.south) -- (03a1.north);
    \draw[-,gray,very thin] (03a12.south) -- (03a2.north);
    \draw[-,gray,very thin] (03a12.south) -- (12a3.north);

    \draw[-,gray,very thin] (01a2.south) -- (0a2.north);
    \draw[-,gray,very thin] (01a2.south) -- (1a2.north);
    \draw[-,gray,very thin] (02a1.south) -- (0a1.north);
    \draw[-,gray,very thin] (02a1.south) -- (1a2.north);
    \draw[-,gray,very thin] (12a0.south) -- (0a1.north);
    \draw[-,gray,very thin] (12a0.south) -- (0a2.north);
    \draw[-,gray,very thin] (01a3.south) -- (0a3.north);
    \draw[-,gray,very thin] (01a3.south) -- (1a3.north);
    \draw[-,gray,very thin] (03a1.south) -- (0a1.north);
    \draw[-,gray,very thin] (03a1.south) -- (1a3.north);
    \draw[-,gray,very thin] (13a0.south) -- (0a1.north);
    \draw[-,gray,very thin] (13a0.south) -- (0a3.north);
    \draw[-,gray,very thin] (02a3.south) -- (0a3.north);
    \draw[-,gray,very thin] (02a3.south) -- (2a3.north);
    \draw[-,gray,very thin] (03a2.south) -- (0a2.north);
    \draw[-,gray,very thin] (03a2.south) -- (2a3.north);
    \draw[-,gray,very thin] (23a0.south) -- (0a2.north);
    \draw[-,gray,very thin] (23a0.south) -- (0a3.north);
    \draw[-,gray,very thin] (12a3.south) -- (1a3.north);
    \draw[-,gray,very thin] (12a3.south) -- (2a3.north);
    \draw[-,gray,very thin] (13a2.south) -- (1a2.north);
    \draw[-,gray,very thin] (13a2.south) -- (2a3.north);
    \draw[-,gray,very thin] (23a1.south) -- (1a2.north);
    \draw[-,gray,very thin] (23a1.south) -- (1a3.north);
     
    \end{tikzpicture}
    \caption{The Hasse diagram of the $\N=3$ MI-poset.
    }
    \label{fig:N3-MI-poset}
\end{figure}
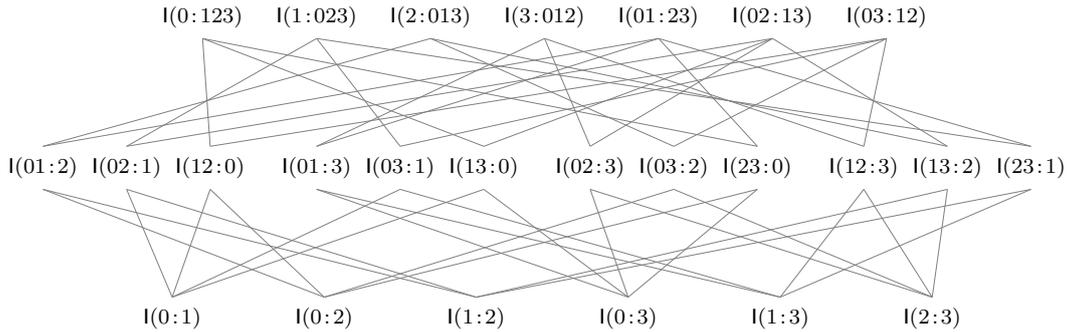

As a condition on the structure of a PMI, KC can be formulated more conveniently by introducing a partial order on the set $\mgs_\N$ of all MI instances 
\begin{equation}
\label{eq:mi_order}
    \mi(\X':\Y')\preceq\mi(\X:\Y) \;\; \iff \;\; \X'\subseteq\X\;\text{{\small and}}\;\Y'\subseteq\Y\;\ \text{or}\;\;\X'\subseteq\Y\;\text{{\small and}}\;\Y'\subseteq\X.
\end{equation}
We call the resulting poset $(\mgs_{\N},\preceq)$ the $\N$-party \textit{mutual information poset}, or simply the MI-\textit{poset} (see \Cref{fig:N3-MI-poset} for a diagrammatic representation of the $\N=3$ case). A PMI $\pmi$ is then KC-PMI if and only if $\pmi$ is a \textit{down-set} in the MI-poset.\footnote{\,A down-set in a poset $\A$ is a subset $\B\subset\A$ such that for every $b\in\B$ and $a\preceq b$, $a\in\B$. An \textit{up-set} is defined analogously, replacing $\preceq$ with $\succeq$.} It is important to keep in mind, however, that an arbitrary down-set $\ds$ in the MI-poset is not necessarily a PMI. The reason for this is that $\ds$ is not necessarily closed with respect to the linear dependence among the MI instances, and even if it is, it is not necessarily compatible with all instances of SA (we will present an example in the next subsection, cf., \Cref{fig:D-not-P}).\footnote{\,More precisely, as discussed in \cite{He:2022bmi}, given $\ds$, there might be some MI instance $\mi\notin\ds$ which is a linear combination of the instances in $\ds$. Furthermore, when this is not the case, one can meaningfully associate to $\ds$ a linear subspace of entropy space (the intersection of the hyperplanes corresponding to the equations of the form $\mi=0$ for all $\mi\in\pmi$), but this subspace is not necessarily spanned by a face of the SAC.}

\subsection{Review of $\bsets$ and the correlation hypergraph}
\label{subsec:correlation-hypergraph}

As reviewed in the previous subsection, any KC-PMI can be conveniently  described as a down-set in the MI-poset, but one of the key results of \cite{Hubeny:2024fjn} was the observation that any KC-PMI also admits a more compact (and much more suggestive) representation in terms of a certain hypergraph, called the \textit{correlation hypergraph} of the PMI. Before we can review the definition, we need to introduce a certain partition of the MI poset and explain how a KC-PMI can be described in terms of this partition.

For any subsystem $\X\subseteq\nsp$ such that $|\X|\geq 2$, we define the \textit{$\bset$ of $\X$} as the set of MI instances whose arguments form a bipartition of $\X$, i.e.,
\begin{equation}
    \bs{\X} =\{\mi(\Y:\Z)\ |\ \; \Y\cup\Z =\X\}.
\end{equation}
The collection of $\bsets$ for all subsystems $\X\subseteq\nsp$ with $|\X|\geq 2$ is a partition of $\mgs_\N$. To simplify the presentation, in what follows we will often keep the constraint $|\X|\geq 2$ implicit when we work with $\bsets$, and will only mention it explicitly when it is particularly relevant. 

Given an arbitrary down-set $\ds$ in the MI-poset (not necessarily a PMI), and a $\bset$ $\bs{\X}$, we say that $\bs{\X}$ is \textit{vanishing} if all its elements are in $\ds$, \textit{positive} if they are all in its complement ($\cds$), and \textit{partial} if at least one is in $\ds$ and one in $\cds$. 
While this nomenclature is of course inspired by the notion of a PMI, here we regard $\ds$ and $\cds$ merely as specific sets of MI instances, without any value assigned to these elements.
We then denote by $\ac$ the antichain in the MI-poset that ``generates'' the up-set $\cds$, i.e.,
\begin{equation}
\label{eq:antichain-gen}
    \cds =\; \uparrow\!\ac \coloneq \{\mi\in\mgs \ |\ \; \exists\,\mi'\in\ac,\, \mi\succeq\mi'\},
\end{equation}
and we say that $\bs{\X}$ is \textit{essential} if it contains an element of $\ac$. In the particular case where all elements of $\bs{\X}$ are in $\ac$, $\bs{\X}$ is said to be \textit{completely essential}. An example of down-set $\ds$ in the $\N=3$ MI-poset which is not a PMI, and its corresponding $\bset$ classification, are shown in \Cref{fig:D-not-P}.

\begin{figure}[tb]    
    \centering
        \begin{tikzpicture}

    \draw[draw=none,rounded corners,fill=PminusEcol] (1.3,2.6) rectangle (12.1,3.4);

    \draw[draw=none,rounded corners,fill=EminusPcol] (-0.55,0.6) rectangle (2.9,1.4);
    \draw[draw=none,rounded corners,fill=Qcol] (3.05,0.6) rectangle (6.5,1.4);
    \draw[draw=none,rounded corners,fill=Mcol] (6.65,0.6) rectangle (10.1,1.4);
    \draw[draw=none,rounded corners,fill=EminusMcol] (10.25,0.6) rectangle (13.7,1.4);

    \draw[draw=none,rounded corners,fill=Vcol] (1.2,-1.4) rectangle (2.3,-0.6);
    \draw[draw=none,rounded corners,fill=Vcol] (3.2,-1.4) rectangle (4.3,-0.6);
    \draw[draw=none,rounded corners,fill=Vcol] (5.2,-1.4) rectangle (6.3,-0.6);
    \draw[draw=none,rounded corners,fill=Vcol] (7.2,-1.4) rectangle (8.3,-0.6);
    \draw[draw=none,rounded corners,fill=Mcol] (9.2,-1.4) rectangle (10.3,-0.6);
    \draw[draw=none,rounded corners,fill=Vcol] (11.2,-1.4) rectangle (12.3,-0.6);

    \node () at (2.2,3) {{\scriptsize $\mi(0\!:\!123)_{_{\!+}}$}};
     \node () at (3.7,3) {{\scriptsize $\mi(1\!:\!023)_{_{\!+}}$}};
     \node () at (5.2,3) {{\scriptsize $\mi(2\!:\!013)_{_{\!+}}$}};
     \node () at (6.7,3) {{\scriptsize $\mi(3\!:\!012)_{_{\!+}}$}};
     \node () at (8.2,3) {{\scriptsize $\mi(01\!:\!23)_{_{\!+}}$}};
     \node () at (9.7,3) {{\scriptsize $\mi(02\!:\!13)_{_{\!+}}$}};
     \node () at (11.2,3) {{\scriptsize $\mi(03\!:\!12)_{_{\!+}}$}};
     
    \node () at (0.1,1) {{\scriptsize $\mi(01\!:\!2)_{_{\!+}}$}};
    \node () at (1.2,1) {{\scriptsize $\mi(02\!:\!1)_{_0}$}};
    \node () at (2.3,1) {{\scriptsize $\mi(12\!:\!0)_{_{\!+}}$}};
    \node () at (3.7,1) {{\scriptsize $\mi(01\!:\!3)_{_{\!+}}$}};
    \node () at (4.8,1) {{\scriptsize $\mi(03\!:\!1)_{_{\!+}}$}};
    \node () at (5.9,1) {{\scriptsize $\mi(13\!:\!0)_{_0}$}};
    \node () at (7.3,1) {{\scriptsize $\mi(02\!:\!3)_{_{\!+}}$}};
    \node () at (8.4,1) {{\scriptsize $\mi(03\!:\!2)_{_{\!+}}$}};
    \node () at (9.5,1) {{\scriptsize $\mi(23\!:\!0)_{_{\!+}}$}};
    \node () at (10.9,1) {{\scriptsize $\mi(12\!:\!3)_{_{\!+}}$}};
    \node () at (12,1) {{\scriptsize $\mi(13\!:\!2)_{_{\!+}}$}};
    \node () at (13.1,1) {{\scriptsize $\mi(23\!:\!1)_{_{\!+}}$}};

    \node () at (1.8,-1) {{\scriptsize $\mi(0\!:\!1)_{_0}$}};
     \node () at (3.8,-1) {{\scriptsize $\mi(0\!:\!2)_{_0}$}};
     \node () at (5.8,-1) {{\scriptsize $\mi(1\!:\!2)_{_0}$}};
     \node () at (7.8,-1) {{\scriptsize $\mi(0\!:\!3)_{_0}$}};
     \node () at (9.8,-1) {{\scriptsize $\mi(1\!:\!3)_{_{\!+}}$}};
     \node () at (11.8,-1) {{\scriptsize $\mi(2\!:\!3)_{_0}$}};

    \node (0123) at (6.7,3.7) {{\scriptsize $\bs{0123}$}};
    
    \node (012) at (1.2,1.7) {{\scriptsize $\bs{012}$}};
    \node (013) at (4.8,1.7) {{\scriptsize $\bs{013}$}};
    \node (023) at (8.4,1.7) {{\scriptsize $\bs{023}$}};
    \node (123) at (12,1.7) {{\scriptsize $\bs{123}$}};

    \node (01) at (1.8,-0.3) {{\scriptsize $\bs{01}$}};
    \node (02) at (3.8,-0.3) {{\scriptsize $\bs{02}$}};
    \node (12) at (5.8,-0.3) {{\scriptsize $\bs{12}$}};
    \node (03) at (7.8,-0.3) {{\scriptsize $\bs{03}$}};
    \node (13) at (9.8,-0.3) {{\scriptsize $\bs{13}$}};
    \node (23) at (11.8,-0.3) {{\scriptsize $\bs{23}$}};

    \draw[-,gray,very thin] (6.7,2.6) -- (012.north);
    \draw[-,gray,very thin] (6.7,2.6) -- (013.north);
    \draw[-,gray,very thin] (6.7,2.6) -- (023.north);
    \draw[-,gray,very thin] (6.7,2.6) -- (123.north);

    \draw[-,gray,very thin] (1.175,0.6) -- (01.north);
    \draw[-,gray,very thin] (1.175,0.6) -- (02.north);
    \draw[-,gray,very thin] (1.175,0.6) -- (12.north);
    \draw[-,gray,very thin] (4.775,0.6) -- (01.north);
    \draw[-,gray,very thin] (4.775,0.6) -- (03.north);
    \draw[-,gray,very thin] (4.775,0.6) -- (13.north);
    \draw[-,gray,very thin] (8.375,0.6) -- (02.north);
    \draw[-,gray,very thin] (8.375,0.6) -- (03.north);
    \draw[-,gray,very thin] (8.375,0.6) -- (23.north);
    \draw[-,gray,very thin] (11.975,0.6) -- (12.north);
    \draw[-,gray,very thin] (11.975,0.6) -- (13.north);
    \draw[-,gray,very thin] (11.975,0.6) -- (23.north);

    \draw[-,very thick] (-0.2,0.75) -- (0.3,0.75);
    \draw[-,very thick] (2,0.75) -- (2.5,0.75);
    \draw[-,very thick] (7,0.75) -- (7.5,0.75);
    \draw[-,very thick] (8.1,0.75) -- (8.6,0.75);
    \draw[-,very thick] (9.2,0.75) -- (9.7,0.75);
    \draw[-,very thick] (11.7,0.75) -- (12.2,0.75);

    \draw[-,very thick] (9.5,-1.25) -- (10,-1.25);
     
    \end{tikzpicture}
    \caption{An example (from \cite{Hubeny:2024fjn}) of a down-set $\ds$ in the $\N=3$ MI-poset which is \textit{not} a PMI, with the corresponding $\bset$ classification. $\ds$ is the set of MI instances indicated by the lower index $0$, while $+$ labels the elements of the complementary up-set $\cds$. The elements of the antichain $\ac$ that generates $\cds$ are underlined. 
    Each $\bset$ is represented by a colored box: positive but not essential (red), essential but not completely essential (orange), partial and not essential (teal), partial and essential (green), completely essential (yellow), and vanishing (blue). The gray lines indicate the order relations among the $\bsets$, which is simply given by inclusion. The key aspect of this example is that $\bs{012}$ is both essential and partial (green), since it contains elements of both $\ac$ and $\ds$, which as discussed in the main text, is not possible for KC-PMIs. To see that $\ds$ is indeed not a PMI, it suffices to notice that it is not consistent with the linear dependence among the MI instances. Imposing that all MI instances in $\ds$ vanish, it is immediate to verify that 
    $\mi(0:1)=0$, $\mi(0:2)=0$, and $\mi(02:1)=0$ imply $\mi(01:2)=0$, which should therefore be included in $\ds$, while it is not. Similar relations further imply that $\mi(12:0)=0$, and that in any KC-PMI $\pmi$ such that $\pmi\supseteq\ds$, the $\bset$ $\bs{012}$ is actually vanishing and no longer essential.
    }
    \label{fig:D-not-P}
\end{figure}
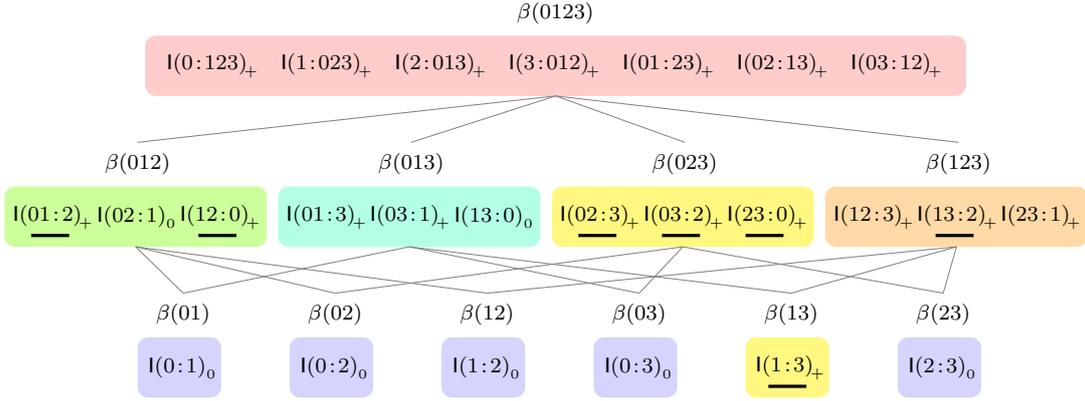

Suppose now that for a fixed down-set $\ds$ we are only given the set of its essential $\bsets$, rather than the specific elements of $\ac$. It should be intuitively clear that in general this information is not sufficient to unambiguously reconstruct $\ds$.\footnote{\,To see this, it suffices to find two down-sets with the same $\bset$ classification. We leave it as a simple exercise for the reader.} The key result of \cite{Hubeny:2024fjn}, however, is that unlike for arbitrary down-sets, for any KC-PMI \textit{all essential $\bsets$ are positive}, and this removes the ambiguity. While we refer the reader to \cite{Hubeny:2024fjn} for the proof of positivity of essential $\bsets$ for KC-PMIs, it is easy to see why this makes the reconstruction possible. 

Let $\pmi$ be a KC-PMI, suppose that we are only given the set of its essential $\bsets$, and let $\bs{\X}$ be one such $\bset$. Similarly to the case of arbitrary down-sets, it appears that we still do not know which elements of $\bs{\X}$ are in $\ac$, and which ones are not (although we will see below how to also efficiently extract this information). Nevertheless, if all elements of $\bs{\X}$ are positive, i.e., they are in $\cpmi$, then for any $\mi\in\bs{\X}$ such that $\mi\notin\ac$ there exists some $\mi'\preceq\mi$ such that $\mi'\in\ac$ (since $\ac$ generates $\cpmi$, \eqref{eq:antichain-gen}). Since $\bsets$ form a partition of the MI-poset, $\mi'\in\bs{\X'}$ for some $\X'$, and by definition, $\bs{\X'}$ is essential. (If it is not completely essential, then it contains some other $\mi''\notin\ac$, but we can repeat the same argument.) Therefore, if we take the up-set of each MI instance in each essential $\bset$ of $\pmi$, and then the union of all these up-sets, we obtain $\cpmi$. In summary, for any KC-PMI $\pmi$, the reconstruction of its complement $\cpmi$ from the essential $\bsets$ is given by
\begin{equation}
\label{eq:reconstruction}
    \cpmi = \!\!\!\bigcup_{\substack{\mi\in\bs{\X} \\ \bs{\X}\, \text{essential}}} \!\!\!\!\!\uparrow \!\mi \, = \, \{\mi'\in\mgs_\N \ | \ \, 
    \mi' \succeq \mi, 
    \, \mi\in\bs{\X},\, \bs{\X}\; \text{is essential}\}.
\end{equation}
Furthermore, from the argument we just presented, it is a simple exercise to verify that exactly the same formula also works if we replace the essential $\bsets$ with the positive ones (as in fact was done in \cite{Hubeny:2024fjn}). Notice that if we try to apply \eqref{eq:reconstruction} to the example in \Cref{fig:D-not-P}, $\cds$ would also include $\mi(02:1)$ (since it belongs to an essential $\bset$), and if we use the alternative version based on positive $\bsets$, $\cds$ does not include $\mi(01:2)$ and $\mi(12:0)$ (since they belong to a partial $\bset$). In either case, the result of the formula is not $\cds$. 
On the other hand, it is easy to see that the formula gives the correct result for a KC-PMI (see the example discussed below and illustrated in \Cref{fig:PMI-down-set-bsets}).

Since the knowledge of positive or essential $\bsets$ is sufficient to characterize any KC-PMI, we can represent any KC-PMI using this data, and accordingly, introduce the following definition.
\begin{defi}[Correlation hypergraph of a KC-PMI]
    For an arbitrary number of parties $\N$, and an arbitrary \emph{KC-PMI} $\pmi$, the correlation hypergraph of $\pmi$ is the  hypergraph $\hp$ with $\N+1$ vertices $v_0,v_1,\ldots,v_\N$ corresponding to the parties (including the purifier), and a hyperedge $h_\X=\{v_\ell|\,\ell\in\X\}$ for each subsystem $\X$ whose $\bset$ $\bs{\X}$ is positive.
\end{defi}
Notice that, given any hypergraph $\gf{H}=(V,E)$ with vertex set $V$ and set of hyperedges $E$, one can assign a party to each vertex, interpret each hyperedge as a $\bset$, and apply \eqref{eq:reconstruction} to obtain the up-set $\cds$ and its complementary down-set $\ds$ in the MI-poset. Since not all down-sets can be represented in this way in terms of a hypergraph, one may then wonder whether the restriction automatically implemented by this hypergraph representation characterizes the subsets of down-sets of the MI-poset that are indeed KC-PMIs. In other words, one may ask whether any hypergraph is the correlation hypergraph of a KC-PMI. Unfortunately, this is not the case, as can be readily verified by considering, for any $|V|\geq 3$, the hypergraph with the single hyperedge $h_{\nsp}$. We leave the details as a simple exercise for the reader. 

The hyperedges of the correlation hypergraph—and, in particular, the connectivity of certain subhypergraphs—can be intuitively interpreted as indicating the presence or absence of correlations among subsystems. To make this connection precise, let us first recall a few additional properties of $\bsets$ for KC-PMIs from \cite{Hubeny:2024fjn}. Suppose we are given a KC-PMI $\pmi$ described solely by its correlation hypergraph, rather than by its set of vanishing MI instances. Observe that for a given KC-PMI $\pmi$, if the $\bset$ $\bs{\X}$ corresponding to some subsystem $\X$ is either vanishing or positive, the sign of all MI instances in $\bs{\X}$ is immediately determined: in the former case, they are all vanishing, and in the latter, all positive.
In contrast, if $\bs{\X}$ is partial, it is not immediately clear which MI instances vanish and which are positive. Likewise, if $\bs{\X}$ is essential but not completely essential, we cannot directly tell which MI instances belong to $\ac$ and which do not. Of course, by using \eqref{eq:reconstruction}, we can always reconstruct $\pmi$ from the positive $\bsets$ and thereby determine this information. However, a key advantage of the correlation hypergraph representation is that it allows us to make these determinations directly, without the need to perform this reconstruction step.

As shown in \cite{Hubeny:2024fjn}, for a given KC-PMI, if a $\bset$ $\bs{\X}$ is partial, or essential but not completely essential, the MI instances that are (respectively) vanishing, or belong to $\ac$, can be conveniently characterized by means of a specific partition $\Gamma(\X)$ of $\X$, constructed as follows. We first consider all subsystems $\Y\subset\X$ such that $\bs{\Y}$ is positive, select the maximal ones with respect to inclusion, and denote them by $\X_i$, where $i\in[n]$ and $n$ is the number of subsystems with this property. It turns out that when $\bs{\X}$ is partial or essential but not completely essential, these subsystems are pairwise disjoint and form a partition of a subset of $\X$; in general, they do not cover all of $\X$. The partition $\Gamma(\X)$ can therefore be written as
\begin{equation}
\label{eq:gamma-partition}
\Gamma(\X) = \{\X_1,\ldots,\X_n,\ell_{n+1},\ldots,\ell_{n+\tilde{n}}\},
\end{equation}
where each $\bs{\X_i}$ is positive and maximal for $i\in[n]$, and the singletons $\ell_{n+1},\ldots,\ell_{n+\tilde{n}}$ correspond to the parties not contained in any $\X_i$, completing the partition. This partition will be called the $\Gamma$-\textit{partition} of $\X$ in what follows. It was then shown in \cite{Hubeny:2024fjn} that, if $\bs{\X}$ is partial, the vanishing MI instances in $\bs{\X}$ are precisely those of the form $\mi(\Y:\Z)$, where both $\Y$ and $\Z$ are unions of collections of elements of $\Gamma(\X)$. Furthermore, in the case where $\bs{\X}$ is essential, the same prescription identifies the MI instances in $\ac$ (recall that there are no partial and simultaneously essential $\bsets$, since for KC-PMIs all essential $\bsets$ are positive). 

We can also reinterpret the remaining classes of $\bsets$ in terms of $\Gamma(\X)$. If a $\bset$ $\bs{\X}$ is positive but not essential, the collection $\{\X_i\}_{i\in[n]}$ of subsystems corresponding to maximal positive $\bsets$ with $\X_i\subset\X$ does not form a partition of $\X$. In this case, we may simply define $\Gamma(\X)$ to be the trivial partition, namely $\{\X\}$. To see this, note that one can think of $\Gamma(\X)$ as the finest partition such that each $\X_i$ is contained within one of its elements. However, as shown in \cite{Hubeny:2024fjn}, for positive $\bsets$,
all $\X_i$ intersect pairwise, and the union of any two equals $\X$. If, instead, there are no $\X_i\subset\X$ such that $\bs{\X_i}$ is positive, then $\Gamma(\X)$ is the singleton partition (the finest possible partition of $\X$), and $\bs{\X}$ is completely essential. Similarly, if $\bs{\X}$ is not positive, it must be either partial or vanishing, and it is vanishing if and only if $\Gamma(\X)$ is again the singleton partition. 

\begin{figure}[tb]    
    \centering
        \begin{tikzpicture}

    \draw[draw=none,rounded corners,fill=PminusEcol] (1.3,2.6) rectangle (12.1,3.4);

    \draw[draw=none,rounded corners,fill=PminusEcol] (-0.55,0.6) rectangle (2.9,1.4);
    \draw[draw=none,rounded corners,fill=PminusEcol] (3.05,0.6) rectangle (6.5,1.4);
    \draw[draw=none,rounded corners,fill=PminusEcol] (6.65,0.6) rectangle (10.1,1.4);
    \draw[draw=none,rounded corners,fill=EminusMcol] (10.25,0.6) rectangle (13.7,1.4);

    \draw[draw=none,rounded corners,fill=Mcol] (1.2,-1.4) rectangle (2.3,-0.6);
    \draw[draw=none,rounded corners,fill=Mcol] (3.2,-1.4) rectangle (4.3,-0.6);
    \draw[draw=none,rounded corners,fill=Mcol] (5.2,-1.4) rectangle (6.3,-0.6);
    \draw[draw=none,rounded corners,fill=Mcol] (7.2,-1.4) rectangle (8.3,-0.6);
    \draw[draw=none,rounded corners,fill=Vcol] (9.2,-1.4) rectangle (10.3,-0.6);
    \draw[draw=none,rounded corners,fill=Vcol] (11.2,-1.4) rectangle (12.3,-0.6);

    \node () at (2.2,3) {{\scriptsize $\mi(0\!:\!123)_{_{\!+}}$}};
     \node () at (3.7,3) {{\scriptsize $\mi(1\!:\!023)_{_{\!+}}$}};
     \node () at (5.2,3) {{\scriptsize $\mi(2\!:\!013)_{_{\!+}}$}};
     \node () at (6.7,3) {{\scriptsize $\mi(3\!:\!012)_{_{\!+}}$}};
     \node () at (8.2,3) {{\scriptsize $\mi(01\!:\!23)_{_{\!+}}$}};
     \node () at (9.7,3) {{\scriptsize $\mi(02\!:\!13)_{_{\!+}}$}};
     \node () at (11.2,3) {{\scriptsize $\mi(03\!:\!12)_{_{\!+}}$}};
     
    \node () at (0.1,1) {{\scriptsize $\mi(01\!:\!2)_{_{\!+}}$}};
    \node () at (1.2,1) {{\scriptsize $\mi(02\!:\!1)_{_{\!+}}$}};
    \node () at (2.3,1) {{\scriptsize $\mi(12\!:\!0)_{_{\!+}}$}};
    \node () at (3.7,1) {{\scriptsize $\mi(01\!:\!3)_{_{\!+}}$}};
    \node () at (4.8,1) {{\scriptsize $\mi(03\!:\!1)_{_{\!+}}$}};
    \node () at (5.9,1) {{\scriptsize $\mi(13\!:\!0)_{_+}$}};
    \node () at (7.3,1) {{\scriptsize $\mi(02\!:\!3)_{_{\!+}}$}};
    \node () at (8.4,1) {{\scriptsize $\mi(03\!:\!2)_{_{\!+}}$}};
    \node () at (9.5,1) {{\scriptsize $\mi(23\!:\!0)_{_{\!+}}$}};
    \node () at (10.9,1) {{\scriptsize $\mi(12\!:\!3)_{_{\!+}}$}};
    \node () at (12,1) {{\scriptsize $\mi(13\!:\!2)_{_{\!+}}$}};
    \node () at (13.1,1) {{\scriptsize $\mi(23\!:\!1)_{_{\!+}}$}};

    \node () at (1.8,-1) {{\scriptsize $\mi(0\!:\!1)_{_{\!+}}$}};
     \node () at (3.8,-1) {{\scriptsize $\mi(0\!:\!2)_{_{\!+}}$}};
     \node () at (5.8,-1) {{\scriptsize $\mi(1\!:\!2)_{_{\!+}}$}};
     \node () at (7.8,-1) {{\scriptsize $\mi(0\!:\!3)_{_{\!+}}$}};
     \node () at (9.8,-1) {{\scriptsize $\mi(1\!:\!3)_{_{0}}$}};
     \node () at (11.8,-1) {{\scriptsize $\mi(2\!:\!3)_{_0}$}};

    \node (0123) at (6.7,3.7) {{\scriptsize $\bs{0123}$}};
    
    \node (012) at (1.2,1.7) {{\scriptsize $\bs{012}$}};
    \node (013) at (4.8,1.7) {{\scriptsize $\bs{013}$}};
    \node (023) at (8.4,1.7) {{\scriptsize $\bs{023}$}};
    \node (123) at (12,1.7) {{\scriptsize $\bs{123}$}};

    \node (01) at (1.8,-0.3) {{\scriptsize $\bs{01}$}};
    \node (02) at (3.8,-0.3) {{\scriptsize $\bs{02}$}};
    \node (12) at (5.8,-0.3) {{\scriptsize $\bs{12}$}};
    \node (03) at (7.8,-0.3) {{\scriptsize $\bs{03}$}};
    \node (13) at (9.8,-0.3) {{\scriptsize $\bs{13}$}};
    \node (23) at (11.8,-0.3) {{\scriptsize $\bs{23}$}};

    \draw[-,gray,very thin] (6.7,2.6) -- (012.north);
    \draw[-,gray,very thin] (6.7,2.6) -- (013.north);
    \draw[-,gray,very thin] (6.7,2.6) -- (023.north);
    \draw[-,gray,very thin] (6.7,2.6) -- (123.north);

    \draw[-,gray,very thin] (1.175,0.6) -- (01.north);
    \draw[-,gray,very thin] (1.175,0.6) -- (02.north);
    \draw[-,gray,very thin] (1.175,0.6) -- (12.north);
    \draw[-,gray,very thin] (4.775,0.6) -- (01.north);
    \draw[-,gray,very thin] (4.775,0.6) -- (03.north);
    \draw[-,gray,very thin] (4.775,0.6) -- (13.north);
    \draw[-,gray,very thin] (8.375,0.6) -- (02.north);
    \draw[-,gray,very thin] (8.375,0.6) -- (03.north);
    \draw[-,gray,very thin] (8.375,0.6) -- (23.north);
    \draw[-,gray,very thin] (11.975,0.6) -- (12.north);
    \draw[-,gray,very thin] (11.975,0.6) -- (13.north);
    \draw[-,gray,very thin] (11.975,0.6) -- (23.north);

    \draw[-,very thick] (10.6,0.75) -- (11.1,0.75);
    
    \draw[-,very thick] (1.5,-1.25) -- (2,-1.25);
    \draw[-,very thick] (3.5,-1.25) -- (4,-1.25);
    \draw[-,very thick] (5.5,-1.25) -- (6,-1.25);
    \draw[-,very thick] (7.5,-1.25) -- (8,-1.25);
     
    \end{tikzpicture}
    \caption{An example of an $\N=3$ KC-PMI with the corresponding classification of $\bsets$. The notation and color coding are the same as in \Cref{fig:D-not-P}. Notice that $\bs{123}$ is essential, and that $\Gamma(123)=\{12,3\}$, since the only $\X_i\subset\X$ with positive $\bs{\X_i}$ is $\X_i=12$. Accordingly, the MI instance in $\bs{123}$ which belongs to $\ac$ is $\mi(12:3)$. For any other positive $\bs{\X}$ instead, $\Gamma(\X)$ is the trivial partition since the maximal $\X_i\subset\X$ with positive $\bs{\X_i}$ do not form a partition of $\X$. For example, for $\X=012$ we have $i=3$ and $\X_1=01$, $\X_2=02$ and $\X_3=12$. Furthermore, notice that as discussed in the main text, these subsystems intersect pairwise, and the union of any pair is $\X=012$. 
    }
    \label{fig:PMI-down-set-bsets}
\end{figure}
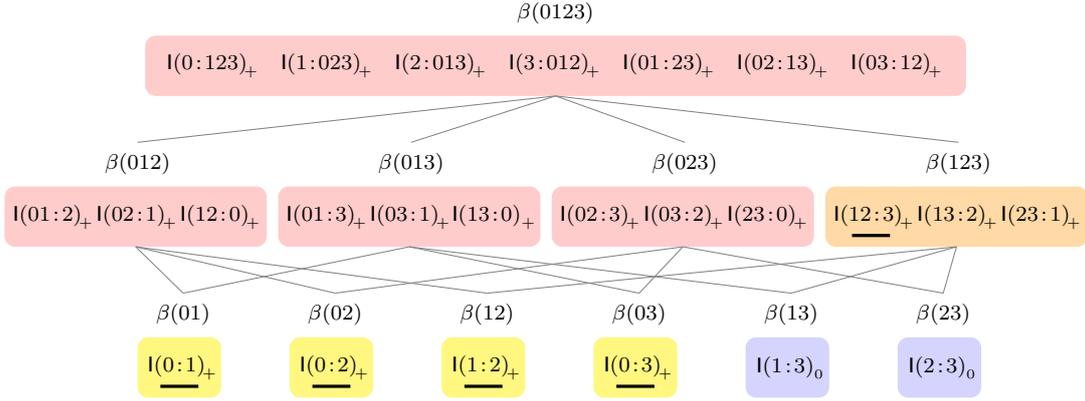

For convenience, we collect these properties of $\Gamma$-partitions into a theorem, where we say that an MI instance $\mi(\Y:Z)$ \textit{splits} $\X$ if $\X\cap\Y\neq\varnothing$ and $\X\cap\Z\neq\varnothing$. For an example about how to derive and interpret these partitions, see \Cref{fig:PMI-down-set-bsets}.

\begin{thm}
\label{thm:gamma-partitions}
    For any \emph{KC-PMI} $\pmi$, and any subsystem $\X$:
    \begin{enumerate}[label={\footnotesize \emph{(\roman*)}}]
    \item $\Gamma(\X)$ is trivial if and only if $\bs{\X}$ is positive and not essential,
    \item $\Gamma(\X)$ is the singleton partition if and only if $\bs{\X}$ is either vanishing, or completely essential (and therefore positive), 
    \item if $\bs{\X}$ is positive (respectively, 
    partial), the \emph{MI} instances in $\bs{\X}$ which belong to $\ac$ (respectively, $\pmi$) are precisely the ones which do not split any $X_i$ in the  $\Gamma$-partition \eqref{eq:gamma-partition}. In case \emph{(i)} above there are no such instances, while in case \emph{(ii)} all instances have this property.
    \end{enumerate}
\end{thm}
\begin{proof}
    See \cite{Hubeny:2024fjn} (in particular Theorem 5 and 6), and the discussion above.
\end{proof}

We can now reinterpret these results more intuitively in the language of the correlation hypergraph. Consider an arbitrary subsystem $\X$. If there is no hyperedge corresponding to $\X$, then, by definition, $\bs{\X}$ is not positive and is therefore either partial or vanishing. In this case, we can ``zoom in'' on $\X$ by deleting all vertices associated with parties not in $\X$, thereby obtaining the corresponding induced subhypergraph, which is necessarily disconnected. Its connected components correspond to the maximal positive $\bsets$, while the vanishing MI instances are precisely those whose arguments do not ``cut'' any of these hyperedges. In the particular case where $\bs{\X}$ is vanishing, the induced subhypergraph is trivial.

Now suppose instead that $\bs{\X}$ is positive and that there exists a corresponding hyperedge. We can again zoom in on $\X$, obtaining a subhypergraph that is trivially connected due to this hyperedge. We may then ``look inside'' $\X$ by ignoring this hyperedge and examining the internal connectivity of the remaining structure. If the resulting subhypergraph remains connected, $\bs{\X}$ is not essential: the parties within $\X$ remain correlated and continue to ``hold together'' even in the absence of that hyperedge. Conversely, if the subhypergraph becomes disconnected, $\bs{\X}$ is essential, and—as in the partial case—the elements of $\ac$ are precisely those MI instances whose arguments do not cut any hyperedge. Finally, if the subhypergraph is trivial, this corresponds to $\bs{\X}$ being completely essential.

\subsection{A new look at the KC-lattice}
\label{subsec:kc-lattice}

As first shown in \cite{He:2022bmi}, the set of all KC-PMIs at any given $\N$ is a \textit{complete lattice},\footnote{\,A poset is a lattice if any pair of elements has a least upper bound, called the \textit{join}, and a greatest lower bound, called the \textit{meet}. A lattice is said to be \textit{complete} if meet and join exist for any collection of elements. Any finite lattice is a complete lattice.} denoted by $\lat{KC}^\N$. The partial order corresponds to inclusion
\begin{equation}
\label{eq:kc-order}
    \pmi_1 \preceq \pmi_2 \iff \pmi_1\subseteq\pmi_2,
\end{equation}
and the \textit{meet}\footnote{\,For completeness we remark that the join can be described geometrically as the highest dimensional face of the SAC$_\N$ whose PMI is a KC-PMI and such that it is contained in the boundary of both $\face_1$ and $\face_2$ (the faces corresponding to $\pmi_1$ and $\pmi_2$). A combinatorial characterization however is currently not known, and we will not need to invoke it in what follows.} to intersection
\begin{equation}
\label{eq:kc-meet}
    \pmi_1 \wedge \pmi_2 = \pmi_1 \cap \pmi_2.
\end{equation}
Geometrically, ``smaller'' PMIs contain fewer MI instances and therefore correspond to higher-dimensional faces of the SAC$_\N$.
Given two entropy vectors $\Svec_1\in\text{int}(\face_1)$ and $\Svec_2\in\text{int}(\face_2)$, where $\face_1$ and $\face_2$ are the faces of SAC$_\N$ corresponding to $\pmi_1$ and $\pmi_2$, any linear combination of $\Svec_1$ and $\Svec_2$ with strictly positive coefficients gives an entropy vector in the interior of the face corresponding to $\pmi_1 \wedge \pmi_2$ \cite{He:2022bmi}, and in particular
\begin{equation}
\label{eq:pmi-meet-svec-sum}
    \pmimap(\Svec_1+\Svec_2) = \pmimap(\Svec_1) \wedge \pmimap(\Svec_2),
\end{equation}
which will be convenient below.

The bottom element of $\lat{KC}^\N$ is $\pmi_{\perp}=\varnothing$, which corresponds to the interior of the SAC$_\N$ (its $\D$-dimensional face), and the top element is $\pmi_{\top}=\mgs_\N$, which corresponds to the tip of the SAC$_\N$ (its 0-dimensional face). In the language of the correlation hypergraph, these correspond (respectively) to the complete hypergraph (with all possible hyperedges\footnote{\,For all hypergraphs considered in this work, we always assume that they have no ``loops'' (single vertex hyperedges).}) and the trivial hypergraph (with no hyperedges).

The goal of this subsection is to rephrase \eqref{eq:kc-order} and \eqref{eq:kc-meet} in the language of $\bsets$ and the correlation hypergraph. For \eqref{eq:kc-order} this is quite straightforward: It is reversed inclusion of the set of hyperedges, as shown by the following result, where we denote by $\pos(\pmi)$ the set of positive $\bsets$ of a KC-PMI $\pmi$.

\begin{thm}
\label{thm:partial-order-h}
    For any $\N$, and any two \emph{KC-PMI}s $\pmi_1,\pmi_2$
    \begin{equation}
        \pmi_1\preceq\pmi_2 \iff \pos(\pmi_1)\supseteq\pos(\pmi_2).
    \end{equation}
\end{thm}
\begin{proof}
    In the forward direction, if $\pmi_1\preceq\pmi_2$, then $\pmi_1\subseteq\pmi_2$ (cf., \eqref{eq:kc-order}). Taking the complements we obtain $\cpmi_1\supseteq\cpmi_2$, implying that for every subsystem $\X$, if $\bs{\X}$ is positive for $\pmi_2$ it is also positive for $\pmi_1$.

    Conversely, given any two KC-PMIs $\pmi_1$ and $\pmi_2$ such that $\pos(\pmi_1)\supseteq\pos(\pmi_2)$, the reconstruction formula \eqref{eq:reconstruction} implies that for any MI instance $\mi$ in $\cpmi_2$, $\mi$ is also in $\cpmi_1$, which means that $\cpmi_1\supseteq\cpmi_2$ and taking the complement we conclude the proof.
\end{proof}

It is important to notice that \Cref{thm:partial-order-h} does not  provide direct access to the cover relations in $\lat{KC}^\N$.\footnote{\,Given a poset $(\A,\preceq)$, and two elements $a,b\in\A$, $a$ is said to \textit{cover} $b$ if $b\prec a$ and there is no element $c\in\A$, such that $b\prec c\prec a$.} In particular, given the correlation hypergraph $\hp$ for some KC-PMI $\pmi$, it is not in general the case that by deleting a hyperedge from $\hp$ one obtains a hypergraph which is the correlation hypergraph for some other KC-PMI $\pmi'\succeq\pmi$. The main reason for this is that, as mentioned above, not all hypergraphs are correlation hypergraphs for some KC-PMI. 

Having described the partial order relation, we now proceed to describe the meet of the KC-lattice. To do so, it will be convenient to first briefly review the definition of meet and join in the \textit{lattice of partitions} of a finite set, and introduce a certain \textit{closure operator}\footnote{\,For any set $\A$, a closure operator \textit{on} $\A$ is a map $\cl:2^{\A}\rightarrow 2^{\A}$, where $2^{\A}$ is the power set of $\A$, such that for all subsets $\B,\C$: $\cl(\B)\supseteq \B$, $\cl (\cl(\B))=\cl(\B)$, and $\B\subseteq\C \implies \cl(\B)\subseteq\cl(\C)$.\label{ft:closure}} on the hyperedges of a hypergraph. 

Given a set $\A$, a partition $\Upsilon$ of $\A$, and an element $a\in \A$, we denote by $\eqcl{a}$ the \textit{block} (i.e., element of the partition) that contains $a$. Since any partition $\Upsilon$ of $\A$ can be interpreted as an equivalence relation $\sim$ on $\A$, and vice versa any equivalence relation naturally induces a partition $\A/\sim$, the block $\eqcl{a}$ can also be interpreted as the \textit{equivalence class} under $\sim$ represented by the element $a$. A partition $\Upsilon'$ is said to be a \textit{refinement} of another partition $\Upsilon$ if every element of $\Upsilon'$ is a subset of some element of $\Upsilon$. We will say that $\Upsilon'$ is \textit{finer} than $\Upsilon$ and $\Upsilon$ is \textit{coarser} than $\Upsilon'$. This relation among partitions is a partial order on the set of all partitions of $\A$, and we will write $\Upsilon'\preceq\Upsilon$. The resulting poset is also a lattice, with meet and join corresponding to the following operations. Given two partitions $\Upsilon_1$ and $\Upsilon_2$ of $\A$, the meet $\Upsilon_1\wedge\Upsilon_2$ is the partition whose blocks are the intersections of all pairs of blocks from $\Upsilon_1$ and $\Upsilon_2$, excluding the empty ones. To obtain the join, we can consider an equivalence relation between the blocks $\eqcl{a}\in\Upsilon_1$ and $\eqcl{b}\in\Upsilon_2$ where $\eqcl{a}\sim\eqcl{b}$ if $\eqcl{a}\cap \eqcl{b}\neq\varnothing$. The join $\Upsilon_1\vee\Upsilon_2$ is then the partition where each block is the union of a family of blocks of $\Upsilon_1$ and $\Upsilon_2$ connected by this relation. Equivalently, since the join is the least upper bound of $\Upsilon_1$ and $\Upsilon_2$ in the poset, we can think about the join as the finest partition which is greater than both $\Upsilon_1$ and $\Upsilon_2$.

To describe the meet in the KC-lattice in the language of the correlation hypergraph, we also need to introduce an operation on the set of hyperedges 
defined for any hypergraph, which we will call the \textit{weak union closure}.   
This operation will have applications beyond what we discuss in this subsection since, as we will see momentarily, any hypergraph that is the correlation hypergraph of a KC-PMI is necessarily invariant under this transformation.

\begin{defi}[Weak union closure]
    Given a hypergraph $\gf{H}=(V,E)$, we define its weak union closure $\text{\emph{cl}}_{\cup}(\gf{H})$ as the hypergraph $\gf{H}'=(V',E')$ with $V'=V$ and $E'$ corresponding to the end point $E_n$ of the 
    sequence $E_0\subset E_1\subset\ldots\subset E_n$ where $E_0=E$ and 
    \begin{equation}
        E_{i+1}=E_i\cup\{h\cup h'\ |\ \; h,h'\in E_i\;,\; h\cap h'\neq\varnothing\}.
    \end{equation}
\end{defi}
In other words, we include all hyperedges corresponding to unions of overlapping hyperedges.
In the above definition notice that the sequence is guaranteed to be finite by the requirement that the inclusion $E_i\subset E_{i+1}$ is strict, and the fact that $V$ is a finite set. 
We leave it as an exercise for the reader to verify that the map from $E_0$ to $E_n$ is indeed a closure operator (cf., \Cref{ft:closure}) on the set $2^V$ of possible hyperedges. 

For any $\N$ and any KC-PMI $\pmi$, the correlation hypergraph of $\pmi$ is invariant under this operation, since the set of hyperedges is already closed,\footnote{\,Since by construction there is a bijection between the positive $\bsets$ and the set of hyperedges, to simplify the notation we use the same symbol for the corresponding weak union closure operators acting on these sets.}
\begin{equation}
\label{eq:corr-hyp-inv}
    \clw(\pos(\pmi)) = \pos(\pmi).
\end{equation}
This property follows from the fact \cite{Hubeny:2024fjn} that, for any KC-PMI, if for two overlapping subsystems $\X$ and $\Y$ both $\bs{\X}$ and $\bs{\Y}$ are positive, then $\bs{\X\cup\Y}$ is necessarily also positive. This behavior is, in fact, closely related to the properties of the $\Gamma$-partitions \eqref{eq:gamma-partition} reviewed in the previous subsection (see \cite{Hubeny:2024fjn} for further details). 
The same analysis further demonstrates that
\begin{equation}
\label{eq:pos-from-ess}
    \pos (\pmi) = \clw (\ess (\pmi)),
\end{equation}
where $\ess(\pmi)$ denotes the set of essential $\bsets$ of $\pmi$. Lastly, notice that \eqref{eq:corr-hyp-inv} imposes a restriction on the set of hypergraphs that could potentially be correlation hypergraphs of KC-PMIs, and one might wonder whether this restriction is sufficient to characterize precisely which hypergraphs correspond to correlation hypergraphs of KC-PMIs. Unfortunately, this is again not the case, as demonstrated already by the previous example, of a hypergraph with a single hyperedge containing all vertices, which is clearly invariant under $\clw$, even if it is not a correlation hypergraph.

\begin{figure}[tbp]
    \centering
    \begin{subfigure}{0.3\textwidth}
    \centering
    \begin{tikzpicture}
    
    \draw[Mcol!80!black, very thick] (-1,0) -- (0,1.5);
    
    \filldraw (-1,0) circle (2pt);
    \filldraw (1,0) circle (2pt);
    \filldraw (0,1.5) circle (2pt);
    
    \node[] () at (-1.2,-0.2) {{\scriptsize $v_1$}};
    \node[] () at (1.2,-0.2) {{\scriptsize $v_2$}};
    \node[] () at (0,1.8) {{\scriptsize $v_0$}};

    \node[] () at (-0.9,0.75) {{\scriptsize $h_{01}$}};
    \end{tikzpicture}
    \subcaption[]{}
    \end{subfigure}
    \begin{subfigure}{0.3\textwidth}
    \centering
    \begin{tikzpicture}
    
    \draw[Mcol!80!black, very thick] (1,0) -- (0,1.5);
    
    \filldraw (-1,0) circle (2pt);
    \filldraw (1,0) circle (2pt);
    \filldraw (0,1.5) circle (2pt);
    
    \node[] () at (-1.2,-0.2) {{\scriptsize $v_1$}};
    \node[] () at (1.2,-0.2) {{\scriptsize $v_2$}};
    \node[] () at (0,1.8) {{\scriptsize $v_0$}};

    \node[] () at (0.9,0.75) {{\scriptsize $h_{02}$}};
    \end{tikzpicture}
    \subcaption[]{}
    \end{subfigure}
    \begin{subfigure}{0.3\textwidth}
    \centering
    \begin{tikzpicture}
    \draw[PminusEcol, rounded corners, very thick] (-1.6,-0.6) rectangle (1.6,2.1);
    
    \draw[Mcol!80!black, very thick] (-1,0) -- (0,1.5);
    \draw[Mcol!80!black, very thick] (1,0) -- (0,1.5);
    
    \filldraw (-1,0) circle (2pt);
    \filldraw (1,0) circle (2pt);
    \filldraw (0,1.5) circle (2pt);
    
    \node[] () at (-1.2,-0.2) {{\scriptsize $v_1$}};
    \node[] () at (1.2,-0.2) {{\scriptsize $v_2$}};
    \node[] () at (0,1.8) {{\scriptsize $v_0$}};

    \node[] () at (-0.9,0.75) {{\scriptsize $h_{01}$}};
    \node[] () at (0.9,0.75) {{\scriptsize $h_{02}$}};
    \node[] () at (0,-0.4) {{\scriptsize $h_{012}$}};
    \end{tikzpicture}
    \subcaption[]{}
    \end{subfigure}
    \caption{A simple example of two KC-PMIs $\pmi_1$ and $\pmi_2$ (whose correlation hypergraphs are indicated respectively in (a) and (b)) such that for their meet (c), $\pos(\pmi_1\wedge\pmi_2)$ strictly contains the union of $\pos(\pmi_1)$ and $\pos(\pmi_2)$. Notice that in (c) the hyperedge $h_{012}$ is added to the set of hyperedges of the correlation hypergraph of the meet 
    by the week union closure of the union of $\pos(\pmi_1)$ and $\pos(\pmi_2)$, since $h_{01}\cap h_{02}=\{v_0\}\neq\varnothing$. The color of the hyperedges corresponds to the classification of the associated $\bsets$ as in \Cref{fig:D-not-P}.
    Following the presentation in \cite{Hubeny:2024fjn}, we are representing a 2-edge by a line segment and higher $k$-edge by a rounded polygon around the corresponding vertices.
    }
    \label{fig:kc-meet-bs-inclusion}
\end{figure}
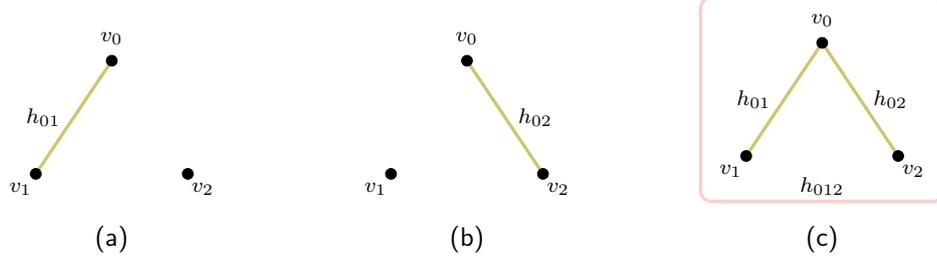

With these definitions we can finally describe the meet in the KC-lattice from the point of view of $\bsets$ and the correlation hypergraph. Because of \eqref{eq:kc-meet}, given the correlation hypergraphs of $\pmi_1$ and $\pmi_2$, it is clear that to find the correlation hypergraph of their meet we could simply reconstruct $\pmi_1$ and $\pmi_2$ explicitly, take their intersection, and then recover the desired hypergraph. The key point of the following result is that it tells us how to obtain the correlation hypergraph of $\pmi_1\wedge\pmi_2$ directly, without going through these steps. Similarly, if we are interested in the transformation of the $\Gamma$-partition for a specific subsystem $\X$, we could read off this data from the correlation hypergraph of the meet. The theorem however also provides a formula for a direct computation from the same data in $\pmi_1$ and $\pmi_2$. For a simple example of the first relation  \eqref{eq:meet-pos-bsets} below, see \Cref{fig:kc-meet-bs-inclusion}.

\begin{thm}
\label{thm:meet}
    For any given $\N$, and any two $\N$-party \emph{KC-PMI}s $\pmi_{1}$ and $\pmi_{2}$, the set of positive $\bsets$ of $\pmi_1\wedge\pmi_2$ is given by
    \begin{equation}
    \label{eq:meet-pos-bsets}
         \pos(\pmi_1\wedge\pmi_2) = \text{\emph{cl}}_{\cup} \left(\pos(\pmi_1) \cup \pos(\pmi_2)\right).
    \end{equation}
    Furthermore, for each subsystem $\X$, 
    \begin{equation}
    \label{eq:meet-partition}
        \Gamma_{\pmi_1\wedge\pmi_2}(\X) = \Gamma_{\pmi_1}(\X) \vee \Gamma_{\pmi_2}(\X),
    \end{equation}
    where $\Gamma_{\pmi}(\X)$ is the $\Gamma$-partition of $\X$ in the \emph{PMI} $\pmi$, and $\vee$ is the join in the lattice of partitions of $\X$.
\end{thm}
\begin{proof}
For given $\N$, consider two $\N$-party KC-PMIs $\pmi_{1}$ and $\pmi_{2}$. For an arbitrary subsystem $\X$, we denote by $\bsi{\X}{1}$, $\bsi{\X}{2}$ and $\bsi{\X}{12}$ the $\bset$ of $\X$ respectively in $\pmi_1$, $\pmi_2$ and their meet $\pmi_1\wedge\pmi_2$ (and similarly for the antichains $\ac$ and $\Gamma$-partitions below). It follows immediately from \eqref{eq:kc-meet} that for any subsystem $\X$, if either $\bsi{\X}{1}$ or $\bsi{\X}{2}$ is positive, $\bsi{\X}{12}$ is necessarily also positive, implying
\begin{equation}
\label{eq:pos-inclusion}
    \pos(\pmi_1\wedge\pmi_2) \supseteq \pos(\pmi_1) \cup \pos(\pmi_2).
\end{equation}
On the other hand, if $\bsi{\X}{12}$ is essential, then there must be an MI instance $\mi\in\bsi{\X}{12}$ such that $\mi$ is in the antichain $\ac_{12}$ for $\pmi_1\wedge\pmi_2$,and therefore $\mi\notin \pmi_1\wedge\pmi_2$ but for all $\mi' \prec \mi$, $\mi'\in \pmi_1\wedge\pmi_2$.  Hence $\mi'$ must vanish in both $\pmi_1$ and $\pmi_2$, but $\mi$ must be positive in at least one of $\pmi_1$ or $\pmi_2$, so at least one of $\bsi{\X}{1}$ or $\bsi{\X}{2}$ is essential, implying

\begin{equation}
\label{eq:ess-inclusion}
    \ess(\pmi_1\wedge\pmi_2) \subseteq \ess(\pmi_1) \cup \ess(\pmi_2).
\end{equation}
Taking the closure $\clw$ on both sides of \eqref{eq:pos-inclusion} and \eqref{eq:ess-inclusion}, and using \eqref{eq:corr-hyp-inv} and
\eqref{eq:pos-from-ess}, we obtain
\begin{align}
\label{eq:closures-proof}
    & \pos(\pmi_1\wedge\pmi_2) \supseteq \clw \left(\pos(\pmi_1) \cup \pos(\pmi_2)\right),\nonumber\\
    & \pos(\pmi_1\wedge\pmi_2) \subseteq \clw \left( \ess(\pmi_1) \cup \ess(\pmi_2)\right).
\end{align}
Since $\ess(\pmi_1)\subseteq\pos(\pmi_1)$, and similarly for $\pmi_2$, from the second relation in \eqref{eq:closures-proof} we get
\begin{equation}
    \pos(\pmi_1\wedge\pmi_2) \subseteq \clw \left( \pos(\pmi_1) \cup \pos(\pmi_2)\right),
\end{equation}
which combined with the first relation in \eqref{eq:closures-proof} gives \eqref{eq:meet-pos-bsets}.

To prove \eqref{eq:meet-partition}, we need to consider the various possibilities listed in \Cref{thm:gamma-partitions}. From the discussion above, it should be clear that if $\bsi{\X}{1}$ is positive but not essential,\footnote{\,The choice of $\bsi{\X}{1}$ is w.l.o.g., since otherwise we can simply swap 1 and 2 (and similarly, later in the proof).} and therefore $\Gamma_1(\X)$ is the trivial partition of $\X$, the same must be true for $\bsi{\X}{12}$, in agreement with \eqref{eq:meet-partition}, since in this case
\begin{equation}
    \Gamma_1(\X)\succeq\Gamma_2(\X) \implies \Gamma_{12}(\X)=\Gamma_1(\X)\vee\Gamma_2(\X)=\Gamma_1(\X).
\end{equation}
We then need to consider three possible cases: Both $\bsi{\X}{1}$ and $\bsi{\X}{2}$ are partial, they are both essential, or one is partial and the other is essential (from \Cref{thm:gamma-partitions}, we can view the situations where these $\bsets$ are vanishing or completely essential as special instances of these main three cases, we leave the details as an exercise for the reader).

Consider then two arbitrary $\N$-party KC-PMIs $\pmi_{1}$ and $\pmi_{2}$, and suppose that there exists a subsystem $\X$ such that both $\bsi{\X}{1}$ and $\bsi{\X}{2}$ are partial. Recall that any vanishing MI instance in $\bsi{\X}{1}$ has the property that its arguments constitute a partition of $\X$ that is coarser than $\Gamma_1(\X)$ (and similarly for 2). Using the partial order in the partition lattice of $\X$ we can then write this condition as
\begin{equation}
    \mi(\Y:\Z)\in\pmi_1\;\; \text{with}\;\; \Y\cup\Z=\X \iff \{\Y,\Z\} \succeq \Gamma_1(\X),
\end{equation}
and similarly for $\pmi_2$ and $\pmi_{12}$. Therefore, by \eqref{eq:kc-meet}, 
it must be that
\begin{equation}
\label{eq:gamma-above-join}
    \Gamma_{12}(\X) \succeq \Gamma_1(\X) \vee \Gamma_2(\X).
\end{equation}
To see this, notice that if $\Gamma_{12}(\X)$ and $\Gamma_1(\X) \vee \Gamma_2(\X)$ are incomparable (obviously it cannot be that $\Gamma_{12}(\X) \preceq \Gamma_1(\X) \vee \Gamma_2(\X)$), there is some bipartition $\{\Y',\Z'\}$ of $\X$ such that $\{\Y',\Z'\}\succeq\Gamma_{12}$ while $\{\Y',\Z'\}\nsucceq \Gamma_1$ (or $\Gamma_2$), contradicting the fact that (by \eqref{eq:kc-meet} and \Cref{thm:gamma-partitions}) $\{\Y',\Z'\}\succeq \Gamma_1(\X) \vee \Gamma_2(\X)$ for all MI instances of $\bsi{\X}{12}$ in $\pmi_{12}$.
Furthermore, if \eqref{eq:gamma-above-join} is strict, there is at least one MI instance which is vanishing in $\pmi_1$ and $\pmi_2$ but not in $\pmi_{12}$, which is a contradiction with \eqref{eq:kc-meet}. We have therefore proven \eqref{eq:meet-partition} in the particular case where $\bsi{\X}{1}$ and $\bsi{\X}{2}$ are both partial. Notice however that for essential $\bsets$, the characterization of the set of MI instances which belong to $\ac$ is exactly the same as for the vanishing MI instances in partial $\bsets$. Therefore, by the exact same argument, \eqref{eq:meet-partition} also holds when $\bsi{\X}{1}$ and $\bsi{\X}{2}$ are essential.

We are then left to prove \eqref{eq:meet-partition} in the case where $\bsi{\X}{1}$ is partial and $\bsi{\X}{2}$ is essential. Since any essential $\bset$ is positive, it follows from the discussion above that $\bsi{\X}{12}$ is also positive, and the question is whether it is essential or not, and if so, what are the elements of $\ac_{12}$. Since any positive MI instance $\mi\in\bsi{\X}{1}$ is not in $\ac_1$, by \eqref{eq:kc-meet} it cannot be in $\ac_{12}$. Therefore the only elements of $\ac_2$ in $\bsi{\X}{2}$ which can be in $\ac_{12}$ are the ones that are in $\pmi_1$ (i.e., that are vanishing in $\bsi{\X}{1}$). But in fact, all MI instances $\mi$ which are in $\ac_2$ and $\pmi_1$ must be in $\ac_{12}$. To see this, notice that $\mi\in\ac_2$ implies that any $\mi'\preceq\mi$ is in $\pmi_2$, and by the fact that $\pmi_1$ is a KC-PMI, $\mi\in\pmi_1$ implies that any $\mi'\preceq\mi$ is also in $\pmi_1$. Therefore, again by \eqref{eq:kc-meet}, any $\mi'\preceq\mi$ is in $\pmi_{12}$, and $\mi$ is in $\ac_{12}$. Since any MI instance $\mi(\Y:Z)$ which is simultaneously in $\ac_2$ and $\pmi_1$ satisfies $\{\Y,\Z\} \succeq \Gamma_1(\X)$ and $\{\Y,\Z\} \succeq \Gamma_2(\X)$, we obtain again \eqref{eq:gamma-above-join}, and by the same argument as above, the relation cannot be strict, implying \eqref{eq:meet-partition}.
\end{proof}

\subsection{Varying the number of parties}
\label{subsec:cg-and-fg}

All our discussions so far pertained to a fixed number of parties $\N$. We now consider mappings between KC-PMIs involving different numbers of parties. We first discuss coarse-grainings, which are transformations where collections of parties are mapped to single parties, thereby effectively reducing $\N$. These transformations are relatively straightforward and have been considered previously to varying degrees of detail (see for example \cite{Hernandez-Cuenca:2022pst}). The focus here will again be on the correlation hypergraph perspective.

We then consider the opposite transformation, in which a single party is replaced by a collection of parties, thereby increasing $\N$. The structure of fine-grainings is considerably more complex, as there is in principle no upper bound on the number of parties that can be introduced. In the present work, we primarily introduce the relevant definitions and discuss only a few fundamental properties that will be useful in later sections. A more detailed treatment of fine-grainings of correlation hypergraphs will be postponed to future work \cite{graphoidal}.

Conventionally, when dealing with coarse-grainings and fine-grainings, we fix $\N' > \N$ throughout this work. Accordingly, a coarse-graining is a transformation from $\N'$ parties to $\N$, and conversely, a fine-graining is a transformation from $\N$ parties to $\N'$.

\subsubsection{Coarse-grainings}
\label{subsubsec:cg}

We structure the discussion of coarse-grainings into five steps, each corresponding to a distinct level of generalization or interpretation. We begin by introducing the natural operational definition of coarse-grainings for a density matrix. Next, we describe the associated transformation at the level of entropy vectors, which—being a linear map—naturally extends to the full entropy space, independently of realizability. From this, we derive the corresponding transformation at the level of KC-PMIs, and then reformulate this transformation in terms of the correlation hypergraph. Finally, we reinterpret the meet of the KC-lattice using coarse-grainings of correlation hypergraphs, highlighting the similarity with the corresponding standard construction in quantum mechanics based on tensor products of density matrices.

\paragraph{The coarse-graining map for density matrices:} To define a coarse-graining, we consider the following physical scenario. For a quantum system with $\N'$ parties described by a density matrix $\rho_{\N'}$, let $\ket{\psi}_{[\N']0'}$ be a purification by an ancillary party $0'$, referred to as the purifier. We now relabel each of the $\N'+1$ parties (including the purifier) with an integer from $\nsp$, which specifies the set of ``new'' parties (as well as the new purifier) after the coarse-graining. By convention, we choose the party labeled by $0$ as the new purifier. The resulting density matrix $\rho_\N$ after coarse-graining is then obtained by tracing out each party labeled by $0$. Notice that this transformation is more general than simply ``grouping'' parties from the original $\N'$-party system and mapping them to individual parties in the $\N$-party system (which effectively preserves $\rho$ and merely relabels the factors of the Hilbert space on which it acts). In particular, if we group only the purifier with certain other parties and leave the remaining ones untouched, the transformation effectively amounts to tracing out the parties grouped with the purifier.

A general coarse-graining procedure of this type is specified by a \textit{coarse-graining map} (CG-map), which is a surjective map 
\begin{align}
    \cgmap:\; \nspp &\rightarrow  \nsp\nonumber\\
    \ell' &\mapsto \ell = \cgmap(\ell').
\end{align}
Since $\N<\N'$, this map is not injective, and we will denote by $\cginv(\ell)$ the \textit{pre-image} in $\nspp$ of a party $\ell$ in the $\N$-party system. For a subsystem $\X\subseteq\nsp$ we then define
\begin{equation}
\label{eq:cgmap-power-set}
    \cginv(\X) = \bigcup_{\ell\in\X} \cginv(\ell).
\end{equation}

In what follows, it will often be convenient to view a CG map as being specified by a partition of $\nspp$, denoted by $\cgpart$, whose blocks are labeled by an injective indexing function, denoted by $\cgindex$, which assigns to each block an index from $\nsp$. Notice that the partition $\cgpart$ captures the fundamental structure of the coarse-graining transformation, since two different CG-maps from $\nspp$ to $\nsp$ with the same $\cgpart$ are related by a permutation of the parties in the $\N$-party system. In concrete examples, we will denote a specific map $\cgmap$ by the set of blocks of $\cgpart$, each one indexed by $\cgindex$. For example, the CG-map $\cgmap$ from $\N=3$ to $\N=2$ given by $\cgmap(1)=0$, $\cgmap(2)=1$, $\cgmap(3)=2$, $\cgmap(0)=1$ is written as 
\begin{equation}
    \cgmap:\quad\{\{1\}_0,\{2,0\}_1,\{3\}_2\}.
\end{equation}

\paragraph{Transformation of entropy vectors under coarse-grainings:} Having defined the transformation of a density matrix under a CG-map, we can now ask what is the corresponding transformation for entropy vectors. Given a density matrix $\rho_{\N'}$, and a CG-map $\cgmap$, consider a subsystem $\X$ after the coarse-graining, i.e., $\varnothing\neq\X\subseteq[\N]$. The entropy $\ent_\X$ of $\X$ is simply the entropy $\ent'_{\X'}$ of the subsystem $\X'$ in the $\N'$-party system, where $\X'$ is the collection of all $\N'$-parties (including possibly the purifier) which are mapped by $\cgmap$ to one of the parties in $\X$. Explicitly, we have
\begin{equation}
\label{eq:ent-vec-cg}
    \ent_{\X}=\ent'_{\cginv(\X)}\;.
\end{equation}
Notice that to describe this transformation we have used subsystems $\X,\X'$ that can include the purifier, even if we are focusing on entropy vectors. The reason is that an arbitrary CG-map may permute the parties in such a way that the purifier in the $\N$-party system is not the image of the purifier in the $\N'$-party system. As usual, in this case we implicitly assume that each new component $\ent_{\X}$ where $\X$ includes the purifier has been replaced with the entropy of the complementary subsystem.

Since $\Svec$ is obtained from $\Svecp$ by simply ignoring some components and possibly permuting the others, this transformation between entropy vectors from $\N'$-party entropy space to $\N$-party entropy space is clearly a linear transformation that is completely determined by $\cgmap$, and we will denote it by $\cgsvec$. Since not all entropy vectors in entropy space can be realized by a density matrix, we then take the map $\cgsvec$ as the definition of a coarse-graining of an entropy vector specified by $\cgmap$, independently from realizability
\begin{align}
    \cgsvec:\; \mathbb{R}^{2^{\N'}-1} &\rightarrow  \mathbb{R}^{2^{\N}-1}\nonumber\\
    \Svecp &\mapsto \Svec = \cgsvec(\Svecp).
\end{align}

\paragraph{Transformation of PMIs under coarse-grainings:} We now want to understand how PMIs, and in particular KC-PMIs, transform under coarse-grainings, but first let us briefly review the basic definitions and properties of maps between PMIs and entropy vectors, and between PMIs and faces of the SAC. As we mentioned in \S\ref{subsec:review} (cf., \eqref{eq:pmi-map-def}), given a face $\face$ of the SAC$_\N$, the PMI of $\Svec$, $\pmimap(\Svec)$, does not depend on the choice of $\Svec$, as long as $\Svec\in\text{int}(\face)$. We can then introduce a map $\pi$ between faces of SAC$_\N$ and PMIs which associates to a face $\face$ of SAC$_\N$ the PMI of an arbitrary entropy vector in the interior of $\face$
\begin{align}
    \pi:\;\lat{SAC}^{\N} &\rightarrow \lat{PMI}^{\N}\nonumber\\
     \face &\mapsto \pmi=\pi(\face)=\{\mi(\X:\Y)\ |\ \; \mi(\X:\Y)(\Svec)=0,\Svec\in\text{int}(\face)\},
\end{align}
where $\lat{SAC}^{\N}$ is the lattice of faces of SAC$_\N$ and $\lat{PMI}^{\N}$ is the lattice of all $\N$-party PMIs. Since each face of SAC$_\N$ is characterized uniquely by the set of MI instances which vanish for an entropy vector in its interior, and we have defined PMIs to be such sets, $\pi$ is a bijection. On the other hand, notice that the map $\pmimap$ is not invertible, because it is not injective (see \cite{He:2022bmi} for more details on these issues).
 
Suppose now that we are given an entropy vector $\Svecp$ in $\N'$-party entropy space and a coarse-graining map $\cgmap$, and we want to find the PMI after the coarse-graining transformation. Let $\pmi'=\pmimap(\Svecp)$ be the PMI of $\Svecp$. To find the PMI $\pmi$ after the coarse-graining, we can first map $\Svecp$ to $\Svec$ via $\cgsvec$, and then compute $\pmi=\pmimap(\Svec)$.\footnote{\,Clearly the map $\pmimap$ depends on $\N$, but we keep this dependence implicit to simplify the notation.} Because of \eqref{eq:ent-vec-cg}, the relation between the \textit{values} of the MI instances in the $\N'$ and $\N$-party systems is 
\begin{equation}
\label{eq:mi-cg-eval}
    \mi(\X:\Y)(\Svec)=\mi(\cginv(\X):\cginv(\Y))(\Svecp).
\end{equation}
Therefore, given $\Svecp$ and $\cgmap$, and introducing the shorthand
\begin{equation}
     \cginv(\mi)\coloneq\cginv(\mi(\X:\Y))\coloneq \mi(\cginv(\X):\cginv(\Y))
\end{equation}
the PMI of $\Svec$ is given by
\begin{equation}
\label{eq:pmi-cg-ent-vec}
    \pmi=\pmimap(\cgsvec(\Svecp))=\{\mi\in\mgs_\N \ |\ \;  \cginv(\mi)=0\}.
\end{equation}
For fixed $\cgmap$, the set of MI instances on the right hand side of \eqref{eq:pmi-cg-ent-vec} does not depend on the specific choice of $\Svecp$, as long as this choice is confined to the interior of the same face. Therefore, for a fixed CG-map $\cgmap$, we can define the following coarse-graining map between PMIs
\begin{align}
\label{eq:pmi-cg}
    \cgpmi:\; \lat{PMI}^{\N'} &\rightarrow \lat{PMI}^{\N}\nonumber\\
     \pmi' &\mapsto \pmi=\cgpmi(\pmi')=
     \{\mi\in\mgs_\N \ |\ \; \cginv(\mi)\in\pmi'\}.
\end{align}

This definition of PMI coarse-grainings is by construction consistent with the definition of coarse-grainings for entropy vectors, in the sense that for any\footnote{\,Recall that the map $\pmimap$ is defined only for entropy vectors in the SAC. This is consistent, since any entropy vector $\Svecp$ in SAC$_\N'$ is mapped by $\cgsvec$ to an entropy vector $\Svec$ in SAC$_\N$.}
given $\Svecp$ and $\cgmap$ we have 
\begin{equation}
\label{eq:pmi-cg-consistency}
    \pmimap\left(\cgsvec(\Svecp)\right)=\cgpmi\left(\pmimap(\Svecp)\right).
\end{equation}
However, it is important to notice that we cannot write a relation analogous to \eqref{eq:pmi-cg-consistency} for faces of the SAC, rather than entropy vectors. Consider a face $\face'$ of the $\N'$-party SAC, and its corresponding PMI $\pmi'=\pi(\face')$. Given a coarse-graining $\cgmap$, we can map $\pmi'$ to $\pmi=\cgpmi(\pmi')$ and find the corresponding face $\face=\pi^{-1}(\pmi)$ of the $\N$-party SAC. However, if we map $\text{int}(\face')$ to $\N$-party entropy space via $\cgsvec$, in general we only have
\begin{equation}
    \label{eq:faces-cg-inclusion}
    \cgsvec\left(\text{int}(\face')\right)\subseteq \text{int}(\face)
\end{equation}
and the inclusion can be strict. While this inclusion guarantees that for every $\Svec\in\cgsvec(\text{int}(\face'))$ we have $\pmimap(\Svec)=\pmi$, the map $\pi(\cgsvec(\text{int}(\face')))$ is not well defined because $\cgsvec(\text{int}(\face'))$ in general is not the full interior of a face of the $\N$-party SAC, but a proper subset. Furthermore, it is also important to notice that even if $\cgsvec$ is a linear map, we can have $\text{dim}(\face)>\text{dim}(\face')$. 

To summarize, $\cgpmi$ is a well-defined coarse-graining map for PMIs, and $\pi$ is a bijection between PMIs and faces of the SAC for both the $\N'$ and $\N$-party systems, but the composition of these maps does not correspond to the entropy-vector-wise coarse-graining of faces of the SAC via $\cgsvec$, i.e., 
\begin{equation}
    (\pi^{-1}\circ \cgpmi\circ \pi)(\face') \neq \cgsvec (\face').
\end{equation}

An example of this situation and the fact that we could have $\text{dim}(\face)>\text{dim}(\face')$, is the ``perfect state'' at $\N'=3$, i.e., the (pure) state of four parties which is maximally entangled for all bipartitions, whose  entropy vector and PMI are
\begin{align}
\label{eq:CG-example}
    & \Svecp=(1,1,1,2,2,2,1)\nonumber\\
    & \pmi'=\{\mi(1:2),\mi(1:3),\mi(1:0),\mi(2:3),\mi(2:0),\mi(3:0)\},
\end{align}
corresponding to an extreme ray of the SAC$_3$.\footnote{\,This is also an extreme ray of HEC$_3$, realized by the star graph with four leaves and unit weights on all edges (see \S\ref{subsec:graph-review} for a brief review on graph models).} Under the CG-map
\begin{equation}
    \cgmap:\quad \{\{1,2\}_1,\{3\}_2,\{0\}_0\},
\end{equation}
the entropy vector becomes $\Svec=(2,1,1)$, and the PMI $\pmi=\{\mi(2:0)\}$, corresponding to a 2-dimensional face $\face$ of SAC$_2$. Since $\Svecp$ generates an extreme ray, which is by itself a face $\face'$ of SAC$_3$, the image under the coarse-graining is a proper subset of $\face$. Furthermore, it only spans a 1-dimensional subspace, generated by $\Svec$, and accordingly, $\text{dim}(\face)=2>1=\text{dim}(\face')$.

Finally, in what follows we will mostly be interested in KC-PMIs, rather than arbitrary PMIs, and correspondingly in the lattice $\lat{KC}\subseteq\lat{PMI}$. Accordingly, we focus on the restriction of the map $\cgpmi$ to $\lat{KC}^{\N'}$, and we need to verify that for any $\N'$ and $\cgmap$, any element of $\lat{KC}^{\N'}$ is mapped to an element of $\lat{KC}^{\N}$. The fact that this is indeed the case can immediately be seen by noticing that for any pair of MI instances $\mi_1,\mi_2\in\mgs_{\N}$ such that $\mi_1\preceq\mi_2$ we have $\cginv(\mi_1)\preceq\cginv(\mi_2)$. To avoid introducing unnecessary notation, in the following we will denote by $\cgpmi$ also its restriction to $\lat{KC}$.

\paragraph{The point of view of $\bsets$ and the correlation hypergraph:} We now want to understand coarse-grainings of KC-PMIs from the point of view of $\bsets$ and the correlation hypergraph. In principle, given any KC-PMI $\pmi'$ and CG-map $\cgmap$, we can simply compute $\pmi=\cgpmi(\pmi')$ and  reconstruct its correlation hypergraph $\hp$ by finding the positive $\bsets$. However, as in the case of the meet of the KC-lattice discussed in the previous subsection, our goal will be to understand the direct transformation of $\hpp$ to $\hp$, without going through this reconstruction. To understand how this works, we begin by directly applying \eqref{eq:pmi-cg} to see how the set of positive $\bsets$ transforms.
We will then prove a theorem about a reformulation of this result, which in terms of the correlation hypergraph takes the form of a hypergraph quotient by the equivalence relation on the vertices (i.e., the parties) specified by the CG-map, followed by the usual indexing of the equivalence classes, and the weak union closure introduced in \S\ref{subsec:kc-lattice} to describe the meet of the KC-lattice.

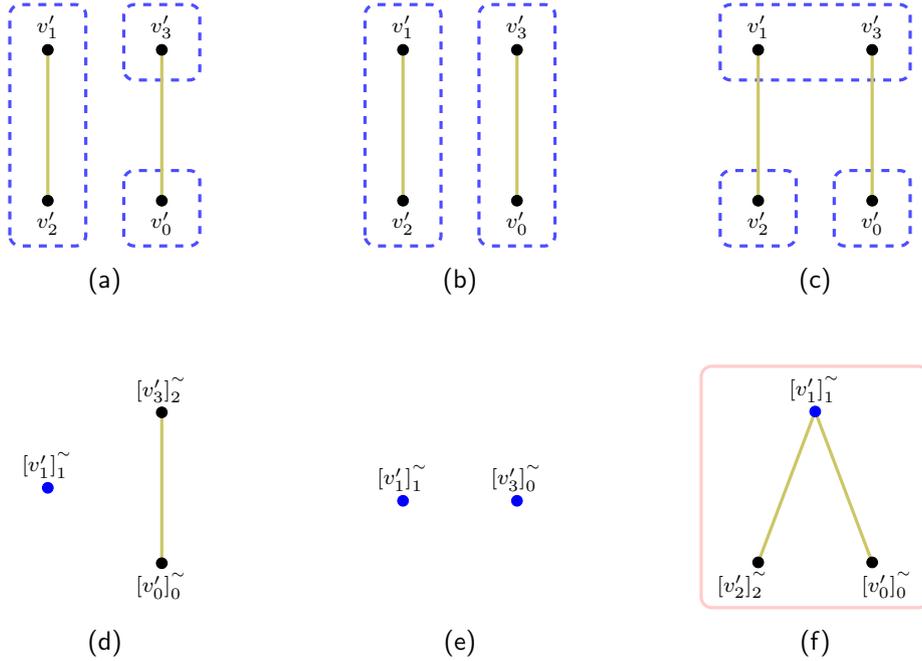
\begin{figure}[tbp]
    \centering
    \begin{subfigure}{0.3\textwidth}
    \centering
    \begin{tikzpicture}

    \draw[blue!70, rounded corners, very thick, dashed] (-0.5,-0.6) rectangle (0.5,2.6);
    \draw[blue!70, rounded corners, very thick, dashed] (1,1.6) rectangle (2,2.6);
    \draw[blue!70, rounded corners, very thick, dashed] (1,0.4) rectangle (2,-0.6);
    
    \draw[Mcol!80!black, very thick] (0,0) -- (0,2);
    \draw[Mcol!80!black, very thick] (1.5,0) -- (1.5,2);
    
    \filldraw (0,0) circle (2pt);
    \filldraw (0,2) circle (2pt);
    \filldraw (1.5,0) circle (2pt);
    \filldraw (1.5,2) circle (2pt);
    
    \node[] () at (0,-0.3) {{\scriptsize $v'_2$}};
    \node[] () at (0,2.3) {{\scriptsize $v'_1$}};
    \node[] () at (1.5,2.3) {{\scriptsize $v'_3$}};
    \node[] () at (1.5,-0.3) {{\scriptsize $v'_0$}};

    \end{tikzpicture}
    \subcaption[]{}
    \end{subfigure}
    \begin{subfigure}{0.3\textwidth}
    \centering
    \begin{tikzpicture}

    \draw[blue!70, rounded corners, very thick, dashed] (-0.5,-0.6) rectangle (0.5,2.6);
    \draw[blue!70, rounded corners, very thick, dashed] (1,-0.6) rectangle (2,2.6);
    
    \draw[Mcol!80!black, very thick] (0,0) -- (0,2);
    \draw[Mcol!80!black, very thick] (1.5,0) -- (1.5,2);
    
    \filldraw (0,0) circle (2pt);
    \filldraw (0,2) circle (2pt);
    \filldraw (1.5,0) circle (2pt);
    \filldraw (1.5,2) circle (2pt);
    
    \node[] () at (0,-0.3) {{\scriptsize $v'_2$}};
    \node[] () at (0,2.3) {{\scriptsize $v'_1$}};
    \node[] () at (1.5,2.3) {{\scriptsize $v'_3$}};
    \node[] () at (1.5,-0.3) {{\scriptsize $v'_0$}};

    \end{tikzpicture}
    \subcaption[]{}
    \end{subfigure}
    \begin{subfigure}{0.3\textwidth}
    \centering
    \begin{tikzpicture}

    \draw[blue!70, rounded corners, very thick, dashed] (-0.5,1.6) rectangle (2,2.6);
    \draw[blue!70, rounded corners, very thick, dashed] (-0.5,0.4) rectangle (0.5,-0.6);
    \draw[blue!70, rounded corners, very thick, dashed] (1,0.4) rectangle (2,-0.6);
    
    \draw[Mcol!80!black, very thick] (0,0) -- (0,2);
    \draw[Mcol!80!black, very thick] (1.5,0) -- (1.5,2);
    
    \filldraw (0,0) circle (2pt);
    \filldraw (0,2) circle (2pt);
    \filldraw (1.5,0) circle (2pt);
    \filldraw (1.5,2) circle (2pt);
    
    \node[] () at (0,-0.3) {{\scriptsize $v'_2$}};
    \node[] () at (0,2.3) {{\scriptsize $v'_1$}};
    \node[] () at (1.5,2.3) {{\scriptsize $v'_3$}};
    \node[] () at (1.5,-0.3) {{\scriptsize $v'_0$}};

    \end{tikzpicture}
    \subcaption[]{}
    \end{subfigure}
    \\ \bigskip\bigskip
       \begin{subfigure}{0.3\textwidth}
    \centering
    \begin{tikzpicture}
    
    \draw[Mcol!80!black, very thick] (1.5,0) -- (1.5,2);
    
    \filldraw[blue] (0,1) circle (2pt);
    \filldraw (1.5,0) circle (2pt);
    \filldraw (1.5,2) circle (2pt);
    
    \node[] () at (0,1.3) {{\scriptsize $\eqcl{v'_1}_1$}};
    \node[] () at (1.5,2.3) {{\scriptsize $\eqcl{v'_3}_2$}};
    \node[] () at (1.5,-0.3) {{\scriptsize $\eqcl{v'_0}_0$}};
    
    \end{tikzpicture}
    \subcaption[]{}
    \end{subfigure}
    \begin{subfigure}{0.3\textwidth}
    \centering
    \begin{tikzpicture}
    
    \filldraw[blue] (0,1) circle (2pt);
    \filldraw[blue] (1.5,1) circle (2pt);
    
    \node[] () at (0,-0.3) {{\scriptsize ${}$}};
    \node[] () at (0,1.3) {{\scriptsize $\eqcl{v'_1}_1$}};
    \node[] () at (1.5,1.3) {{\scriptsize $\eqcl{v'_3}_0$}};
    \node[] () at (1.5,-0.3) {{\scriptsize ${}$}};
    
    \end{tikzpicture}
    \subcaption[]{}
    \end{subfigure}
    \begin{subfigure}{0.3\textwidth}
    \centering
    \begin{tikzpicture}
    
    \draw[PminusEcol, rounded corners, very thick] (-1.5,-0.6) rectangle (1.5,2.6);
    
    \draw[Mcol!80!black, very thick] (-0.75,0) -- (0,2);
    \draw[Mcol!80!black, very thick] (0.75,0) -- (0,2);
    
    \filldraw (-0.75,0) circle (2pt);
    \filldraw (0.75,0) circle (2pt);
    \filldraw[blue] (0,2) circle (2pt);
    
    \node[] () at (-0.95,-0.3) {{\scriptsize $\eqcl{v'_2}_2$}};
    \node[] () at (0.95,-0.3) {{\scriptsize $\eqcl{v'_0}_0$}};
    \node[] () at (0,2.3) {{\scriptsize $\eqcl{v'_1}_1$}};
    
    \end{tikzpicture}
    \subcaption[]{}
    \end{subfigure}
    \caption{An example of a KC-PMI with a partial $\bset$, and its transformations under different CG-maps. Panels (a), (b) and (c) show the  (disconnected) correlation hypergraph $\hpp$
    for the same $\N=3$ KC-PMI $\pmi'$, corresponding to the tensor product of two Bell pairs between parties $\{1,2\}$ and $\{3,0\}$. The dashed blue lines indicate the partition $\cgpart$ associated to three different CG-maps (to $\N=2$ for (a) and (c), and to $\N=1$ for (b)). Panels (d), (e), and (f) show (respectively) the correlation hypergraph $\hp$ after the coarse-graining transformation, with blue vertices corresponding to the non-trivial equivalent classes. To simplify the figure, we have omitted labels for the hyperedges. Notice that in $\pmi'$, the $\bset$ $\bs{1230}$ is partial, but it is mapped to a partial $\bset$ in (d), a vanishing one in (e), and a positive one in (f). Furthermore, in (f), notice that the quotient hypergraph does not have the red hyperedge, which is instead added by the weak union closure (since the two yellow hyperedges have non-empty intersection). This example also shows that an essential $\bset$ in the $\N$-party system, for example $\bs{12}$ in (f), can result from the coarse-graining of a partial $\bset$ in the $\N'$-party system, in this case $\bs{123}$ in (c).}
    \label{fig:quotient}
\end{figure}

For any given $\N$, KC-PMI $\pmi$, and subsystem $\X$, we denote by $\mipos{\X,\pmi}$ (respectively, $\mivan{\X,\pmi}$) the set of positive (respectively, vanishing) MI instances in the $\bset$ $\bs{\X}$, explicitly
\begin{align}
    & \mipos{\X,\pmi}=\{\mi\in\bs{\X},\; \mi\in \cpmi\},\nonumber\\
    & \mivan{\X,\pmi}=\{\mi\in\bs{\X},\; \mi\in \pmi\}.
\end{align}
It follows from \eqref{eq:pmi-cg} that for any KC-PMI $\pmi'$, CG-map $\cgmap$, and subsystem $\X$ in the $\N$-party system with $|\X|\geq 2$, the sets $\mipos{\X,\pmi}$ and $\mivan{\X,\pmi}$ for $\pmi=\cgpmi(\pmi')$ are given by 
\begin{align}
    \label{eq:bs-cg-components}
    & \mipos{\X,\pmi}=\{\mi\in\bs{\X},\; \cginv(\mi)\in(\pmi')^\complement\},\nonumber\\
    & \mivan{\X,\pmi}=\{\mi\in\bs{\X},\; \cginv(\mi)\in\pmi'\}.
\end{align}
From these relations we immediately obtain the following implications
\begin{align}
\label{eq:cg-pos-van-bs}
    & \bs{\cginv(\X)}\;  \text{is positive and $|\X|\geq 2$} \; \implies\; \bs{\X}\;  \text{is positive}\nonumber\\
    & \bs{\cginv(\X)}\;  \text{is vanishing} \; \implies\; \bs{\X}\;  \text{is vanishing}.
\end{align}
On the other hand, if $\bs{\cginv(\X)}$ is partial, $\bs{\X}$ can be positive, partial or vanishing, as the example in \Cref{fig:quotient} demonstrates. 

Consider then a subsystem $\X$ such that $\bs{\cginv(\X)}$ is partial, and let $\X'=\cginv(\X)$. We want to determine under what condition $\bs{\X}$ is positive, partial or vanishing. In $\pmi'$, an MI instance in $\bs{\X'}$ vanishes if and only if it does not split any component of $\Gamma(\X')$. However, not all MI instances in $\bs{\X'}$ are of the form $\cginv(\mi)$ for some $\mi\in\mgs_{\N}$, and we can conveniently identify those that are using another partition of $\X'$, namely the one inherited from the coarse-graining map. To this end, consider the restriction $\cgpartrest{\X'}$ of $\cgpart$ to the subsystem $\X'$ (this is well defined because by assumption $\X'$ is a union of elements of $\cgpart$). The MI instances of $\bs{\X'}$ with the form $\cginv(\mi)$ for some $\mi\in\bs{\X}$ are then precisely those that do not split any element of $\cgpartrest{\X'}$. Therefore, the elements of $\mivan{\X,\pmi}$ are the MI instances such that $\cginv(\mi)$ does not split any element of the partition
\begin{equation}
\label{eq:cg-partial-partition}
     \Gamma(\X') \vee \cgpartrest{\X'} \ .
\end{equation}
In order for $\bs{\X}$ to be positive then, \eqref{eq:cg-partial-partition} must be the trivial partition of $\X'$. 

In summary, the positive $\bsets$ of $\pmi$, and therefore the hyperedges of $\hp$, are the ones such that $\bs{\X'}$ is positive for $\pmi'$, or it is partial but \eqref{eq:cg-partial-partition} is trivial. We will now see how this result for the correlation hypergraph can be rephrased in a nicer form using the notion of a hypergraph quotient.

\begin{defi}[Hypergraph quotient] 
    Given a hypergraph $\gf{H}=(V,E)$, and an equivalence relation $\sim$ on the vertex set $V$, the quotient hypergraph $\gf{H}^{\sim}=\gf{H}/\sim$ is the hypergraph $\gf{H}^{\sim}=(V^{\sim},E^{\sim})$ where $V^{\sim}$ is the set of equivalence classes of $V$, i.e.,
    \begin{equation}
        V^{\sim}=\{\eqcl{v}, v\in V\},
    \end{equation}
    and $E^{\sim}$ is the set of collections $h^{\sim}\subseteq V^{\sim}$ of equivalence classes for which $|h^{\sim}|\geq 2$, and there is $h\in E$ such that 
    \begin{equation}
    \label{eq:hypere-quotient-hyperg}
      h\subseteq \bigcup\; \{[v]^{\sim}\;|\; [v]^{\sim}\!\in h^{\sim}\}\;\;\;  \text{\emph{and}}\;\;\;  h\cap [v]^{\sim}\neq\varnothing\;\; \forall\; [v]^{\sim}\!\in h^{\sim}.
    \end{equation}
\end{defi}

Pictorially, one can visualize this transformation as one where the vertices in the same equivalence class are identified, and a new hyperedge is a collection of equivalence classes such that there exists a hyperedge in the original hypergraph that contains at least one vertex for each class of the collection, and no vertex in an equivalence class which is not in the collection (cf., \eqref{eq:hypere-quotient-hyperg}). 
With this definition, we can now prove the following suggestive result about the transformation of the correlation hypergraph of a KC-PMI under coarse-grainings. In the statement of the theorem, we denote by $\cgmapeq$ the equivalence relation induced by the partition $\cgpart$ associated to the CG-map $\cgmap$, and introduce the notation $\cgpmi(\hpp)$ for the coarse-graining of a correlation hypergraph.\footnote{\,To simplify the notation, we use the same symbol that we used for the coarse-graining of PMIs, and its restriction to KC-PMIs, since this is effectively the same map. The only difference is how a KC-PMI is represented.}

\begin{thm}
\label{thm:corr-hyp-cg}
    For any $\N'$-party \emph{KC-PMI} $\pmi'$, and any \emph{CG}-map $\cgmap$, the correlation hypergraph $\hp$ of $\pmi=\cgpmi(\pmi')$ is 
    \begin{equation}
    \label{eq:corr-hyp-cg}
        \hp = \cgpmi (\hpp) =(\text{\emph{cl}}_{\cup}\circ\cgindex)\left(\hpp/\cgmapeq\right).
    \end{equation}
\end{thm}
\begin{proof}
    Let $\gf{H}$ denote the hypergraph obtained from the right hand side of \eqref{eq:corr-hyp-cg}. We want to show that this is equal to $\hp$, the correlation hypergraph of $\pmi$ obtained from the coarse-graining. The fact that $\gf{H}$ has the correct number of vertices is trivial, and we therefore focus on the hyperedges. We first show that each hyperedge in $\hp$ is also in $\gf{H}$, and then that $\gf{H}$ does not contain any additional hyperedge.

    Consider a hyperedge $h_\X$ of $\hp$, and let $\X'=\cginv(\X)$. As discussed above, either $\bs{\X'}$ is positive, or it is partial and \eqref{eq:cg-partial-partition} is the trivial partition. We will consider these two possibilities in turn.
    Suppose that $\bs{\X'}$ is positive, and therefore there is a hyperedge $h'_{\X'}$ in $\hpp$. To see that $h_\X$ is also in $\gf{H}$, it suffices to notice that
    \begin{equation}
        h'_{\X'}= \bigcup_{\ell\in\X} \eqcl{v'}_{\ell}
    \end{equation}
    and therefore satisfies \eqref{eq:hypere-quotient-hyperg}, implying that $h_\X$ is an element of $\cgindex(\hpp/\cgmapeq)$. If instead $\bs{\X'}$ is partial and \eqref{eq:cg-partial-partition} is the trivial partition, the hyperedge $h_\X$ of $\hp$ is in $\cgindex(\hpp/\cgmapeq)$ only if there is a block of $\Gamma(\X')$ with non-empty intersection with all blocks of $\cgpartrest{\X'}$, otherwise \eqref{eq:hypere-quotient-hyperg} is not satisfied (see \Cref{fig:quotient} for an example). 
    
    We now show that in this case the hyperedge $h_\X$ is in $\gf{H}$ because of the weak union closure. Consider an element $\Y'$ of $\Gamma(\X')$, and the collection of all blocks of $\cgpartrest{\X'}$ that have non-empty intersection with $\Y'$. We denote the union of these blocks by $\Y'_\sim$. By the definition of the quotient, there is in $\gf{H}$ a hyperedge $h_{\Y_\sim}$ such that $\cginv(\Y_\sim)=\Y'_\sim$. If $\Y'_\sim\subset\X'$  strictly (which is the case we are considering), there must then be at least one element $\Z'$ of $\Gamma(\X')$ with non-empty intersection with $\Y'_\sim$, otherwise \eqref{eq:cg-partial-partition} is not the trivial partition. Repeating the same construction we obtain a subsystem $\Z'_\sim$ such that, by the quotient, there is in $\gf{H}$ a hyperedge $h_{\Z_\sim}$ where $\cginv(\Z_\sim)=\Z'_\sim$. But since $\Y'_\sim\cap\Z'_\sim \neq \varnothing$, it follows that $\Y_\sim\cap\Z_\sim \neq \varnothing$, and the hyperedge $h_{\Y_\sim\cup\Z_\sim}$ is added to $\gf{H}$ by the weak union closure. Lastly, if $\Y'_\sim\cap\Z'_\sim\subset \X'$, we can repeat this construction considering another element of $\Gamma(\X')$ with non-empty intersection with $\Y'_\sim\cap\Z'_\sim$, and proceed in this fashion until we have covered the full subsystem $\X'$ and we obtain $h_\X$ in $\gf{H}$.
    
    Having showed that all hyperedges of $\hp$ are in $\gf{H}$, to complete the proof we now show that the opposite inclusion also holds. Since any correlation hypergraph is closed with respect to the weak union closure (cf., \eqref{eq:corr-hyp-inv}), it is sufficient to show that every element of the quotient is in $\hp$. To see this, consider a hyperedge $h_{\X'}$ of $\hpp$ for some subsystem $\X'$. Similarly to before, we denote by $\X'_\sim$ the union of all blocks of $\cgpart$ (now without restrictions to a particular subsystem) which have non-empty intersection with $\X'$. The quotient then adds to $\gf{H}$ a hyperedge $h_{\X_\sim}$ where $\cginv(\X_\sim)=\X'_\sim$. But by construction, the partition \eqref{eq:cg-partial-partition} for the subsystem $\X'_\sim$ is trivial. Therefore $\bs{\X_\sim}$ is positive, and $h_{\X_\sim}$ is also a hyperedge of $\hp$, completing the proof.
\end{proof}

To conclude this paragraph, let us make a few comments about the transformation of $\Gamma$-partitions and essential $\bsets$ under coarse-grainings. If for a subsystem $\X$ in the $\N$-party system, the $\bset$ of $\X'=\cginv(\X)$ is partial, we have seen that $\bs{\X}$ is positive only if \eqref{eq:cg-partial-partition} is trivial. If it is not, and $\bs{\X}$ is partial, we have also seen that the vanishing MI instances in $\bs{\X'}$ which are mapped to vanishing ones $\bs{\X}$ are the ones that do not split any block of \eqref{eq:cg-partial-partition}. Therefore, if we simply map each block of \eqref{eq:cg-partial-partition} to the corresponding subsystem in the $\N$-party system, we obtain $\Gamma(\X)$. Accordingly, if $\bs{\X'}$ is vanishing, and $\Gamma(\X')$ is the singleton partition, so will be $\Gamma(\X)$, and $\bs{\X}$ is also vanishing, in agreement with \eqref{eq:cg-pos-van-bs}. Similarly, when $\bs{\X'}$ is positive but not essential, $\Gamma(\X')$ is trivial, and so is $\Gamma(\X)$. These observations seem to suggest that if $\bs{\X'}$ is partial, and \eqref{eq:cg-partial-partition} is trivial, $\bs{\X}$ can never be essential, since $\Gamma(\X)$ is trivial. However, this is not the case, as demonstrated by the examples in \Cref{fig:quotient}. Furthermore, notice that for any $\N'$-party KC-PMI and positive but not essential $\bset$ $\bs{\X'}$, if we specify a CG-map $\cgmap$ such that $\X'=\cginv(\X)$ and $|\X|=2$, then $\bs{\X}$ is necessarily essential, demonstrating that in general an essential $\bset$ can also result from the coarse-graining of a positive but not essential $\bset$. We leave the proper treatment of the transformation of $\Gamma$-partitions and essential $\bsets$ under coarse-grainings for future work \cite{graphoidal}.

\paragraph{Relationship between coarse-grainings and the meet:} The last goal of this subsection is to discuss how, using the correlation hypergraph, the meet in the KC-lattice can be formulated in terms coarse-grainings in complete analogy to a standard construction for density matrices and holographic graph models (see \S\ref{subsec:graph-review}). 

To make this connection clear, let us first describe a simple example in quantum mechanics. Consider two density matrices $\rho$ and $\sigma$ acting on a 2-party Hilbert space $\mathbb{H}_{12}$, with entropy vectors $\Svec_\rho,\,\Svec_\sigma$, and corresponding KC-PMIs $\pmi_\rho$ and $\pmi_\sigma$. If we introduce a copy $\mathbb{H}_{\hat{1}\hat{2}}$ of the Hilbert space and relabel the parties in $\sigma$ by $\hat{1}$ and $\hat{2}$, we can construct a 4-party quantum state acting on $\mathbb{H}_{12}\otimes\mathbb{H}_{\hat{1}\hat{2}}$ by simply taking the tensor product
\begin{equation}
    \tau_{12\hat{1}\hat{2}}=\rho_{12}\otimes\sigma_{\hat{1}\hat{2}}.
\end{equation}
If we then coarse-grain the parties according to the CG-map
\begin{equation}
    \cgmap:\quad \{\{1,\hat{1}\}_1,\{2,\hat{2}\}_2,\{0\}_0\}
\end{equation}
we rewrite $\tau_{12\hat{1}\hat{2}}$ as
\begin{equation}
    \tau_{12}=\rho_{1_12_1}\otimes\sigma_{1_22_2},
\end{equation}
where $1_1$ and $1_2$ indicate two components of the same party 1 (and similarly for 2). By additivity of the von Neumann entropy under tensor products, the entropy vector of $\tau_{12}$ is $\Svec_\tau=\Svec_\rho+\Svec_\sigma$, and its KC-PMI is (cf., \S\ref{subsec:kc-lattice})
\begin{equation}
    \pmi_\tau=\pmi_\rho\wedge\pmi_\sigma.
\end{equation}
We will now see how exactly the same construction can be applied to arbitrary KC-PMIs using the correlation hypergraph.

For any given $\N$, consider two KC-PMIs $\pmi_\rho$ and $\pmi_\sigma$, and the corresponding correlation hypergraphs $\gf{H}_\rho$ and $\gf{H}_\sigma$. We emphasize that we are no longer assuming that any of these PMIs is realizable, and that $\rho,\sigma$ are now simply labels to distinguish the PMIs and make the analogy more transparent, rather than actual density matrices. As in the example above, we now relabel the parties in $\pmi_\sigma$ as $\ell\mapsto\hat{\ell}$, and accordingly in $\gf{H}_\sigma$, denote the resulting correlation hypergraph by $\widehat{\gf{H}}_{\sigma}$. We then construct the \textit{hypergraph sum}\footnote{\,The sum of two hypergraphs is defined as the hypergraph whose vertex and hyperedge sets are the disjoint union of the corresponding sets for the two hypergraphs in the sum. In this case, since we have already labeled vertices and hyperedges by different parties, the disjoint union is simply the union.\label{ft:sum}}
\begin{equation}
    \gf{H}_{\tau}= \gf{H}_\rho\; \oplus\; \widehat{\gf{H}}_\sigma,
\end{equation}
which is a hypergraph with $2(\N+1)$ vertices. As we stressed before, an arbitrary hypergraph in general is not the correlation hypergraph of a KC-PMI, however, it follows from the analysis of \cite{He:2022bmi} that $\gf{H}_{\tau}$ \textit{is} a correlation hypergraph.\footnote{\,While \cite{He:2022bmi} did not use the correlation hypergraph, the translation is straightforward, and we leave the details as an exercise for the reader.} Intuitively, the reason is the following. Since $\gf{H}_{\tau}$ is disconnected, it must be that $\mi=\mi(01\ldots\N:\hat{0}\hat{1}\ldots\hat{\N})$ vanishes, and so does any other MI instance $\mi'\preceq\mi$ in the MI-poset, and the set of faces of SAC$_{2\N+1}$ which obeys this condition is isomorphic to the direct sum of two copies of SAC$_\N$. As long as $\gf{H}_\rho$ and $\widehat{\gf{H}}_\sigma$ are correlation hypergraphs of $\N$-party KC-PMIs (and here they are by construction), $\gf{H}_{\tau}$ is a correlation hypergraph.

Following the example in quantum mechanics, we now introduce the coarse graining map
\begin{equation}
    \cgmeet = \{\{\ell,\hat{\ell}\}_\ell \ |\ \;\ell\in \nsp\},
\end{equation}
and prove that the coarse-graining of $\gf{H}_{\tau}$ under $\cgmeet$ is precisely the correlation hypergraph of the meet of $\pmi_\rho$ and $\pmi_{\sigma}$.

\begin{thm}
\label{thm:meet-from-cg}
    For any $\N$, and any two \emph{KC-PMI}s $\pmi_\rho$ and $\pmi_\sigma$,
    \begin{equation}
    \gf{H}_{\pmi_\rho\wedge\pmi_\sigma} = \cgmeetpmi(\gf{H}_\rho \oplus \widehat{\gf{H}}_\sigma)
    \end{equation}
\end{thm}
\begin{proof}
    The statement follows immediately from \Cref{thm:meet} and \Cref{thm:corr-hyp-cg} by noticing that the set of hyperedges of the hypergraph 
    \begin{equation}
        \chi_{\cgmeet}\left( (\gf{H}_\rho \oplus \widehat{\gf{H}}_\sigma)/ \sim_{_{\cgmeet}} \right)
    \end{equation}
    is the union of the sets of hyperedges of $\gf{H}_\rho$ and $\gf{H}_\sigma$.
\end{proof}

\Cref{thm:meet-from-cg} can immediately be generalized to the meet of an arbitrary number of KC-PMIs. Furthermore, the reader familiar with holographic graph model will already recognize the analogy with the same operation in that context, which we will review in \S\ref{subsec:graph-review}.

\subsubsection{Fine-grainings}
\label{subsubsec:fg}

Having examined coarse-grainings in detail, we now focus on the opposite operation, where one or more parties in an $\N$-party system are mapped to collections of parties in a larger system. Physically, we can imagine a situation where one or more of the parties has a finer internal structure that we have initially ignored (for example, a collection of qubits), and we now want to reconsider.\footnote{\,In the context of holography this is of course very natural, since for any boundary spacial region we can always consider a refinement into subregions.} As we shall soon see, this transformation is defined in terms of the pre-image of a given coarse-graining map, and is thus, in principle, straightforward given the framework we established for coarse-grainings. Consequently, our primary focus will be on KC-PMIs within the language of the correlation hypergraph, with only a brief analysis of fine-grainings of entropy vectors provided at the conclusion of this discussion.

Consider then an arbitrary KC-PMI $\pmi$ in $\lat{KC}^\N$ and suppose, for simplicity, that we want to replace a single party $\ell$ with a collection of $k$ parties $\ell_1,\ldots\ell_k$. We could then define a \textit{fine-graining} of $\pmi$ as \textit{any} KC-PMI $\pmi'$ in $\lat{KC}^{\N+k-1}$ such that there exists a CG-map mapping $\pmi'$ to $\pmi$. It should be intuitively clear that typically this choice of $\pmi'$ is highly non unique, implying that there is no proper fine-graining map (unless we make a ``canonical'' choice, we will come back to this point below). Furthermore, part of this non-uniqueness is trivial, since it originates from the fact that we can freely permute the parties in the $\N'$-party system, and then ``compensate'' by redefining the coarse-graining map. For these reasons, we define \textit{a} fine-graining of $\pmi$ under a fixed coarse-graining map, as any element of the \textit{pre-image} of $\pmi$ under this map.\footnote{\,The following definition, as well as others below, have an obvious generalization to arbitrary PMIs, but we will ignore such generalizations in what follows.} 

\begin{defi}[Fine-grainings of a KC-PMI for fixed CG-map]
    For any $\N,\N'$, any $\N$-party \emph{KC-PMI} $\pmi$, and any fixed \emph{CG}-map $\cgmap$, the fine-grainings of $\pmi$ \emph{with respect to} $\cgmap$ are the elements of the set $\cgpmiinv(\pmi)$.
\end{defi}

At this stage, according to this definition, the set of fine-grainings of a KC-PMI under a fixed coarse-graining map is simply a set, with no additional internal structure. However, given that this set is a subset of lattice, and therefore it is naturally at least a poset, it is interesting to ask if it has additional structure. The answer to this question is intimately related to the general properties of a CG-map, which we now briefly explore.

It will be useful to first consider a particularly simple example of fine-grainings which exist for any $\N$-party PMI $\pmi$, any $\N'$, and any CG-map $\cgmap$. These fine-grainings, which in this work we will call \textit{canonical fine-graining}, were already introduced in \cite{He:2022bmi}, for arbitrary PMIs.\footnote{\,They were called ``canonical embeddings'' in \cite{He:2022bmi} and often ``canonical lifts'' in other works.} Here instead, since we are focusing on KC-PMIs,  we introduce the definition using correlation hypergraphs. 

\begin{defi}[Canonical fine-graining of a KC-PMI]
\label{defi:canonical-fg}
For any $\N$-party \emph{KC-PMI} $\pmi$, and any \emph{CG}-map $\cgmap$ from an $\N'$-party system, the correlation hypergraph of a canonical fine-graining of $\pmi$ with respect to $\cgmap$ is obtained from $\hp$ by first relabeling each vertex $v_l$ by a choice of single party $\ell'\in\cginv(\ell)$, and then adding isolated vertices for the remaining $\N'-\N$ parties.
\end{defi}

\begin{figure}[tbp]
    \centering
     \begin{subfigure}{0.3\textwidth}
    \centering
    \begin{tikzpicture}

    \draw[blue!70, rounded corners, very thick, dashed] (-1.5,-0.5) rectangle (1.5,0.5);
    \draw[blue!70, rounded corners, very thick, dashed] (-0.5,1) rectangle (0.5,2);
    
    \draw[Mcol!80!black, very thick] (-1,0) -- (0,1.5);
    
    \filldraw (-1,0) circle (2pt);
    \filldraw (1,0) circle (2pt);
    \filldraw (0,1.5) circle (2pt);
    
    \node[] () at (-1.2,-0.2) {{\scriptsize $v_1$}};
    \node[] () at (1.2,-0.2) {{\scriptsize $v_2$}};
    \node[] () at (0,1.8) {{\scriptsize $v_0$}};

    \end{tikzpicture}
    \subcaption[]{}
    \end{subfigure}
    \begin{subfigure}{0.3\textwidth}
    \centering
    \begin{tikzpicture}

    \draw[blue!70, rounded corners, very thick, dashed] (-1.5,-0.5) rectangle (1.5,0.5);
    \draw[blue!70, rounded corners, very thick, dashed] (-0.5,1) rectangle (0.5,2);
    
    \draw[Mcol!80!black, very thick] (1,0) -- (0,1.5);
    
    \filldraw (-1,0) circle (2pt);
    \filldraw (1,0) circle (2pt);
    \filldraw (0,1.5) circle (2pt);
    
    \node[] () at (-1.2,-0.2) {{\scriptsize $v_1$}};
    \node[] () at (1.2,-0.2) {{\scriptsize $v_2$}};
    \node[] () at (0,1.8) {{\scriptsize $v_0$}};

    \end{tikzpicture}
    \subcaption[]{}
    \end{subfigure}
    \begin{subfigure}{0.3\textwidth}
    \centering
    \begin{tikzpicture}
    
    \draw[Mcol!80!black, very thick] (0,0) -- (0,1.5);
    
    \filldraw[blue] (0,0) circle (2pt);
    \filldraw (0,1.5) circle (2pt);
    
    \node[] () at (0,-0.3) {{\scriptsize $\eqcl{v_1}$}};
    \node[] () at (0,1.8) {{\scriptsize $\eqcl{v_0}$}};

    \end{tikzpicture}
    \subcaption[]{}
    \end{subfigure}
    \caption{
    A simple example showing that the join in the KC-lattice is not preserved under coarse-grainings (as defined below), demonstrating the second statement of (iii) in \Cref{thm:cg-map-properties}. (a) and (b) show the correlation hypergraphs of two $\N=2$ KC-PMIs $\pmi_1'$ and $\pmi_2'$, corresponding to two extreme rays of SAC$_2$ generated by the entropy vectors $\Svec_1=(1,0,1)$ and $\Svec_2=(0,1,1)$ (respectively). The partition of the vertex set by the dashed blue line indicates the fixed CG-map, and the correlation hypergraph of the resulting $\N=1$ KC-PMI $\pmi_\bot$, which is the same for both PMIs and is the bottom of $\lat{KC}^1$, is shown in (c). Since $\pmi_1'$ and $\pmi_2'$ correspond to ERs of SAC$_2$, their join is the 0-dimensional KC-PMI $\pmi'_\top$ (the top of $\lat{KC}^2$), and its coarse-graining is $\pmi_\top$, the top of $\lat{KC}^1$. On the other hand, since the coarse-graining of $\pmi_1'$ and $\pmi_2'$ is the same, $\pmi_\bot\vee\pmi_\bot=\pmi_\bot\neq\pmi_\top$, and the join is therefore not preserved. Finally, notice that (a) and (b) can be viewed as two distinct canonical fine-grainings of (c) with respect to the same CG-map.}
    \label{fig:join-not-preserved}
\end{figure}
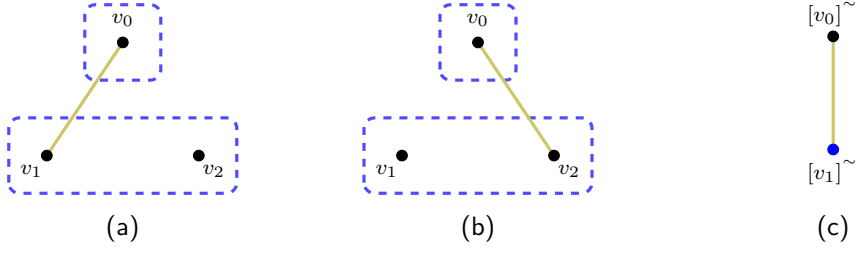

Notice that in general a canonical fine-graining is non-unique \textit{only} because we have some freedom to permute the parties in the $\N'$ party system, and the number of canonical fine-grainings is therefore 
\begin{equation}
\label{eq:n-can-fg}
    \phi(\cgmap) = \prod_{\ell\in\nsp} |\cginv(\ell)|.
\end{equation}
Up to these permutations however, canonical fine-grainings are unique. The non-uniqueness of canonical fine-grainings is illustrated by the simple example in \Cref{fig:join-not-preserved}.

We will also need a few basic definitions from the theory of partially ordered sets and lattices. Given two elements $x,y$ of a poset $\A$ with $x\preceq y$, the \textit{interval} between $x$ and $y$ is 
\begin{equation}
    [x,y]=\; \uparrow \!x\; \cap \downarrow \!y = \{z\in \A|\; x\preceq z\preceq y\}.
\end{equation}
A map $\phi$ from a poset to $\A$ to another poset $\B$ is said to be \textit{order preserving} if $x\preceq y$ in $\A$ implies $\phi(x)\preceq\phi(y)$ in $\B$. If $\A$ and $\B$ are lattices, $\phi$ is said to \textit{preserve the meet} if 
\begin{equation}
\label{eq:meet-preserved}
    \phi(x\wedge y) = \phi(x) \wedge \phi(y),
\end{equation}
and similarly for the join. We now prove a few main properties of the restriction of an arbitrary CG-map to the KC-lattice.
\begin{thm}
\label{thm:cg-map-properties}
    For any $\N,\N'$, and any \emph{CG}-map $\cgmap$, the restriction of the map $\cgpmi$ to the \emph{KC}-lattice:
    \begin{enumerate}[itemsep=0mm,label={\footnotesize \emph{\roman*)}}]
    \item is surjective, but not injective 
    \item is order-preserving
    \item preserves the meet, but not the join
    \end{enumerate}
\end{thm}
\begin{proof}
    (i) Surjectivity follows from the fact at least one canonical fine-graining exists for any KC-PMI and any CG-map. Since for any $\N,\N'$ we have $|\lat{KC}^\N|<|\lat{KC}^{\N'}|$ any CG-map is not injective.

    (ii) Follows immediately from the definition of coarse-grainings for PMIs and the fact that the partial order in the KC-lattice corresponds to  inclusion.

    (iii) To prove that the join in general is not preserved, it is sufficient to find an example, and we present one in \Cref{fig:join-not-preserved}. On the other hand, to show that the meet is preserved, we need to show that for any $\N,\N'$, any CG-map $\cgmap$, and any two $\N'$-party KC-PMIs $\pmi'_1,\pmi_2'$ we have
    \begin{equation}
        \cgpmi\left(\pmi'_1\wedge\pmi'_2\right) = \cgpmi(\pmi'_1) \wedge \cgpmi(\pmi'_2).
    \end{equation}
    To see this, it is convenient to use entropy vectors. Consider two arbitrary $\N'$-party entropy vectors $\Svecp_1$ and $\Svecp_2$ such that $\pmimap(\Svecp_1)=\pmi'_1$ and $\pmimap(\Svecp_2)=\pmi'_2$. We then have
    \begin{align}
        \cgpmi\left(\pmi'_1\wedge\pmi'_2\right) & = \pmimap\left(\cgsvec(\Svec'_1+\Svec'_2)\right)\nonumber\\
        & = \pmimap\left(\cgsvec(\Svec'_1) + \cgsvec(\Svec'_2)\right)\nonumber\\
        & = \pmimap(\cgsvec(\Svec'_1)) \wedge \pmimap(\cgsvec(\Svec'_2))\nonumber\\
        & = \cgpmi(\pmi'_1) \wedge \cgpmi(\pmi'_2)
    \end{align}
    where in the first line we used \eqref{eq:pmi-cg-ent-vec} and \eqref{eq:pmi-meet-svec-sum}, the second line follows from the fact that $\cgsvec$ is a linear map, and in the third line we used \eqref{eq:pmi-meet-svec-sum} again.
\end{proof}

Notice that the fact that $\cgpmi$ is order preserving implies that given three KC-PMIs in $\lat{KC}^{\N'}$ such that $\pmi_1'\preceq\pmi_2'\preceq\pmi'_3$ 
we have
\begin{equation}
    \cgpmi(\pmi_1')=\cgpmi(\pmi_3') \implies \cgpmi(\pmi_2')=\cgpmi(\pmi_1').
\end{equation}
Therefore, for any $\pmi\in\lat{KC}^{\N}$, its set of fine-grainings with respect to a CG-map $\cgmap$ is a union of intervals of $\lat{KC}^{\N'}$. Furthermore, the fact that any CG-map preserves the meet immediately implies that the poset $\cgpmiinv(\pmi)$ has a minimum. Indeed,
\begin{equation}
\label{eq:minimum}
    \cgmap\left(\bigwedge \cginv(\pmi)\right) = \pmi \wedge \pmi \wedge
     \ldots\wedge \pmi =\pmi.
\end{equation}
Therefore, $\cgpmiinv(\pmi)$ is a union of intervals with the same bottom element. 

The next result describes this \textit{minimum fine-graining} (as usual, with respect to a fixed CG-map) explicitly. Essentially, it says that the minimum fine-graining is simply the down-set of all MI instances $\cginv (\mi)$ for all $\mi$ in the original KC-PMI,\footnote{\,This is intuitively clear, but from this formulation it is not obvious that such down-set is a PMI, i.e., that it is compatible with linear dependence among the $\N'$-party MI instances and SA.} and that the corresponding hypergraph contains all possible hyperedges compatible with the absence of correlations in the lower party system. An example of the correlation hypergraph of a minimal fine-graining is shown in \Cref{fig:min-fg-hp}. To simplify the expressions in the following theorem, for given $\ell\in\nsp$ we introduce the notation
\begin{equation}
    V'_{\ell} = \{v_{\ell'}|\; \ell'\in\cginv(\ell)\},\qquad 
    V'_{\X} = \bigcup_{\ell\in\X} V'_\ell. 
\end{equation}

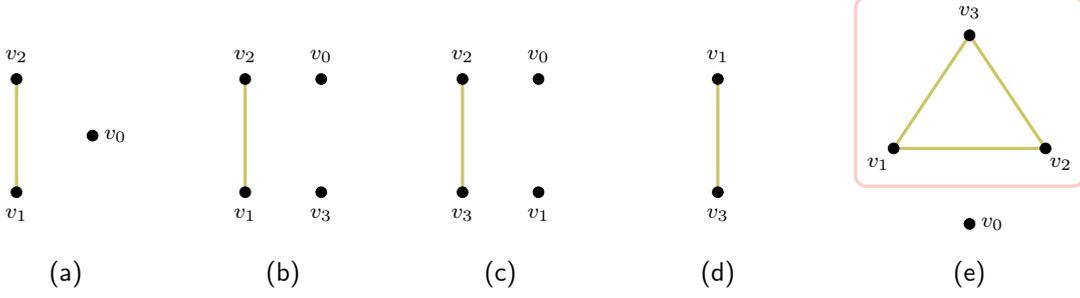
\begin{figure}[tbp]
    \centering
     \begin{subfigure}{0.18\textwidth}
    \centering
    \begin{tikzpicture}
    
    \draw[Mcol!80!black, very thick] (0,0) -- (0,1.5);
    
    \filldraw (0,0) circle (2pt);
    \filldraw (1,0.75) circle (2pt);
    \filldraw (0,1.5) circle (2pt);
    
    \node[] () at (0,-0.3) {{\scriptsize $v_1$}};
    \node[] () at (0,1.8) {{\scriptsize $v_2$}};
    \node[] () at (1.3,0.75) {{\scriptsize $v_0$}};

    \node[] () at (0,-0.5) {};

    \end{tikzpicture}
    \subcaption[]{}
    \end{subfigure}
    \begin{subfigure}{0.18\textwidth}
    \centering
    \begin{tikzpicture}
    
    \draw[Mcol!80!black, very thick] (0,0) -- (0,1.5);
    
    \filldraw (0,0) circle (2pt);
    \filldraw (1,1.5) circle (2pt);
    \filldraw (1,0) circle (2pt);
    \filldraw (0,1.5) circle (2pt);
    
    \node[] () at (0,-0.3) {{\scriptsize $v_1$}};
    \node[] () at (0,1.8) {{\scriptsize $v_2$}};
    \node[] () at (1,1.8) {{\scriptsize $v_0$}};
    \node[] () at (1,-0.3) {{\scriptsize $v_3$}};

    \node[] () at (0,-0.5) {};

    \end{tikzpicture}
    \subcaption[]{}
    \end{subfigure}
    \begin{subfigure}{0.18\textwidth}
    \centering
    \begin{tikzpicture}
    
    \draw[Mcol!80!black, very thick] (0,0) -- (0,1.5);
    
    \filldraw (0,0) circle (2pt);
    \filldraw (1,1.5) circle (2pt);
    \filldraw (1,0) circle (2pt);
    \filldraw (0,1.5) circle (2pt);
    
    \node[] () at (0,-0.3) {{\scriptsize $v_3$}};
    \node[] () at (0,1.8) {{\scriptsize $v_2$}};
    \node[] () at (1,1.8) {{\scriptsize $v_0$}};
    \node[] () at (1,-0.3) {{\scriptsize $v_1$}};

    \node[] () at (0,-0.5) {};

    \end{tikzpicture}
    \subcaption[]{}
    \end{subfigure}
    \begin{subfigure}{0.18\textwidth}
    \centering
    \begin{tikzpicture}
    
    \draw[Mcol!80!black, very thick] (0,0) -- (0,1.5);
    
    \filldraw (0,0) circle (2pt);
    \filldraw (0,1.5) circle (2pt);
    
    \node[] () at (0,-0.3) {{\scriptsize $v_3$}};
    \node[] () at (0,1.8) {{\scriptsize $v_1$}};

    \node[] () at (0,-0.5) {};

    \end{tikzpicture}
    \subcaption[]{}
    \end{subfigure}
    \begin{subfigure}{0.24\textwidth}
    \centering
    \begin{tikzpicture}

    \draw[PminusEcol, rounded corners, very thick] (-1.5,-0.5) rectangle (1.5,2);
    
    \draw[Mcol!80!black, very thick] (-1,0) -- (0,1.5);
    \draw[Mcol!80!black, very thick] (-1,0) -- (1,0);
    \draw[Mcol!80!black, very thick] (1,0) -- (0,1.5);
    
    \filldraw (-1,0) circle (2pt);
    \filldraw (1,0) circle (2pt);
    \filldraw (0,1.5) circle (2pt);
    \filldraw (0,-1) circle (2pt);
    
    \node[] () at (-1.2,-0.2) {{\scriptsize $v_1$}};
    \node[] () at (1.2,-0.2) {{\scriptsize $v_2$}};
    \node[] () at (0,1.8) {{\scriptsize $v_3$}};
    \node[] () at (0.3,-1) {{\scriptsize $v_0$}};

    \end{tikzpicture}
    \subcaption[]{}
    \end{subfigure}
    \caption{An example of a minimum fine-graining from the perspective of the correlation hypergraph as described by \Cref{thm:min-fg} and \Cref{cor:min-fg}: (a) shows the correlation hypergraph of an $\N=2$ KC-PMI $\pmi$, (b) and (c) show the canonical fine-grainings of $\pmi$ with respect to the CG-map $\cgmap$ specified by $\{\{1,3\}_1,\{2\}_2,\{0\}_0\}$, (d) is the complete hypergraph $\gf{H}^{(1)}_{\circledast}$ for the parties $1$ and $3$, and (e) is the correlation hypergraph of the minimal fine-graining $\pmi'$ of $\pmi$ with respect to $\cgmap$. Notice that according to \Cref{cor:min-fg}, (e) is the meet (cf., \Cref{thm:meet}) of (b), (c) and the canonical fine-graining of (d) obtained by adding isolated vertices for parties 0 and 2 (cf., \eqref{eq:bott-lifts}).}
    \label{fig:min-fg-hp}
\end{figure}

\begin{thm}[Minimum fine-graining]
\label{thm:min-fg}
    For any $\N,\N'$, $\N$-party \emph{KC-PMI} $\pmi$, and \emph{CG}-map $\cgmap$ from $\N'$ to $\N$, the minimum fine-graining of $\pmi$ with respect to $\cgmap$ is the $\N'$-party \emph{KC-PMI} $\pmi'$ given by
    \begin{equation}
    \label{eq:min-fg-pmi}
        \pmi'=\bigcup_{\mi\in\pmi} \downarrow \cginv (\mi).
    \end{equation}
    Its correlation hypergraph $\hpp$ is obtained via a ``blow-up'' of the vertices and hyperedges of $\hp$ as follows. Each vertex $v_l$ is replaced by the complete hypergraph with vertices $v_{\ell'}$ for all $\ell'\in\cginv(\ell)$, resulting in the hyperedges
    \begin{equation}
    \label{eq:min-fg-hyperedges-1}
        \bigcup_{\ell\in\nsp} \{h'\ |\; h'\subseteq V'_{\ell}, \  |\cginv(\ell)|\ge2\}.
    \end{equation}
    Each hyperedge $h_\X$ is replaced by the collection of all possible hyperedges which are subsets of $V'_\X$ and contain at least one vertex $v_{\ell'}$ with $\ell'\in\cginv(\ell)$ for each $v_\ell\in h_\X$, resulting in the hyperedges
    \begin{equation}
    \label{eq:min-fg-hyperedges-2}
       \bigcup_{h_\X\,\text{of}\;\hp} \!\{h'\subseteq V'_{\X} |\; h'\cap V'_\ell\neq\varnothing,\; \forall\, \ell\in\X \}.
    \end{equation}
\end{thm}
\begin{proof}
    The core of the proof is organized into two parts, indicated by (i) and (ii) below. In (i), we show that the \textit{down-set} $\ds'$ specified on the r.h.s. of \eqref{eq:min-fg-pmi}---which we have to prove to be a PMI---is the complement of the up-set obtained via the reconstruction formula \eqref{eq:reconstruction} from the hypergraph $\gf{H}'$ defined by the set of hyperedges listed in \eqref{eq:min-fg-hyperedges-1} and \eqref{eq:min-fg-hyperedges-2}. Next, in (ii), we show that $\gf{H}'$ is indeed the correlation hypergraph of a KC-PMI $\pmi'$. The statement of the theorem then follows from two additional simple observations. First, the fact that the coarse-graining of $\pmi'$ via $\cgmap$ is indeed $\pmi$, implying that $\pmi'$ is indeed a fine-graining of $\pmi$ with respect to $\cgmap$. This is clear from the fact that for each hyperedge $h_\X$ of $\hp$, all hyperedges $h'$ in the set in \eqref{eq:min-fg-hyperedges-2} (for fixed $\X$) are mapped to $h_\X$ by the quotient, while the hyperedges in \eqref{eq:min-fg-hyperedges-1} do not contribute new hyperedges of $\hp$ (they are mapped to single vertices). Second, the fact that any other fine-graining of $\pmi$ (with respect to $\cgmap$) must contain $\pmi'$ (cf., \eqref{eq:pmi-cg}), and  $\pmi'$ is therefore a minimal element of $\cgpmiinv(\pmi)$ by \Cref{thm:partial-order-h}, and the minimum by the uniqueness of minimal elements (cf., \eqref{eq:minimum}).

    (i) We need to show that the positive $\bsets$ associated to $\ds'$ precisely correspond to the hyperedges of $\gf{H}'$. We begin by organizing the $\bsets$ in the $\N'$-party system into families labeled by the subsystems $\X$ of the $\N$-party system. Specifically, for an arbitrary $\varnothing\neq\X\subseteq\nsp$, we denote by $\mathfrak{B}'_\X$ the following family of $\bsets$
    \begin{equation}
    \label{eq:bigB-def}
        \mathfrak{B}'_\X = \{\bs{\Y'}|\, \Y'\subseteq \cginv(\X)\; \text{and}\; \Y'\cap \cginv(\ell)\neq\varnothing\; \forall \ell\in\X\}.
    \end{equation}
    Note that the collection of all $\mathfrak{B}'_\X$, for all $\X$ such that $\varnothing\neq\X\subseteq\nsp$, 
   is indeed a partition of the set of all $\bsets$ of the $\N'$-party system. 
   
   We now show the following implications:
    \begin{align}
        & |\X|\geq 2\; \text{and}\; \bs{\X}\in\pos \implies \bs{\Y'}\in\pos\quad \forall\, \bs{\Y'}\in \mathfrak{B}'_\X, \label{eq:min1}\\
        & |\X|\geq 2\; \text{and}\; \bs{\X}\notin\pos \implies \bs{\Y'}\notin\pos\quad \forall\, \bs{\Y'}\in \mathfrak{B}'_\X, \label{eq:min2}\\
        & |\X|=1\; \implies \bs{\Y'}\in\pos\quad \forall\, \bs{\Y'}\in \mathfrak{B}'_\X.\label{eq:min3}
    \end{align}
   The desired result then follows from observing that the positive $\bsets$ $\bs{\Y'}$ resulting from \eqref{eq:min1} correspond to the hyperedges in \eqref{eq:min-fg-hyperedges-2}, while the $\bsets$ in \eqref{eq:min3} correspond to the ones in \eqref{eq:min-fg-hyperedges-1}.

   The key to show the implications above is that an arbitrary $\N'$-party MI instance $\mi'_1(\W'_1:\Z'_1)$ is in $\ds'$ if and only if the MI instance $\mi'_2(\W'_2:\Z'_2)\succeq\mi'_1(\W'_1:\Z'_1)$, where 
    \begin{equation}
    \label{eq:subs-i2}
        \W'_2 = \bigcup\, \{\eqcl{\ell'}|\, \ell'\in\W'_1\}   \qquad   \Z'_2 = \bigcup\, \{\eqcl{\ell'}|\, \ell'\in\Z'_1\},
    \end{equation}
   is also in $\ds'$. To see this, first notice that, by construction (cf., \eqref{eq:min-fg-pmi}), $\mi'_1$ is in $\ds'$ if and only if there is an MI instance $\mi'_3(\W'_3:\Z'_3)$ such that $\mi'_3\succeq\mi'_1$, where $\mi'_3$ takes the form $\mi'_3=\cginv(\mi_3)$ for some $\N$-party MI instance $\mi_3$, and $\mi_3\in\pmi$. Moreover, any $\mi'_3$ with these properties necessarily satisfies $\mi'_3\succeq\mi'_2$. Therefore, since $\ds'$ is a down-set, and by construction there is $\mi_2$ such that $\mi'_2=\cginv(\mi'_2)$, it follows that $\mi'_1$ is in $\ds'$ if and only if $\mi_2$ is in $\pmi$. 
   
   This observation immediately implies \eqref{eq:min1} and \eqref{eq:min3}. First notice that for an arbitrary $\Y'$ with $|\Y'|\geq 2$, and arbitrary $\mi'_1\in\bs{\Y'}$, the corresponding $\mi'_2$, as defined in \eqref{eq:subs-i2}, belongs to $\bs{\X'}$, where $\X'=\cginv(\X)$ and $\X$ is label of the set $\mathfrak{B}'_\X$ that contains $\bs{\Y'}$. Therefore, given an arbitrary subsystem $\X$, if $|\X|\geq 2$ and $\bs{\X}$ is positive, by construction $\bs{\cginv(\X)}$ is positive, and for any $\mi'_1\in\bs{\Y'}$, with $\bs{\Y'}\in \mathfrak{B}'_\X$, there is no $\mi'_2$ which belongs to $\ds'$. The same conclusion also holds if $|\X|=1$, since in this case $\mi'_2$ does not exist for any choice of $\mi'_1$.

   Lastly, to show \eqref{eq:min2}, consider a subsystem $\X$ with $|\X|\geq 2$ such that $\bs{\X}$ is not positive. It is sufficient to show that any $\bset$ in $\mathfrak{B}'_\X$ contains an MI instance $\mi'_1\preceq\cginv(\mi_2)$, for some $\mi_2\in\bs{\X}$, which easily follows from the condition that specifies which $\bs{\Y'}$ belong to $\mathfrak{B}'_\X$. Indeed, given an MI instance $\mi_2(\W_2:\Z_2)\in\bs{\X}$ in $\pmi$, consider $\mi'_2(\W'_2:\Z'_2)$, where $\W_2'=\cginv(\W_2)$ and $\Z'_2=\cginv(\Z_2)$, which by construction is in $\ds'$. For any $\Y'$ such that $\Y'\subseteq \cginv(\X)$, the condition $\Y'\cap \cginv(\ell)\neq\varnothing$ for all $\ell\in\X$ (cf., \eqref{eq:bigB-def}) implies that $\W'_1=\W'_2\cap\Y'\neq\varnothing$ and similarly for $\Z'_1$. Therefore $\bs{\Y'}$ contains the MI instance $\mi'_1(\W'_1:\Z'_1)$, and by construction, $\mi'_1\preceq\mi'_2$.
       
   (ii) To show that $\gf{H}'$ is indeed the correlation hypergraph of a KC-PMI, we will show that it can be obtained from the meet of a collection of $\N'$-party hypergraphs which are indeed correlation hypergraphs of $\N'$-party KC-PMIs. Consider a party $\ell\in\nsp$, and let $\gf{H}^{(\ell)}_{\circledast}$ be the complete hypergraph with vertices $v_{\ell'}$ for all $\ell'\in\cginv(\ell)$ (cf., \eqref{eq:min-fg-hyperedges-1}). The hypergraph 
    \begin{equation}
    \label{eq:full-hyp-sum}
        \gf{H}_{\circledast}^{\text{all}}=\bigoplus_{\ell\in\nsp}\, \gf{H}^{(\ell)}_{\circledast}
    \end{equation}
   then is a hypergraph with $\N'+1$ vertices and it is a correlation hypergraph, since each $\gf{H}^{(\ell)}_{\circledast}$ is the correlation hypergraph of the bottom of the KC-lattice for $|\cginv(\ell)|-1$ parties,\footnote{\,Up to a trivial relabeling of the parties from $[\![|\cginv(\ell)|-1]\!]$ to the elements of $\cginv(\ell)$.} and the sum is well defined because each $\gf{H}^{(\ell)}_{\circledast}$ corresponds to a different set of parties.\footnote{\,The KC-PMI with hypergraph given by \eqref{eq:full-hyp-sum} is realized by the tensor product of density matrices $\rho_{\cginv(\ell)}$ with the property that none of their marginals factorizes.} Starting from $\hp$ now, we denote by $\hp^{(i)}$ the correlation hypergraph of an arbitrary canonical fine-graining, where $i\in[\phi(\cgmap)]$ (cf., \eqref{eq:n-can-fg}). By \Cref{thm:meet-from-cg}, the correlation hypergraph of the meet of the KC-PMIs with correlation hypergraphs $\gf{H}_{\circledast}^{\text{all}}$ and $\hp^{(i)}$ for all $i\in[\phi(\cgmap)]$ has all hyperedges in \eqref{eq:min-fg-hyperedges-1} (from $\gf{H}_{\circledast}$). Moreover, the hyperedges of the hypergraphs $\hp^{(i)}$ corresponds to the hyperedges in \eqref{eq:min-fg-hyperedges-2} with $|V'_\ell\cap h'|=1$, while the other hyperedges in \eqref{eq:min-fg-hyperedges-2} are added by the weak union closure with the hyperedges of $\gf{H}_{\circledast}^{\text{all}}$; we leave further details for the reader. 
\end{proof}

The proof of \Cref{thm:min-fg} also implies the following result describing the minimum fine-graining in terms of the meet in the $\N'$-party KC-lattice.

\begin{cor}[Minimum fine-graining]
\label{cor:min-fg}
    For any $\N,\N'$, $\N$-party \emph{KC-PMI} $\pmi$, and \emph{CG}-map $\cgmap$ from $\N'$ to $\N$, the minimum fine-graining of $\pmi$ with respect to $\cgmap$ is the meet, in the $\N'$-party \emph{KC}-lattice, of all canonical fine-grainings of $\pmi$, and all $\N'$-party \emph{KC-PMIs} with correlation hypergraphs
    \begin{equation}
    \label{eq:bott-lifts}
        \gf{H}^{(\ell)}_{\circledast} \oplus v_{\ell_1} \oplus \cdots \oplus v_{\ell_k} \quad \forall\,\ell\in \nsp,
    \end{equation}
    where $\gf{H}^{(\ell)}_{\circledast}$ is defined in the proof of \Cref{thm:min-fg}, and the isolated vertices $v_{\ell_1}, \ldots,v_{\ell_k}$ correspond, for each $\ell$, to the parties in $\nspp\setminus\cginv(\ell)$.
\end{cor}
\begin{proof}
    See the proof of \Cref{thm:min-fg} and \Cref{thm:meet-from-cg}.
\end{proof}
In \Cref{cor:min-fg}, notice that each correlation hypergraph in \eqref{eq:bott-lifts} can also be seen as a canonical fine-graining of the KC-PMI at the bottom of the $\N_\ell$-party KC-lattice, where $\N_\ell=|\cginv(\ell)|-1$, and the $\N'+1$ parties have been permuted such that the vertices in the non-trivial component of this hypergraph are precisely the ones in $\cginv(\ell)$.

We conclude the analysis of fine-grainings of KC-PMIs by showing that any canonical fine-graining from \Cref{defi:canonical-fg} is a \textit{maximal} fine-graining, i.e., it is a maximal element of the poset $\cgpmiinv(\pmi)$.

\begin{thm}
    The canonical fine-graining is a maximal fine-graining.
\end{thm}
\begin{proof}
    For given $\N,\N'$, $\N$-party KC-PMI $\pmi_1$, and CG-map $\cgmap$, let $\pmi'_1$ be an arbitrary canonical fine-graining of $\pmi_1$ with respect to $\cgmap$. Consider an arbitrary $\N'$-party KC-PMI $\pmi'_2$ that \textit{covers} $\pmi'_1$, i.e., such that there is no other element of $\lat{KC}^{\N'}$ in the interval $[\pmi'_1,\pmi'_2]$ distinct from $\pmi'_1$ and $\pmi'_2$. We write the correlation hypergraph of $\pmi'_1$ as
    \begin{equation}
    \label{eq:can-fine-grain-1}
        \gf{H}_{\pmi'_1} = \gf{H}_1 \oplus v_{\ell_1} \oplus \cdots \oplus v_{\ell_k},
    \end{equation}
    where $\gf{H}_1$ is a hypergraph obtained from $\gf{H}_{\pmi_1}$ by simply relabeling the parties according to $\cgmap$ (and is therefore isomorphic to $\hp$), and other terms are isolated vertices for the additional $\N'-\N$ parties. Since $\pmi'_2\succ\pmi'_1$, by \Cref{thm:partial-order-h} we have $E'_2\subset E'_1$, where $E'_2$ and $E'_1$ denote (respectively) the sets of hyperedges of $\gf{H}_{\pmi'_2}$ and $\gf{H}_{\pmi'_1}$. Because of \eqref{eq:can-fine-grain-1}, we can then write $\gf{H}_{\pmi'_2}$ as
    \begin{equation}
    \label{eq:can-fine-grain-2}
        \gf{H}_{\pmi'_2} = \gf{H}_2 \oplus v_{\ell_1} \oplus \cdots \oplus v_{\ell_k},
    \end{equation}
    and the set of hyperedges of $\gf{H}_2$ is a subset of the set of hyperedges of $\gf{H}_1$. Therefore, the coarse-graining of $\gf{H}_{\pmi'_2}$ by $\cgmap$ is the correlation hypergraph of an $\N$-party KC-PMI $\pmi_2$ such that $\pmi_2\succ\pmi_1$ (again by \Cref{thm:partial-order-h}). This means that $\pmi'_2$ is not a fine-graining of $\pmi_1$ with respect to $\cgmap$, and since it covers $\pmi'_1$, the canonical fine-graining $\pmi'_1$ is maximal.
\end{proof}

Fine-grainings will play a key role in \S\ref{sec:strategies} (particularly, \S\ref{subsec:non-chordal-case}), where we discuss the construction of graph models, and it is convenient to clarify a few aspects about the pre-image of a specific entropy vector under a CG-map. We have seen in the previous subsection that given a CG-map $\cgmap$, the transformation $\cgsvec$ of an entropy vector is a linear map which simply ``projects out'' certain components (cf., \eqref{eq:ent-vec-cg}). Given an $\N$-party entropy vector $\Svec$ and a CG-map $\cgmap$ then, we could define the set of its fine-grainings of $\Svec$ with respect to the CG-map $\cgmap$ as the affine space $\mathbb{A}_\cgmap(\Svec)$ of $\N'$-party entropy space specified by the set of equations \eqref{eq:ent-vec-cg}, for all subsystems $\X$. 

Such a definition, however, despite its simplicity, is clearly too general, as it even includes entropy vectors with negative components. A more natural approach, considering our focus on PMIs, is to impose subadditivity and define the fine-grainings of $\Svec$ with respect to the CG-map $\cgmap$ as the intersection of $\mathbb{A}_\cgmap(\Svec)$ with the $\N'$-party SAC 
\begin{equation}
    \mathscr{A}_\cgmap(\Svec) = \mathbb{A}_\cgmap(\Svec)\,  \cap\, \text{SAC}_{\N'}\,\,.
\end{equation}
Geometrically, this set is an unbounded polyhedron, as some components of an element $\Svecp$ can take arbitrarily large values. However, it is not a cone, since it is contained within an affine subspace. Furthermore, note that, by construction, this set has a natural partition according to the faces of the SAC$_{\N'}$, although it never contains an entire face.\footnote{\,Except for the trivial case where $\Svec$ is the null vector.}

Finally, notice that, in general, this definition implies that for a given $\Svec$ and $\cgmap$, the PMI of an entropy vector $\Svecp$ in $\mathscr{A}_\cgmap(\Svec)$ is not necessarily a KC-PMI. The ``KC fine-grainings'' of $\Svec$, i.e., the entropy vectors $\Svecp$ whose PMI \textit{is} a KC-PMI, are the elements of $\mathscr{A}_\cgmap(\Svec)$ contained in a ``KC-face'' of the SAC$_{\N'}$. However, while KC is a useful condition for the analysis of PMIs, when working with entropy vectors, it is typically more natural to directly impose strong subadditivity. Accordingly, we define the SSA-compatible fine-grainings of $\Svecp$ with respect to the CG-map $\cgmap$ as
\begin{equation}
    \mathscr{A}^{\Sigma}_\cgmap(\Svec) = \mathbb{A}_\cgmap(\Svec)\,  \cap\, \Sigma_{\N'}\,\,,
\end{equation}
where $\Sigma_{\N'}$ is the polyhedral cone specified by all instances of SA and SSA for $\N'$ parties.\footnote{\,Similarly to the definition of the SAC, the set of instances of SSA also includes all instances of weak monotonicity, which can simply be derived from SSA using purifications.}

\subsection{Review of holographic graph models}
\label{subsec:graph-review}

To make this work self contained, in this subsection we briefly review the basic definitions about holographic graph models. For further details, we refer the reader to the original literature, starting from \cite{Bao:2015bfa}, where these models were first introduced. 

A \textit{holographic graph model}, or simply \textit{graph model}, is a weighted graph $\gf{G}=(V,E)$, with non-negative edge weights, and a labeling of a subset of vertices $\partial V\subseteq V$ called \textit{boundary vertices}. For a fixed number of parties $\N$, the subset $\partial V$ in an $\N$-party graph model needs to satisfy $|\partial V|\geq \N+1$, each $v\in\partial V$ must be labeled by \textit{exactly one} element of $\nsp$, and each $\ell\in\nsp$ must label \textit{at least one} boundary vertex. The vertices in $V\setminus\partial V$ are referred to as \textit{bulk vertices}.

An arbitrary subset $C\subseteq V$ is called a \textit{cut}, and an edge $e\in E$ is a \textit{cut edge} if and only if one of its endpoints is in $C$ and the other is in $\comp{C}$. Any cut $C$ has a \textit{cost}, which is defined as the sum of the weights of all cut edges. A cut $C$ is said to be a $\X$-\textit{cut} (denoted by $C_\X$) if the labels of the vertices in $C\cap\partial V$ are precisely the parties in $\X$. For a fixed subsystem $\X$, a cut with minimum cost among all $\X$-cuts is called a \textit{min-cut for $\X$}. In general, min-cuts are not unique, but it was shown in \cite{Avis:2021xnz} that if one introduces the partial order among the min-cuts given by inclusion, there is a unique \textit{minimal min-cut} (denoted by $\mmC_{\X}$). Given $\X$, the ``entropy'' of $\X$ is defined as the cost of a min-cut for $\X$, and via this prescription, one obtains an entropy vector for any given graph model.\footnote{\,It was shown in \cite{hayden2016holographic} that any entropy vector obtained in this way is realizable by a density matrix. In fact, stabilizer states are sufficient.} 

An $\N$-party graph model is said to be \textit{simple} if it has precisely $\N+1$ boundary vertices (so that each party $\ell$ corresponds to precisely one vertex $v_\ell\in \partial V$). Given any graph model, it is always possible to obtain a simple one realizing the same entropy vector by identifying boundary vertices labeled by the same party. Simple graph models with tree topology, or \textit{simple trees} \cite{Hernandez-Cuenca:2022pst}, will be at the core of the next section. Notice that the tree topology is typically not preserved under the transformation of a graph model to its simple version, since the identification of vertices often introduces cycles in the graph.

Given two graph models, their \textit{sum} is defined as the graph whose vertex (respectively, edge) set is the disjoint union of the vertex (respectively, edge) sets of the two graphs. Notice that this definition is very similar to the one we introduced above for correlation hypergraphs, but with one key difference. By definition, the correlation hypergraph has precisely one vertex for each party, and the sum of two correlation hypergraphs is not well defined if they have any party in common. Nevertheless, it is easy to see what the correlation hypergraph of the sum of two graph models is. Specifically, since the entropy vector of the sum of two graph models is the sum of the corresponding entropy vectors,\footnote{\,For simplicity, here we are assuming that the two graph models in the sum correspond to the same set of parties.} the correlation hypergraph of the sum is the correlation hypergraph of the meet of the KC-PMIs of these vectors (cf., \eqref{eq:pmi-meet-svec-sum}).

Finally, below we will often consider coarse-grainings and fine-grainings of graph models \cite{Hernandez-Cuenca:2022pst}. Given an arbitrary $\N'$-party graph model $\gf{G}_{\N'}$, and an arbitrary coarse-graining map $\cgmap$, it is straightforward to verify that the entropy vector after the coarse-graining is realized by the graph model $\gf{G}_{\N}$ obtained from $\gf{G}_{\N'}$ by simply relabeling each boundary vertex $v_{\ell'}$ by $\cgmap(\ell')$. Conversely, if an $\N$-party graph model $\gf{G}_{\N}$ is not simple, we can straightforwardly obtain a simple graph $\gf{G}_{\N'}$ by a fine-graining map which relabels boundary vertices labeled by the same party with distinct parties (and possibly permuting the parties). Finally, in analogy to correlation hypergraphs, the canonical fine-grainings of a given graph model are obtained by adding isolated boundary vertices for the new parties, and by permuting all parties as necessary.

\section{Strategies for graph constructions}
\label{sec:strategies}

In this section, we initiate the program of developing systematic strategies and techniques to construct, whenever possible, a holographic graph model realization of a given entropy vector. Importantly, we emphasize that we will not assume knowledge of any holographic entropy inequality, requiring only that the entropy vector satisfy subadditivity and strong subadditivity. The motivation for this approach is twofold. On the one hand, while complete knowledge of holographic entropy inequalities would allow one to determine conclusively whether an entropy vector is realizable, it does not provide a constructive method for obtaining a corresponding graph model. On the other hand, we hope that this program will lead to a deeper understanding and interpretation of the inequalities themselves. Accordingly, our goals are to develop methods for constructing holographic graph model realizations and to detect unrealizability independently from the inequalities. Although we do not fully solve either problem here, this section presents first steps in both directions. 

We begin in \S\ref{subsec:preliminary-steps} by outlining a series of preliminary steps designed primarily to simplify the problem and isolate the most genuinely interesting cases. This part concludes with a ``chordality test''—a yes-or-no criterion that determines whether the entropy vector could be realized by a simple tree graph model. If the entropy vector passes this test, \S\ref{subsec:chordal-case} introduces an efficient algorithm to construct a candidate simple tree graph model realizing $\Svec$. Although we will not prove it here, we believe that this algorithm will always succeed, and we plan to report on this in future work \cite{sufficiency}. If, on the other hand, the chordality test fails, the construction becomes more involved, and we do not present a complete algorithm in this paper. Nevertheless, \S\ref{subsec:non-chordal-case} explains how the procedure developed in \S\ref{subsec:chordal-case} for the chordal case could still be applied by first fine-graining the system to a larger number of parties, $\N' > \N$. We will also comment on how one might use this approach to establish the unrealizability of $\Svec$ by arbitrary (not necessarily simple) tree graph models. We return to the broader question of unrealizability by \textit{any} graph model in the discussion \S\ref{sec:discussion}.

\subsection{Preliminary analysis and reduction}
\label{subsec:preliminary-steps}

For some fixed $\N$, consider an arbitrary entropy vector $\Svec$ that obeys all instances of subadditivity and strong subadditivity, and let $\pmi$ be its PMI. We emphasize that we do not make any other assumption about $\Svec$. As reviewed in \S\ref{subsec:review}, Klein's condition is strictly weaker than strong subadditivity, therefore $\pmi$ is a KC-PMI, and we can construct its correlation hypergraph $\hp$. For simplicity, in what follows we will often call $\hp$ the correlation hypergraph ``of'' $\Svec$.

Notice that in general $\Svec$ might have some vanishing components. The first paragraph below will focus on this possibility, and clarify how in this case the problem of finding a graph model for $\Svec$ can be reduced to the analogous problem for entropy vectors corresponding to systems with fewer parties.

A central role in this subsection will be played by the graph $\lhp$, called the \textit{line graph} of $\hp$, which is defined as the \textit{intersection graph}\footnote{\,Given an arbitrary set $\A$ and a collection of subsets $\{\B_i\}_{i=1}^k$, the intersection graph of the collection is a graph with a vertex for each subset $\B_i$, and an edge between two vertices if and only if the corresponding subsets have non-empty intersection.} of its set of hyperedges. We will denote a \textit{max-clique}\footnote{\,In a graph $\gf{G}=(V,E)$, a clique is a complete subgraph induced by a subset of $V$. A max-clique is a clique which is not contained in any other clique.} of $\lhp$ by $Q$, and for an arbitrary max-clique we denote by $Q^\cap$ the intersection of the hyperedges of $\hp$ which are elements of $Q$, i.e.,
\begin{equation}
Q^\cap = \bigcap_{h_\X\in Q} h_\X.
\end{equation}
It was shown in \cite{Hubeny:2024fjn} that for any correlation hypergraph, unlike for arbitrary hypergraphs, the cardinality of $Q^\cap$ satisfies the very stringent bound $|Q^\cap|\leq 2$, and that the values 1 and 2 can only be attained when $\hp$ has a particular structure. In the second paragraph we will explain how, if $|Q^\cap|> 0$, we can reduce the problem of finding a graph model for $\Svec$ to the analogous problem for a collection of entropy vectors such that $|Q^\cap|=0$ for each element of the collection.

We stress that this reduction, as well as the one mentioned above for entropy vectors with vanishing components, are completely general and are \textit{not} specific to any particular subclass of graph models, in particular, they are not specific to trees. On the other hand, simple tree graph models are the focus of the third paragraph, where we introduce the ``chordality'' test from \cite{Hubeny:2024fjn}.

\paragraph{Reduction to strictly positive entries:} Given an entropy vector $\Svec$ that obeys all instances of SA and SSA, suppose that it has at least one vanishing entry $\ent_\I$. Since
\begin{equation}
\label{eq:top-mi-poset}
    \mi(\I:\I^\complement)=2\,\ent_\I \quad \text{with} \quad \I^\complement\coloneq\nsp\setminus\I,
\end{equation}
the $\bset$ $\bs{\nsp}$ is not positive, the correlation hypergraph is disconnected, and its connected components correspond to the elements of the partition $\Gamma(\nsp)$. The following result shows how in this case the problem of finding a graph model for $\Svec$ can be reduced to that of finding graph models for a collection of entropy vectors for fewer parties, none of which has a vanishing entry, corresponding to the connected components of $\hp$.

To state the theorem, we introduce a collection of CG-maps, one for each connected component of $\hp$; each map traces out all parties in every other component via a coarse-graining of the purifier. Specifically, for any element $\X_i$ of $\Gamma(\nsp)$, consider the CG-map $\cgmap_i$ specified by the partition
\begin{equation}
    \cgpart_i=\{\X_i,\nsp\setminus\X_i\},
\end{equation}
and indexing function
\begin{equation}
    \chi_i(\ell) = 
    \begin{cases}
        \lambda_i(\ell) & \ell\in\X_i\\
        0 & \ell\notin \X_i,
    \end{cases}
\end{equation}
where $\lambda_i$ is an arbitrary choice of bijection between $\X_i$ and $[\![\N_i]\!]$, and $\N_i=|\X_i|-1$. We then obtain the following reduction theorem.

\begin{figure}[tbp]
    \centering
     \begin{subfigure}{0.45\textwidth}
    \centering
    \begin{tikzpicture}
    \draw[PminusEcol, rounded corners, very thick] (-1.5,-0.6) rectangle (1.5,2.4);
    \draw[Mcol!80!black, very thick] (-1,0) -- (0,1.5);
    \draw[Mcol!80!black, very thick] (1,0) -- (0,1.5);
    \draw[Mcol!80!black, very thick] (-1,0) -- (0
    1,0);
    \draw[Mcol!80!black, very thick] (2.5,0) -- (0
    2.5,1.5);

    \node[] () at (-1,2.2) {{\scriptsize $h_{124}$}};  
    
    \filldraw (-1,0) circle (2pt);
    \filldraw (1,0) circle (2pt);
    \filldraw (0,1.5) circle (2pt);
    \filldraw (2.5,0) circle (2pt);
    \filldraw (2.5,1.5) circle (2pt);
    
    \node[] () at (-1.2,-0.3) {{\scriptsize $v_1$}};
    \node[] () at (1.2,-0.3) {{\scriptsize $v_2$}};
    \node[] () at (0,1.8) {{\scriptsize $v_4$}};
    \node[] () at (2.5,-0.3) {{\scriptsize $v_3$}};
    \node[] () at (2.5,1.8) {{\scriptsize $v_0$}};
    \end{tikzpicture}
    \subcaption[]{}
    \end{subfigure}
    \begin{subfigure}{0.45\textwidth}
    \centering
    \begin{tikzpicture}
    \draw (-1,0) -- (0,0.5);
    \draw (1,0) -- (0,0.5);
    \draw (0,1.5) -- (0,0.5);
    \draw (2.5,0) -- (2.5,1.5);
    
    \filldraw (-1,0) circle (2pt);
    \filldraw (1,0) circle (2pt);
    \filldraw (0,1.5) circle (2pt);
    \filldraw (2.5,0) circle (2pt);
    \filldraw (2.5,1.5) circle (2pt);
    
    \filldraw[fill=bulkcol] (0,0.5) circle (2pt);
    
    \node[] () at (-1.2,-0.3) {{\scriptsize $v_1$}};
    \node[] () at (1.2,-0.3) {{\scriptsize $v_2$}};
    \node[] () at (0,1.8) {{\scriptsize $v_4$}};
    \node[] () at (2.5,-0.3) {{\scriptsize $v_3$}};
    \node[] () at (2.5,1.8) {{\scriptsize $v_0$}};
    
    \node[red] () at (-0.15,1) {{\scriptsize $3$}};
    \node[red] () at (-0.5,0.05) {{\scriptsize $2$}};
    \node[red] () at (0.5,0.05) {{\scriptsize $2$}};
    \node[red] () at (2.35,0.75) {{\scriptsize $4$}};

    \end{tikzpicture}
    \subcaption[]{}
    \end{subfigure}
    \caption{An example of the reduction described in the first paragraph of the main text, for an $\N=3$ entropy vector with a vanishing entry, $\Svec=(2,2,4,3,3,6,2,6,2,7,7,0,6,6,4)$. The correlation hypergraph $\hp$ for the KC-PMI of $\Svec$ is shown in (a). Choosing the CG-map $\cgmap_L$ for the left component of $\hp$ specified by $\{\{1\}_1,\{2\}_2,\{0,3,4\}_0\}$, and the CG-map $\cgmap_R$ for the right component specified by $\{\{0,1,2,4\}_0,\{3\}_1\}$, we obtain the entropy vectors $\Svec^{(L)}=(2,2,3)$, and $\Svec^{(R)}=(4)$. The graph model realization of $\Svec$, shown in (b), is obtained from the sum of a choice of canonical fine-grainings of graphs realizing $\Svec^{(L)}$ and $\Svec^{(R)}$. In (b) we have omitted the (irrelevant) isolated vertices which would result from this sum.}
    \label{fig:van-comp-red-v2}
\end{figure}
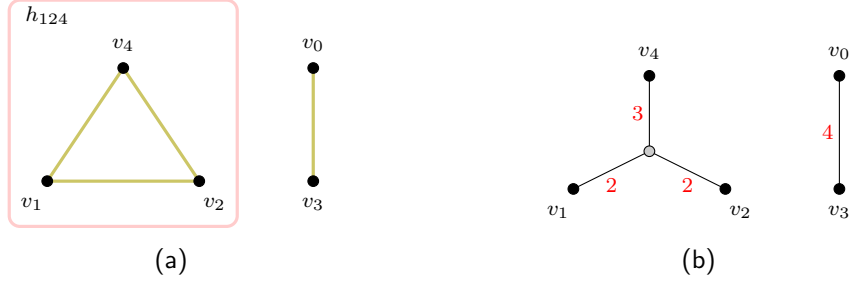

\begin{thm}
\label{thm:reduction-disconnected}
    An entropy vector $\Svec$ that obeys all instances of \emph{SA} and \emph{SSA}, and has at least one vanishing entry, is realizable by a holographic graph model if and only if all vectors in the collection
\begin{equation}
\label{eq:red-connected-collection}
    \{\Svec^{(i)} = \Lambda_{\Theta_i}^\bullet (\Svec)\,|\,\X_i\in\Gamma(\nsp)\},
\end{equation}
none of which has a vanishing entry, is also realizable. Furthermore, denoting by $\gf{G}^{(i)}$ a graph model realization of $\Svec^{(i)}$, a graph model for $\Svec$ is given by
\begin{equation}
\label{eq:red-connected-realization}
    \gf{G} = \!\bigoplus_{\X_i\in\Gamma(\nsp)} \!{}^\uparrow\gf{G}^{(i)},
\end{equation}
where ${}^\uparrow\gf{G}^{(i)}$ is an arbitrary canonical fine-graining of $\gf{G}^{(i)}$ with respect to $\cgmap_i$. 
\end{thm}
\begin{proof}

We first show that $\Svec$ is realizable if and only if all entropy vectors in \eqref{eq:red-connected-collection} are realizable. One direction is trivial; since any vector $\Svec^{(i)}$ in \eqref{eq:red-connected-collection} is obtained via a CG-map from $\Svec$, if $\Svec$ is realizable by a holographic graph model $\gf{G}$, then the same CG-map applied to $\gf{G}$ gives a graph model realization $\gf{G}^{(i)}$ of $\Svec^{(i)}$. Conversely, suppose that each $\Svec^{(i)}$ in \eqref{eq:red-connected-collection} is realizable by a graph model $\gf{G}^{(i)}$. We will prove that $\Svec$ is realizable by showing that the entropy vector of $\gf{G}$ in \eqref{eq:red-connected-realization} is precisely $\Svec$. 

We first construct the graph models ${}^\uparrow\gf{G}^{(i)}$ by choosing, for each $\gf{G}^{(i)}$, an arbitrary canonical fine-graining with respect to the CG-map $\cgmap_i$. The entropy vectors corresponding to these graphs are denoted by ${}^\uparrow\Svec^{(i)}$. Consider then an arbitrary non-vanishing entry $\ent_\J$ of $\Svec$. Since the CG-maps $\cgmap_i$ are defined starting from a partition of $\nsp$, in the collection of vectors ${}^\uparrow\Svec^{(i)}$ there is a single one such that ${}^\uparrow\ent^{(i)}_\J$ does not vanish, and we denote it ${}^\uparrow\Svec^{(i(\J))}$. Moreover, by construction, $\ent^{(i(\J))}_{\cgmap_i(\J)}=\ent_\J$, where $\Svec^{(i(\J))}$ is the element of the collection of entropy vectors $\Svec^{(i)}$ which is lifted to ${}^\uparrow\Svec^{(i(\J))}$, and $\ent^{(i(\J))}_{\cgmap_i(\J)}$ is its $\cgmap_i(\J)$-entry. Since the entropy vector of a sum of graph models is the sum of the entropy vectors for the individual graphs, and $\ent^{(j)}_\J=0$ for all $j\neq i(\J)$, it follows that the $\J$-entry of the graph model in \eqref{eq:red-connected-realization} is precisely $\ent_\J$. Repeating the same argument for all entries of $\Svec$, we obtain the desired result.

Finally, notice that none of the $\Svec^{(i)}$ in \eqref{eq:red-connected-collection} has a vanishing entry, since the correlation hypergraph of each $\Svec^{(i)}$ is, by construction, and up to relabeling of the parties, one of the connected components of the correlation hypergraph of $\Svec$.
\end{proof}

\paragraph{Reduction to $|Q^\cap|=0$:} 

It was shown in \cite[Theorem 8]{Hubeny:2024fjn} that $|Q^\cap|=2$ is only possible if the correlation hypergraph is disconnected, and since we are focusing on entropy vectors which have strictly positive components, we can ignore this possibility.\footnote{\,Above we used the fact that if an entropy vector has vanishing entries its correlation hypergraph is disconnected. It is clear from \Cref{thm:gamma-partitions} and \eqref{eq:top-mi-poset} that the conver also holds.} Furthermore, \cite[Theorem 8]{Hubeny:2024fjn} showed that if $|Q^\cap|=1$ for some max-clique $Q$ of $\lhp$, then there is an instance of the mutual information $\mi(\X:\Y)$ that vanishes for $\Svec$, and whose arguments satisfy $\X\cup\Y=\nsp\setminus\ell_*$ for some party $\ell_*$. As we will now explain, this allows us to reduce the problem of finding a graph model for $\Svec$, similarly to what we discussed in the previous paragraph.

Consider the CG-map $\cgmap_\X$ specified by the partition 
\begin{equation}
    \cgpart_{\X}=\{\X\cup\ell_*,\Y\},
\end{equation}
and an indexing function 
\begin{equation}
    \chi_{_{\X}}(\ell)=\begin{cases}
        0 & \ell\in\Y \\
        \lambda_{\X}(\ell) & \ell\notin\Y,
    \end{cases}
\end{equation}
where $\lambda_{\X}$ is an arbitrary choice of bijection between $\X\cup\ell_*$ and $[\![\N_{\X}]\!]$, and $\N_\X=|\X|$. The CG-map $\cgmap_\Y$ is define analogously, swapping $\X$ and $\Y$ in the definitions above. The following result then shows how the problem of finding a graph model for $\Svec$ can be reduced to that of finding graph models for two entropy vectors for fewer parties. 

\begin{thm}
\label{thm:reduction-cliqint}
    An entropy vector $\Svec$ that obeys all instances of \emph{SA} and \emph{SSA}, has no vanishing entries, and for which there is at least one max-clique $Q$ of $\lhp$ such that $|Q^\cap|=1$, is realizable if and only if the entropy vectors
    \begin{equation}
        \Svec^{(\X)} = \Lambda^\bullet_\X (\Svec) \qquad \Svec^{(\Y)} = \Lambda^\bullet_\Y (\Svec)
    \end{equation}
    are realizable. Furthermore, denoting by $\gf{G}^{(X)}$ and $\gf{G}^{(Y)}$ the graph models for $\Svec^{(\X)}$ and $\Svec^{(\Y)}$, a graph model for $\Svec$ is given by
\begin{equation}
\label{eq:red-cliqint-realization}
    \gf{G} = {}^\uparrow\gf{G}^{(X)} \oplus {}^\uparrow\gf{G}^{(Y)},
\end{equation}
where ${}^\uparrow\gf{G}^{(X)}$ and ${}^\uparrow\gf{G}^{(Y)}$ are arbitrary canonical fine-grainings of $\gf{G}^{(X)}$ and $\gf{G}^{(Y)}$ with respect to (respectively) $\cgmap_X$ and $\cgmap_Y$. 
\end{thm}
\begin{proof}
    If $\Svec$ is realizable, $\Svec^{(\X)}$ and $\Svec^{(\Y)}$ are obviously realizable, and we focus on the opposite direction. We need to show that the entropy vector of the graph in \eqref{eq:red-cliqint-realization} is $\Svec$. For simplicity, we assume that $\ell_*$ is the purifier; if it is not, one can simply permute the parties as necessary. Consider an arbitrary component $\ent_\I$ of $\Svec$. Since we have chosen $\ell_*$ to be the purifier, it follows that $\I\subseteq \X\cup\Y=[\N]$. We then partition $\I$ into two subsystems $\I_\X=\I\cap\X$ and $\I_\Y=\I\cap\Y$, and since $\mi(\X:\Y)=0$, we have $\ent_\I = \ent_{\I_\X} + \ent_{\I_\Y}$. The argument is then almost identical to that in the proof of \Cref{thm:reduction-disconnected}, and we leave the details as an exercise for the reader. 
\end{proof}

Lastly, while $\Svec^{(\X)}$ and $\Svec^{(\Y)}$ clearly have no vanishing entries (since, by assumption, $\Svec$ also did not), it is important to note that the theorem above does not imply that, for either of these vectors, all max-cliques of the line graph satisfy $|Q^\cap|=0$. Nevertheless, if $|Q^\cap|=1$ for some max-clique, one can simply iterate the reduction until either $|Q^\cap|=0$ for all max-cliques of the line graph of each entropy vector produced by the reduction, or one obtains $\N=1$ entropy vectors, which are trivial to realize. An example of the latter situation, in which the iteration terminates only at $\N=1$, is an entropy vector realized by a graph model that is a path of boundary vertices (with arbitrary weights).

\begin{figure}[tbp]
    \centering
     \begin{subfigure}{0.45\textwidth}
    \centering
    \begin{tikzpicture}

    \coordinate (V0) at (0.0000, 2.0000);
    \coordinate (V1) at (-1.5637, 1.2470);
    \coordinate (V2) at (-1.9499, -0.4947);
    \coordinate (V3) at (-0.8678, -1.8019);
    \coordinate (V4) at (0.8678, -1.8019);
    \coordinate (V5) at (1.9499, -0.4947);
    \coordinate (V6) at (1.5637, 1.2470);

    \foreach \i in {0,...,6} {
        \foreach \j in {\i,...,6} {
            \ifnum\i=\j\relax\else
                \draw[gray!70, thin] (V\i) -- (V\j);
            \fi
        }
    }

    \filldraw[Mcol!80!black] (V0) circle (3pt);
    \filldraw[Mcol!80!black] (V1) circle (3pt);
    \filldraw[PminusEcol] (V2) circle (3pt);
    \filldraw[PminusEcol] (V3) circle (3pt);
    \filldraw[PminusEcol] (V4) circle (3pt);
    \filldraw[PminusEcol] (V5) circle (3pt);
    \filldraw[Mcol!80!black] (V6) circle (3pt);

    \node[above=2pt] at (V0) {{\scriptsize $h_{12}$}};
    \node[above right=2pt] at (V6) {{\scriptsize $h_{23}$}};
    \node[right=2pt] at (V5) {{\scriptsize $h_{123}$}};
    \node[below=2pt] at (V4) {{\scriptsize $h_{124}$}};
    \node[below=2pt] at (V3) {{\scriptsize $h_{234}$}};
    \node[left=2pt] at (V2) {{\scriptsize $h_{1234}$}};
    \node[above left=2pt] at (V1) {{\scriptsize $h_{24}$}};

    \end{tikzpicture}
    \subcaption[]{}
    \end{subfigure}
    \begin{subfigure}{0.45\textwidth}
    \centering
    \begin{tikzpicture}
    \draw (0,0) -- (-1,-1);
    \draw (0,0) -- (-1,1);
    \draw (0,0) -- (2,0);
    
    \filldraw (-1,-1) circle (2pt);
    \filldraw (-1,1) circle (2pt);
    \filldraw (2,0) circle (2pt);
    \filldraw (1,0) circle (2pt);
    
    \filldraw[fill=bulkcol] (0,0) circle (2pt);
    
    \node[] () at (-1.3,1) {{\scriptsize $v_1$}};
    \node[] () at (-1.3,-1) {{\scriptsize $v_0$}};
    \node[] () at (1,0.25) {{\scriptsize $v_2$}};
    \node[] () at (2,0.25) {{\scriptsize $v_3$}};
    
    \node[red] () at (-0.5,0.8) {{\scriptsize $1$}};
    \node[red] () at (-0.5,-0.75) {{\scriptsize $1$}};
    \node[red] () at (0.5,-0.25) {{\scriptsize $2$}};
    \node[red] () at (1.5,-0.25) {{\scriptsize $1$}};

    \node[] () at (2,-2) {{}};

    \end{tikzpicture}
    \subcaption[]{}
    \end{subfigure}
    \caption{An example of a realizable entropy $\Svec=(1, 3, 1, 2, 2, 2, 1)$ for which there is a max-clique $Q$ of $\lhp$ such that $|Q^\cap|=1$. The line graph $\lhp$ is shown in (a) and a graph model realization in (b). Notice that $\lhp$ is a complete graph, and therefore it has a single max-clique $Q$, which is the graph itself. Furthermore, $Q^\cap=\{2\}$ and, in agreement with \cite[Theorem 8]{Hubeny:2024fjn}, $\mi(01:3)=0$. Defining $\X=01$ and $\Y=3$, the coarse-graining $\cgmap_\X$ specified by $\{\{1\}_1,\{2,3\}_0,\{0\}_2\}$, applied to $\Svec$, gives $\Svec^{(\X)}=(1,1,2)$, while the coarse-graining $\cgmap_\Y$ specified by $\{\{0,1,2\}_0,\{3\}_1\}$ gives $\Svec^{(\Y)}=(1)$. Constructing graph models for $\Svec^{(\X)}$ and $\Svec^{(\Y)}$, choosing canonical fine-grainings with respect to the maps specified above, and taking the sum, we obtain a disconnected but not simple graph model for $\Svec$. The simple version of this graph, obtained by identifying vertices with the same party label, is (b).}
    \label{fig:q-int-one-red-v2}
\end{figure}
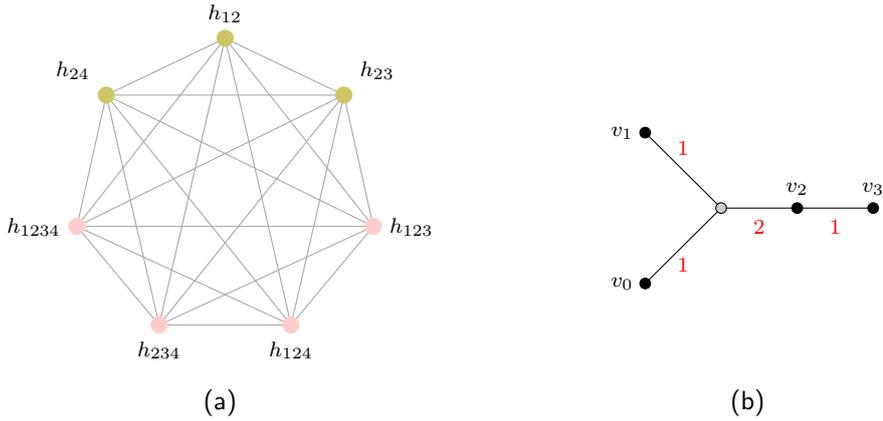

\paragraph{Chordality test:} It was shown in \cite{Hubeny:2024fjn} that an arbitrary entropy vector $\Svec$ which obeys all instances of subadditivity and strong subadditivity can be realized by a simple tree holographic graph model \textit{only if} ``its'' line graph $\lhp$ is a \textit{chordal} graph, i.e., it does not contain any cycle of length four or more without a chord. Given an entropy vector obtained from the reductions described in the previous two paragraphs, we can now proceed to verify whether its line graph is chordal, a check that can be performed efficiently. If $\lhp$ is chordal, we present in \S\ref{subsec:chordal-case} an efficient algorithm that potentially constructs a simple tree graph model realizing the entropy vector. As we will see, the algorithm always produces a \textit{candidate} graph model for $\Svec$; the outstanding question is whether it invariably yields a correct graph model, in the sense that computing its entropy vector using the standard prescription reviewed in \S\ref{subsec:graph-review} reproduces the desired entropy vector. Extensive tests indicate that the algorithm succeeds in all cases examined, and we hope to report a proof of its general validity soon \cite{sufficiency}. Conversely, if the line graph is not chordal, we discuss in \S\ref{subsec:non-chordal-case} a possible strategy for constructing a non-simple tree graph model, and we outline how one might detect unrealizability by arbitrary tree graph models in the absence of any holographic entropy inequality.

\subsection{The chordal case}
\label{subsec:chordal-case}

Let us briefly summarize the assumptions about the entropy vectors for which we want to find a realization by a simple tree graph model. For any given $\N$, in this subsection we consider an entropy vector $\Svec$ such that:
\begin{align}
\label{eq:assumptions}
    \text{(i) }  &\text{it satisfies all instances of SA and SSA},\nonumber\\
    \text{(ii) } &\text{it has no vanishing components, equivalently $\hp$ is connected},\nonumber\\
    \text{(iii) } &\text{$Q^\cap=\varnothing$ for all max-cliques $Q$ of $\lhp$}, \nonumber\\
    \text{(iv) } &\text{$\lhp$ is a chordal graph} 
\end{align}
An entropy vector that satisfies the first three conditions above will be said to be \textit{irreducible}, and for simplicity, we will say that an entropy vector satisfying (i) and (iv) \textit{is chordal}. A simple choice of $\Svec$, that we will use throughout this subsection to exemplify the various steps of the algorithm, is
    \begin{equation}
    \label{eq:Svec-ex-chordal}
        \Svec=(1, 1, 2, 2, 2, 3, 3, 3, 2, 2, 3, 3, 2, 1, 1) \ ,
    \end{equation}
in particular, the complement\footnote{\,For most KC-PMIs the correlation hypergraph and its line graph have a complicated structure which is hard to visualize. To present the relevant data more clearly, we often show either the complement of the line graph (for which most vertices are isolated), or even just the non-trivial subgraph omitting isolated vertices. In some cases, to make the structure of chordless cycles more evident, we instead present the subgraph of the line graph induced by the vertices which belong to chordless cycles.} of $\lhp$ is shown in \Cref{fig:chordal-eg-part-1}.\footnote{\,In the first version of this work, we considered a different entropy vector. While the construction was correct, it was not an example of an irreducible entropy vector. Accordingly, the result of the algorithm was a simple tree graph model which realized the entropy vector, but it had a contractible edge (specifically, the edge incident on the leaf $v_1$).}

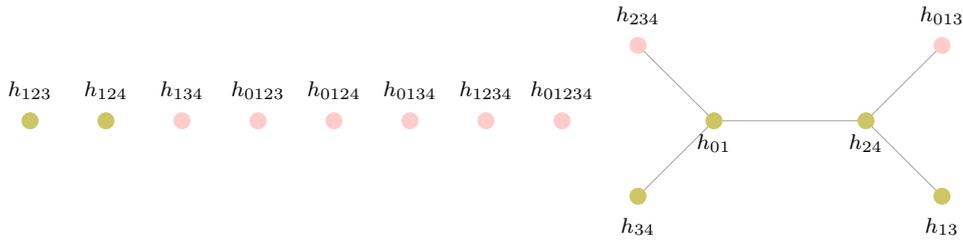
\begin{figure}[tbp]
    \centering
    \begin{tikzpicture}

    \draw[gray!70, thin] (0,0) -- (-1,-1);
    \draw[gray!70, thin] (0,0) -- (-1,1);
    \draw[gray!70, thin] (0,0) -- (1.5,0);
    \draw[gray!70, thin] (1.5,0) -- (3,0);

    \filldraw[Mcol!80!black] (0,0) circle (3pt);
    \filldraw[Mcol!80!black] (-1,-1) circle (3pt);
    \filldraw[PminusEcol] (-1,1) circle (3pt);
    \filldraw[Mcol!80!black] (1.5,0) circle (3pt);
    \filldraw[EminusMcol!90!black] (3,0) circle (3pt);

    \filldraw[Mcol!80!black] (-11,0) circle (3pt);
    \filldraw[EminusMcol!90!black] (-10,0) circle (3pt);
    \filldraw[Mcol!80!black] (-9,0) circle (3pt);
    \filldraw[EminusMcol!90!black] (-8,0) circle (3pt);
    \filldraw[PminusEcol] (-3,0) circle (3pt);
    \filldraw[EminusMcol!90!black] (-6,0) circle (3pt);
    \filldraw[PminusEcol] (-7,0) circle (3pt);
    \filldraw[PminusEcol] (-5,0) circle (3pt);
    \filldraw[PminusEcol] (-4,0) circle (3pt);
    \filldraw[PminusEcol] (-2,0) circle (3pt);

    \node[] () at (-1,1.4) {{\scriptsize $h_{234}$}};
    \node[] () at (-1,-1.4) {{\scriptsize $h_{34}$}};
    \node[] () at (3,-0.3) {{\scriptsize $h_{013}$}};
    \node[] () at (0,-0.3) {{\scriptsize $h_{01}$}};
    \node[] () at (1.5,-0.3) {{\scriptsize $h_{24}$}};

    \node[] () at (-11,0.4) {{\scriptsize $h_{023}$}};
    \node[] () at (-10,0.4) {{\scriptsize $h_{034}$}};
    \node[] () at (-9,0.4) {{\scriptsize $h_{123}$}};
    \node[] () at (-8,0.4) {{\scriptsize $h_{134}$}};
    \node[] () at (-7,0.4) {{\scriptsize $h_{0123}$}};
    \node[] () at (-6,0.4) {{\scriptsize $h_{0124}$}};
    \node[] () at (-5,0.4) {{\scriptsize $h_{0134}$}};
    \node[] () at (-4,0.4) {{\scriptsize $h_{0234}$}};
    \node[] () at (-3,0.4) {{\scriptsize $h_{1234}$}};
    \node[] () at (-2,0.4) {{\scriptsize $h_{01234}$}};

    \end{tikzpicture}
    \caption{The complement of $\lhp$ for the $\N=4$ chordal irreducible entropy vector in \eqref{eq:Svec-ex-chordal}.
    Throughout this subsection, we will describe the various steps of the algorithm for this example, until we obtain a simple tree graph model realizing $\Svec$ (in \Cref{fig:graph-construction-eg-chordal}). In $\lhp$, the isolated vertices of the graph shown above are connected to all other vertices, and therefore belong to all max-cliques; for convenience, we denote the set of these vertices by $\widetilde{Q}$. It is easy to see that $\lhp$ has three max-cliques given by
    $Q_1=\widetilde{Q}\cup\{h_{34},h_{013},h_{234}\}$, $Q_2=\widetilde{Q}\cup\{h_{01},h_{013}\}$ and $Q_3=\widetilde{Q}\cup\{h_{24},h_{34},h_{234}\}$.
    (In the non-trivial connected component of the complement of $\lhp$, any max-clique can contain at most one vertex for each edge.)}
    \label{fig:chordal-eg-part-1}
\end{figure}

To explain the logic behind the construction, it is useful to first consider more closely the origin of the chordality condition from \cite{Hubeny:2024fjn}, and for this, we first need to review two results from \cite{Hubeny:2024fjn} about holographic graph models. The first one is essentially the graph version of the known relationship between positivity of the mutual information and connectivity of the entanglement wedge in holography \cite{Headrick:2013zda}.

\begin{lemma}
\label{lem:ew_connectivity-v2}
    For any simple graph model $\gf{G}$ and subsystem $\X$, if $\bs{\X}$ is positive then the subgraph $\gf{G}_\X$ induced by the minimal min-cut $\mmC_{\X}$ for $\X$ is connected.
\end{lemma}
\begin{proof}
    See \cite[Corollary 7]{Hubeny:2024fjn}.
\end{proof}

The second result establishes the relationship between the line graph of an entropy vector realizable by a graph model, and the set of subsystems $\X\subseteq\nsp$ whose minimal min-cuts induce connected subgraphs. Essentially, it follows from the fact that, like in the holographic context, the minimal min-cuts for two disjoint subsystems have empty intersection. 

\begin{lemma}
\label{lem:isomorphic_LH}
    For any entropy vector $\Svec$ that can be realized by a holographic graph model, and any simple graph model realizing $\Svec$, let 
    \begin{equation}
    \label{eq:sigma}
        \mathfrak{C}=\{\mmC_{\X}|\; \bs{\X}\in\pos\}.
    \end{equation}
    Then the intersection graph $\gf{L}_\mathfrak{C}$ of $\mathfrak{C}$ is isomorphic to the line graph $\lhp$ of the correlation hypergraph $\hp$ of the \emph{PMI} $\pmi$ of $\Svec$.
\end{lemma}
\begin{proof}
    See \cite[Lemma 4]{Hubeny:2024fjn}.
\end{proof}

For a given entropy vector $\Svec$, suppose now that we want to determine whether there exists a realization by a simple tree graph model $\gf{T}$. 
It should be intuitively clear that in general there is tension between the tree topology, and the requirement that multiple non-overlapping subgraphs are connected.
More precisely, since any connected subgraph of a tree is a \textit{subtree}, if there exists a simple tree graph model $\gf{T}$ realizing $\Svec$, then $\gf{L}_\mathfrak{C}$ in \Cref{lem:isomorphic_LH} is the intersection graph of a collection of subtrees of a tree. Graphs of this kind are known as \textit{subtree graphs} and are completely characterized by the following well-known result in graph theory, from which the necessary condition in \cite{Hubeny:2024fjn} for the existence of $\gf{T}$ follows immediately using \Cref{lem:isomorphic_LH}.

\begin{thm}[Buneman, Gavril, and Walter]
\label{thm:subtree}
     A graph is a subtree graph if and only if it is a chordal graph.
\end{thm}
\begin{proof}
    See for example \cite[Theorem 2.4]{book:intersection_graphs}.
\end{proof}

Having reviewed the key principles behind the proof of the necessary condition, we now turn to the explicit construction of $\gf{T}$. Suppose that, for a given chordal irreducible $\Svec$, there exists a realization  via a simple tree graph model. It is important to note that the chordality condition applies to $\lhp$, or equivalently, if $\gf{T}$ exists, to $\gf{L}_{\mathfrak{C}}$, and that in this intersection graph, a subtree of $\gf{T}$ induced by the minimal min-cut for a subsystem is only represented by a single vertex. In other words, while $\Svec$ specifies $\gf{L}_{\mathfrak{C}}$ for any $\gf{T}$ that could possibly realize $\Svec$, this graph does not immediately provide additional information about $\gf{T}$, in particular, not even about the number of vertices it needs to have. Indeed, the correlation hypergraph, by definition, has exactly $\N+1$ vertices, while $\gf{T}$ has $\N+1$ boundary vertices and, in principle, any number of bulk vertices. Accordingly, the first step in our construction is to determine a candidate vertex set for $\gf{T}$, as well as for each element of $\mathfrak{C}$. For this, we will rely on a result from the theory of hypergraphs, which we now review \cite{book:intersection_graphs}.

Given an arbitrary set $\A$, a family $\F=\{\B_1,\ldots, \B_n\}$ of subsets of $\A$ is said to satisfy the \textit{Helly condition} if the following holds: For every subfamily $\F'\subseteq \F$, if the members of $\F'$ intersect pairwise, then there is an element of $\A$ that belongs to all elements of $\F'$. In other words, if $\B_i\cap \B_j\neq\varnothing$ for all $\B_i,\B_j\in \F'$, then $\bigcap_{\B_i\in \F'} \B_i\neq\varnothing$. A hypergraph $\gf{H}=(V,E)$ is said to be a \textit{Helly hypergraph} if its set of hyperedges $E$ satisfies the Helly condition. As we will see momentarily, the Helly condition is useful to characterize a class of hypergraphs called \textit{tree hypergraphs} which will play a central role in what follows. 

\begin{defi}[Tree Hypergraphs]
\label{def:tree-hypergraph}
    A hypergraph $\gf{H}=(V,E)$ is called a tree hypergraph if there exists a tree $\gf{T}$, with vertex set $V$, such that, for each hyperedge $h$ of $\gf{H}$, there is a subtree $\gf{T}_h$ of $\gf{T}$ with vertex set $h$.
\end{defi}

The tree $\gf{T}$ in \Cref{def:tree-hypergraph} is in general non-unique, and any such tree will be called a \textit{host tree} of the tree hypergraph $\gf{H}$. Tree hypergraphs are characterized by the following result.

\begin{thm}[Duchet, Flament, and Slater]
\label{thm:tree-hypergraphs}
     A hypergraph is a tree hypergraph if and only if it is a Helly hypergraph with a chordal line graph.
\end{thm}
\begin{proof}
    See \cite[Theorem 2.7]{book:intersection_graphs}.
\end{proof}

We can now use this result to make an informed guess about what the vertex set of $\gf{T}$, as well as the minimal min-cuts for subsystems corresponding to positive $\bsets$, could be. If there exists a simple tree $\gf{T}$ that realizes $\Svec$, the subgraphs of $\gf{T}$ induced by the minimal min-cuts of subsystems whose $\bset$ is positive are subtrees, and we can view the set of all such minimal min-cuts as the set of hyperedges of a tree hypergraph with the same vertex set as $\gf{T}$. Our goal then is to construct a candidate for this tree hypergraph (as an intermediate step before the construction of $\gf{T}$ itself) starting from the correlation hypergraph, and our strategy will be to do so in a minimal fashion, i.e, by adding the smallest possible number of bulk vertices.

Notice that while we are assuming that $\lhp$ is chordal, the assumption $|Q^\cap|=0$ implies that $\hp$ is not a Helly hypergraph, and that indeed we do need to add new vertices. According to \Cref{thm:tree-hypergraphs} then, these new vertices must be distributed among the hyperedges in such a way that the resulting hypergraph is a Helly hypergraph and its line graph remains chordal. 

\begin{figure}
    \centering
    \begin{subfigure}{0.45\textwidth}
    {\footnotesize
    \begin{equation*}
    \begin{array}{c|ccccccccccccccc|}
    \hline
    v_0 & 1 & 1 & 0 & 0 & 1 & 1 & 1 & 1 & 0 & 1 & 1 & 0 & 0 & 1 & 0\\
    v_1 & 0 & 0 & 1 & 1 & 1 & 1 & 1 & 0 & 1 & 1 & 1 & 0 & 0 & 1 & 0\\
    v_2 & 1 & 0 & 1 & 0 & 1 & 1 & 0 & 1 & 1 & 1 & 0 & 1 & 0 & 0 & 1\\
    v_3 & 1 & 1 & 1 & 1 & 1 & 0 & 1 & 1 & 1 & 1 & 0 & 0 & 1 & 1 & 1\\
    v_4 & 0 & 1 & 0 & 1 & 0 & 1 & 1 & 1 & 1 & 1 & 0 & 1 & 1 & 0 & 1\\
    \hline
    \end{array}
    \end{equation*}
    }
    \caption{}
    \end{subfigure}
    \begin{subfigure}{0.45\textwidth}
    {\footnotesize
    \begin{equation*}
    \begin{array}{c|ccccccccccccccc|}
    \hline
    v_0 & 1 & 1 & 0 & 0 & 1 & 1 & 1 & 1 & 0 & 1 & 1 & 0 & 0 & 1 & 0\\
    v_1 & 0 & 0 & 1 & 1 & 1 & 1 & 1 & 0 & 1 & 1 & 1 & 0 & 0 & 1 & 0\\
    v_2 & 1 & 0 & 1 & 0 & 1 & 1 & 0 & 1 & 1 & 1 & 0 & 1 & 0 & 0 & 1\\
    v_3 & 1 & 1 & 1 & 1 & 1 & 0 & 1 & 1 & 1 & 1 & 0 & 0 & 1 & 1 & 1\\
    v_4 & 0 & 1 & 0 & 1 & 0 & 1 & 1 & 1 & 1 & 1 & 0 & 1 & 1 & 0 & 1\\
    \ddot{v}_{Q_1} & 1 & 1 & 1 & 1 & 1 & 1 & 1 & 1 & 1 & 1 & 0 & 0 & 1 & 1 & 1 \\
    \ddot{v}_{Q_2} & 1 & 1 & 1 & 1 & 1 & 1 & 1 & 1 & 1 & 1 & 1 & 0 & 0 & 1 & 0\\
    \ddot{v}_{Q_3} & 1 & 1 & 1 & 1 & 1 & 1 & 1 & 1 & 1 & 1 & 0 & 1 & 1 & 0 & 1\\
    \hline
    \end{array}
    \end{equation*}
    }
    \caption{}
    \end{subfigure}
    \caption{The incidence matrix (a) for the correlation hypergraph of the entropy vector in \eqref{eq:Svec-ex-chordal}, and (b), where we added three vertices corresponding to the max-cliques of $\lhp$. In particular, (b) shows the extension of the hyperedges of $\hp$, which will be interpreted as the vertex sets of the minimal min-cuts on $\gf{T}$ for the subsystems whose $\bset$ is positive. In both matrices, the first ten columns correspond to the isolated vertices in \Cref{fig:chordal-eg-part-1}, and the last five to those which are endpoints of at least one edge in the same figure.}
    \label{fig:inc-matrices}
\end{figure}

For an arbitrary chordal KC-PMI $\pmi$ such that $|Q^\cap|=0$ for all max-cliques $Q$ of $\lhp$, we define the \textit{minimal Helly extension} $\ddot{\gf{H}}_\pmi=(\ddot{V},\ddot{E})$ of $\hp$ as the Helly hypergraph constructed from $\hp$ as follows. We first construct $\ddot{V}$ by adding to $V$ one vertex $\ddot{v}_Q$ for each max-clique $Q$ of $\lhp$, i.e.,
\begin{equation}
    \ddot{V}=V\cup\{\ddot{v}_Q|\, Q\;\text{is a max-clique of $\lhp$}\}.
\end{equation}
The set of hyperedges $\ddot{E}$ is then obtained from $E$ by extending each hyperedge of $\hp$ as follows 
\begin{equation}
\label{eq:hyperedge-extension}
    h_{\X}\; \mapsto\; \ddot{h}_{\X} = h_{\X} \cup \{\ddot{v}_Q|\, \bs{\X}\in Q\}.
\end{equation}
In words, starting from a hyperedge $h_{\X}$ of $\hp$, we add to it all vertices $\ddot{v}_Q$ such that $\bs{\X}$ is an element of the max-clique $Q$ of $\lhp$. This construction ensures that any collection of pairwise intersecting hyperedges in $\ddot{\gf{H}}_\pmi$ (which by definition forms a clique) has a common element (namely the vertex corresponding to a max-clique containing this collection), and therefore satisfies the Helly condition.
The minimal Helly extension of the correlation hypergraph for the entropy vector in \Cref{fig:chordal-eg-part-1} is shown in \Cref{fig:inc-matrices}.

In general, the minimal Helly extension then has the following property. 
\begin{lemma}
\label{lem:minimal-helly-extension}
     The line graph $\ddot{\gf{L}}_\pmi$ of the minimal Helly extension $\ddot{\gf{H}}_\pmi$ of $\hp$, and the line graph $\lhp$ of $\hp$, are isomorphic.
\end{lemma}
\begin{proof}
    It suffices to show that for any $\X,\Y$ such that $\bs{\X},\bs{\Y}\in\pos$
    \begin{equation}
        h_{\X}\cap h_{\Y}=\varnothing \quad \iff \quad \ddot{h}_{\X}\cap \ddot{h}_{\Y}=\varnothing.
    \end{equation}
    In one direction, if $h_{\X}\cap h_{\Y}=\varnothing$ then there is no clique of $\lhp$ that contains both $h_{\X}$ and $h_{\Y}$, and by the definition in \eqref{eq:hyperedge-extension}, it follows that $\ddot{h}_{\X}\cap \ddot{h}_{\Y}=\varnothing$. Conversely, if $\ddot{h}_{\X}\cap \ddot{h}_{\Y}=\varnothing$, the inclusions $h_{\X}\subseteq \ddot{h}_{\X}$ and $h_{\Y}\subseteq \ddot{h}_{\Y}$ immediately imply that $h_{\X}\cap h_{\Y}=\varnothing$.
\end{proof}
\Cref{lem:minimal-helly-extension} then immediately implies the following key corollary.
\begin{cor}
   The minimal Helly extension of a correlation hypergraph with chordal line graph is a tree hypergraph. 
\end{cor}
\begin{proof}
    The statement follows immediately from \Cref{lem:minimal-helly-extension}, the fact that by construction the minimal Helly extension is a Helly hypergraph, and \Cref{thm:tree-hypergraphs}.
\end{proof}

Let us briefly summarize what we have established thus far. Given a chordal irreducible entropy vector $\Svec$, we have constructed a tree hypergraph $\ddot{\gf{H}}_\pmi$, the minimal Helly extension of $\hp$, whose line graph coincides with $\lhp$, as well as with the intersection graph of $\mathfrak{C}$, for any simple tree graph model $\gf{T}$ that can possibly realize $\Svec$. The next step then, is to explicitly construct a candidate topology for $\gf{T}$. 

In general, for arbitrary tree hypergraphs, the host tree is highly non-unique, and in principle this might be true even in our restricted set-up. In particular, given $\Svec$, there might exist a host tree of $\ddot{\gf{H}}_\pmi$ such that at least one boundary vertex is not a leaf. We now argue that we can ignore all such trees.

Consider an arbitrary simple tree graph model $\gf{G}$ such that all edge weights are strictly positive,\footnote{\,Given a simple tree graph model with vanishing edge weights, we can delete the corresponding edges to obtain a simple forest that realizes the same entropy vector. Furthermore, this entropy vector necessarily has vanishing entries and it is not irreducible, so for the purpose of this discussion, we can ignore this situation.} and for some party $\ell$, the boundary vertex $v_\ell$ is not a leaf. We denote by $k\geq 2$ the degree of $v_\ell$, and by $\{\gf{G}^{(i)}\}_{i=1}^k$ the connected components of the graph obtained from $\gf{G}$ by deleting $v_\ell$. We further denote by $\X_i$ the subsystem corresponding to the boundary vertices of $\gf{G}^{(i)}$, and assume that for all $i\in[k]$ we have $\X_i\neq\varnothing$.\footnote{\,If this is not the case we can simplify $\gf{G}$ by deleting the vertices in $\gf{G}^{(i)}$ and the resulting graph model is a simple tree realizing the same entropy vector as $\gf{G}$.} In $\lhp$, the vertices corresponding to the positive $\bsets$ $\bs{Y}$ such that $\Y$ contains $\ell$ form a clique, and we denote it by $Q^\ell$. Notice that $Q^\ell$ is non-empty, and in particular, that $\bs{\X_i\ell}$ is positive for all $\X_i$.\footnote{\,Notice that, since $v_\ell$ is not a leaf and $\gf{G}$ is a simple tree, the entropy vector of $\gf{G}$ can also be realized by a sum of simpler tree graph models obtained from the collection $\{\gf{G}^{(i)}\}_{i=1}^k$ by introducing a “copy’’ of the vertex $v_\ell$ for each element of the collection and connecting this copy to the same vertex to which $v_\ell$ was connected in $\gf{G}$, with the corresponding edge assigned the same weight. Therefore, since for each of these subgraphs the system $\X_i\ell$ is the ``full system'', if $\bs{\X_i\ell}$ is not positive, the subgraph must have vanishing edge weights, contradicting our assumption.} If $Q^\ell$ is not a max-clique, there exists a subsystem $\Z$ that does not contain $\ell$, whose $\bset$ is positive, and such that it has non-empty intersection with each subsystem corresponding to an element of $Q^\ell$. But since minimal min-cuts for non-intersecting subsystems cannot intersect \cite{Avis:2021xnz}, any subsystem which does not contain $\ell$ and whose $\bset$ is positive is necessarily a subset of $\X_i$ for some $i$. Therefore $\Z$ can have non-empty intersection only with a single $\X_i\ell$. Since $k\geq 2$, this implies that $Q^\ell$ is a max-clique $Q$ such that $\ell\in Q^\cap$, contradicting (iii) in \eqref{eq:assumptions}.

This argument demonstrates that for any chordal irreducible entropy vector $\Svec$, if there is a simple tree graph model $\gf{T}$ that realizes $\Svec$, all boundary vertices must be leaves. Therefore, a necessary condition for the realizability of $\Svec$ by a simple tree graph model is that there exists a host tree of the minimal Helly extension of the correlation hypergraph of $\Svec$ such that all boundary vertices are leaves. We will now show that such a host tree always exists, and how to construct it.

\begin{figure}[tbp]
    \centering
    \begin{subfigure}{0.4\textwidth}
    \centering
    \begin{tikzpicture}
    \filldraw[] (0,1) circle (2pt);
    \filldraw[] (1,2) circle (2pt);
    \filldraw[] (1,0) circle (2pt);
    \filldraw[] (2,1) circle (2pt);
    \filldraw[] (3,2) circle (2pt);
    \filldraw[] (3,0) circle (2pt);
    \filldraw[] (4,1) circle (2pt);
    
    \draw[] (0,1) -- (1,2);
    \draw[] (0,1) -- (2,1);
    \draw[] (0,1) -- (1,0);
    \draw[] (1,2) -- (1,0);
    \draw[] (1,2) -- (2,1);
    \draw[] (1,0) -- (2,1);
    \draw[] (2,1) -- (3,2);
    \draw[] (2,1) -- (4,1);
    \draw[] (1,0) -- (4,1);
    \draw[] (1,0) -- (3,0);
    \draw[] (3,0) -- (4,1);
    
    \node[] () at (-0.3,1) {{\scriptsize $1$}};
    \node[] () at (1,2.3) {{\scriptsize $2$}};
    \node[] () at (1,-0.3) {{\scriptsize $3$}};
    \node[] () at (2,1.3) {{\scriptsize $4$}};
    \node[] () at (3,2.3) {{\scriptsize $5$}};
    \node[] () at (3,-0.3) {{\scriptsize $6$}};
    \node[] () at (4,1.3) {{\scriptsize $7$}};
    \end{tikzpicture}
    \subcaption[]{}
    \end{subfigure}
    \begin{subfigure}{0.4\textwidth}
    \centering
    \begin{tikzpicture}
    \draw[] (0,0) -- (2,0);
    \draw[] (0,2) -- (2,0);
    \draw[] (0,2) -- (2,2);
    \draw[] (2,2) -- (2,0);
    
    \filldraw[fill=green!80!black] (0,0) circle (2pt);
    \filldraw[fill=green!80!black] (2,0) circle (2pt);
    \filldraw[fill=green!80!black] (0,2) circle (2pt);
    \filldraw[fill=green!80!black] (2,2) circle (2pt);
    
    \node[] () at (0,-0.3) {{\scriptsize $367$}};
    \node[] () at (0,2.3) {{\scriptsize $1234$}};
    \node[] () at (2,-0.3) {{\scriptsize $347$}};
    \node[] () at (2,2.3) {{\scriptsize $45$}};

    \node[green!80!black] () at (1,2.3) {{\scriptsize $1$}};
    \node[green!80!black] () at (2.3,1) {{\scriptsize $1$}};
    \node[green!80!black] () at (1,-0.3) {{\scriptsize $2$}};
    \node[green!80!black] () at (1.2,1.2) {{\scriptsize $2$}};
    \end{tikzpicture}
    \subcaption[]{}
    \end{subfigure}
    \bigskip \\
    \bigskip
    \begin{subfigure}{0.4\textwidth}
    \centering
    \begin{tikzpicture}
    \draw[] (0,0) -- (2,0);
    \draw[] (0,2) -- (2,0);
    \draw[] (0,2) -- (2,2);
    
    \filldraw[fill=bulkcol] (0,0) circle (2pt);
    \filldraw[fill=bulkcol] (2,0) circle (2pt);
    \filldraw[fill=bulkcol] (0,2) circle (2pt);
    \filldraw[fill=bulkcol] (2,2) circle (2pt);
    
    \node[] () at (0,-0.3) {{\scriptsize $367$}};
    \node[] () at (0,2.3) {{\scriptsize $1234$}};
    \node[] () at (2,-0.3) {{\scriptsize $347$}};
    \node[] () at (2,2.3) {{\scriptsize $45$}};
    \end{tikzpicture}
    \subcaption[]{}
    \end{subfigure}
    \begin{subfigure}{0.4\textwidth}
    \centering
    \begin{tikzpicture}
    \draw[] (0,0) -- (2,0);
    \draw[] (0,2) -- (2,0);
    \draw[] (2,2) -- (2,0);
    
    \filldraw[fill=bulkcol] (0,0) circle (2pt);
    \filldraw[fill=bulkcol] (2,0) circle (2pt);
    \filldraw[fill=bulkcol] (0,2) circle (2pt);
    \filldraw[fill=bulkcol] (2,2) circle (2pt);
    
    \node[] () at (0,-0.3) {{\scriptsize $367$}};
    \node[] () at (0,2.3) {{\scriptsize $1234$}};
    \node[] () at (2,-0.3) {{\scriptsize $347$}};
    \node[] () at (2,2.3) {{\scriptsize $45$}};
    \end{tikzpicture}
    \subcaption[]{}
    \end{subfigure}
    \caption{An example of subtree graph (a) and its weighted clique graph (b). Each vertex of (b) corresponds to a max-clique of (a), and each edge weight is the cardinality of the intersection of the max-cliques of (a) corresponding to its endpoints. The panels (c) and (d) show the two maximum spanning trees of (b); they are the two clique tree representations of (a). Indeed, notice that in both (c) and (d), for any vertex $i\in[7]$ of (a), the set of vertices which contain $i$ induce a subtree.
    For simplicity, in this example we label the vertices of the subtree graph in (a) by single integers, but in our context the subtree graph is $\lhp$, and its vertices correspond to $h_{\X}$ associated with subsystems $\X$ with positive $\bs{\X}$.}
    \label{fig:clique-tree-eg}
\end{figure}
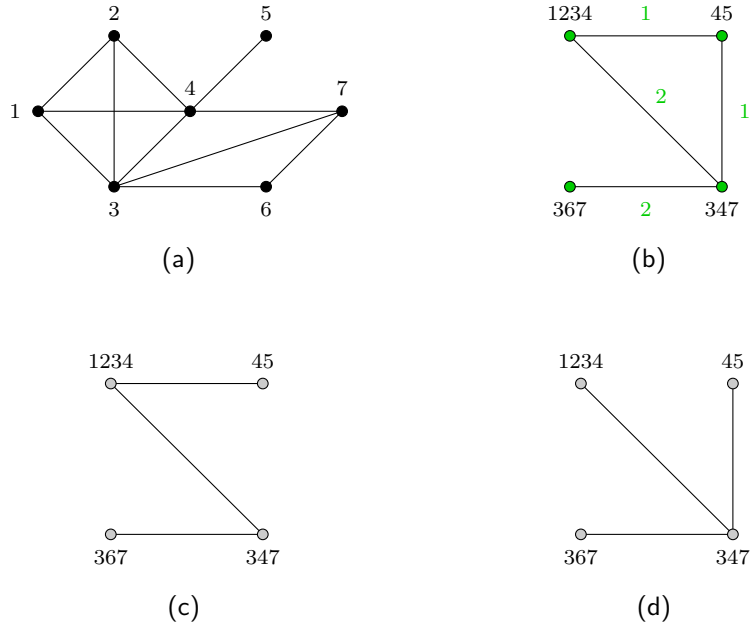

Recall that since $\lhp$ is chordal, it is a subtree graph by \Cref{thm:subtree}, i.e., it is the intersection graph of a family of subtrees of a tree. For any subtree graph $\gf{G}$, the tree and the subtrees in the family are called a \textit{tree representation} of the graph. Furthermore, it turns out that a tree representation can always be found such that the vertex set of the tree is the set of max-cliques of $\gf{G}$, and the subtree corresponding to a vertex $v$ of $\gf{G}$ is induced by the vertices of the tree corresponding to the max-cliques of $\gf{G}$ that contain $v$ \cite{book:intersection_graphs}. A tree representation of this kind is called a \textit{clique tree representation}, or simply a \textit{clique tree}. It is a well known result in graph theory that for any subtree graph a clique tree representation can always be found efficiently. The procedure amounts to first constructing the \textit{weighted clique graph} $\gf{K}_\pmi$ of $\gf{G}$, which is the intersection graph of the set of max-cliques of $\gf{G}$ where each edge has a weight given by the cardinality of the intersection of the max-cliques corresponding to its endpoints, and then a maximum \textit{spanning tree}\footnote{\,A spanning tree of a graph $\gf{G}$ is a tree which is a subgraph of $\gf{G}$ and has the same vertex set as $\gf{G}$. If $\gf{G}$ is weighted, a maximum spanning tree is a spanning tree which maximizes the sum of the weights on its edges.} of the weighted clique graph \cite{book:intersection_graphs}. As for the host tree of a tree hypergraph, in general a clique tree for a subtree graph is not unique, and we show an example of this construction in \Cref{fig:clique-tree-eg}.

Given a chordal irreducible entropy vector $\Svec$, we now denote by $\widetilde{\gf{T}}$ an arbitrary clique tree for $\lhp$. Notice that its vertices are precisely the bulk vertices that we have added to $\hp$ via the minimal Helly extension, i.e., the max-cliques of $\lhp$. However, $\widetilde{\gf{T}}$ is not yet a host tree for $\ddot{\gf{H}}_\pmi$, or even a ``topological'' holographic graph model, since it does not have boundary vertices. 

For an arbitrary party $\ell\in\nsp$, consider again the (not necessarily maximal) clique $Q^{\ell}$ of $\lhp$ corresponding to the set of all hyperedges of $\hp$ that contain the vertex $v_\ell$. It was shown in \cite[Theorem 9]{Hubeny:2024fjn} that for any KC-PMI and party $\ell$, this clique is contained in precisely one max-clique of $\lhp$.\footnote{\,Technically, $Q^\ell$ can be empty, and in that case it would be trivially contained in any max-clique of $\lhp$. However, this is only possible if $v_\ell$ does not belong to any hyperedge of $\hp$, or equivalently, it is an isolated vertex. In this case, however, $\hp$ is disconnected; a possibility we have excluded from the beginning of this subsection.} We can then use this information to determine, for each boundary vertex $v_\ell$, to which vertex of $\widetilde{\gf{T}}$ it should be connected. In practice, for each $\ell$, we find the unique max-clique $Q$ that contains $Q^\ell$, and then add to $\widetilde{\gf{T}}$ a vertex $v_\ell$ and an edge between $v_\ell$ and $v_Q$. It remains to be shown that the tree $\gf{T}$ resulting from this construction is indeed a host tree of $\ddot{\gf{H}}_\pmi$.

\begin{figure}[tbp]
    \centering
     \begin{subfigure}{0.28\textwidth}
    \centering
    \begin{tikzpicture}
    \draw[] (-1,0) -- (0,1.5);
    \draw[] (1,0) -- (0,1.5);
    \draw[] (-1,0) -- (1,0);

    \filldraw[fill=green!80!black] (-1,0) circle (2pt);
    \filldraw[fill=green!80!black] (1,0) circle (2pt);
    \filldraw[fill=green!80!black] (0,1.5) circle (2pt);
    
    \node[] () at (-1.2,-0.3) {{\scriptsize $Q_1$}};
    \node[] () at (1.2,-0.3) {{\scriptsize $Q_2$}};
    \node[] () at (0,1.8) {{\scriptsize $Q_3$}};

    \node[green!80!black] () at (-0.7,0.9) {{\scriptsize $12$}};
    \node[green!80!black] () at (0.7,0.9) {{\scriptsize $10$}};
    \node[green!80!black] () at (0,-0.2) {{\scriptsize $11$}};
    \end{tikzpicture}
    \subcaption[]{}
    \end{subfigure}
    \begin{subfigure}{0.28\textwidth}
    \centering
    \begin{tikzpicture}
    \draw[] (-1,0) -- (0,1.5);
    \draw[] (-1,0) -- (1,0);

    \filldraw[fill=bulkcol] (-1,0) circle (2pt);
    \filldraw[fill=bulkcol] (1,0) circle (2pt);
    \filldraw[fill=bulkcol] (0,1.5) circle (2pt);
    
    \node[] () at (-1.2,-0.3) {{\scriptsize $v_{Q_1}$}};
    \node[] () at (1.2,-0.3) {{\scriptsize $v_{Q_2}$}};
    \node[] () at (0,1.8) {{\scriptsize $v_{Q_3}$}};
    \end{tikzpicture}
    \subcaption[]{}
    \end{subfigure}
    \begin{subfigure}{0.4\textwidth}
    \centering
    \begin{tikzpicture}
    \draw (0,0) -- (-1,-1);
    \draw (0,0) -- (-1,1);
    \draw (0,0) -- (2,0);
    \draw (1,0) -- (1,1);
    \draw (2,0) -- (3,1);
    \draw (2,0) -- (3,-1);
    
    \filldraw (-1,-1) circle (2pt);
    \filldraw (-1,1) circle (2pt);
    \filldraw (1,1) circle (2pt);
    \filldraw (3,1) circle (2pt);
    \filldraw (3,-1) circle (2pt);
    
    \filldraw[fill=bulkcol] (0,0) circle (2pt);
    \filldraw[fill=bulkcol] (1,0) circle (2pt);
    \filldraw[fill=bulkcol] (2,0) circle (2pt);
    
    \node[] () at (-1.3,1) {{\scriptsize $v_1$}};
    \node[] () at (-1.3,-1) {{\scriptsize $v_0$}};
    \node[] () at (3.3,1) {{\scriptsize $v_2$}};
    \node[] () at (3.3,-1) {{\scriptsize $v_4$}};
    \node[] () at (1,1.3) {{\scriptsize $v_3$}};
    \node[] () at (-0.4,0) {{\scriptsize $v_{Q_2}$}};
    \node[] () at (1,-0.3) {{\scriptsize $v_{Q_1}$}};
    \node[] () at (2.4,0) {{\scriptsize $v_{Q_3}$}};
    
    \node[red] () at (-0.5,0.8) {{\scriptsize $1$}};
    \node[red] () at (-0.5,-0.8) {{\scriptsize $1$}};
    \node[red] () at (0.5,0.2) {{\scriptsize $1$}};
    \node[red] () at (1.2,0.5) {{\scriptsize $2$}};
    \node[red] () at (1.5,-0.2) {{\scriptsize $2$}};
    \node[red] () at (2.5,0.8) {{\scriptsize $1$}};
    \node[red] () at (2.5,-0.8) {{\scriptsize $2$}};


    \end{tikzpicture}
    \subcaption[]{}
    \end{subfigure}
    \caption{The construction of the simple tree graph model realizing the entropy vector $\Svec$ in \eqref{eq:Svec-ex-chordal}. The weighted clique graph of the minimal Helly extension of the correlation hypergraph of $\Svec$ is shown in (a). Each vertex corresponds to a max-clique of $\lhp$, and for each edge, the weight is given by the cardinality of the intersection of the max-cliques corresponding to the endpoints. Notice that (a) has a unique maximum spanning tree, which is the clique tree of $\lhp$ in (b). Starting from (b), we add the boundary vertices as leaves, and then the edge weights, as described in the main text. The result of this construction is the simple tree graph model in (c) that one can readily verify to be a realization of $\Svec$.
    (In particular, note that only 7 out of 15 components of $\Svec$ were used to determine all edge weights in the graph model (c).  From the remaining components, the first 4 correspond to subsystems with vanishing $\bsets$, while the second 4 to positive ones; nevertheless, because of non-trivial linear dependence relations among the MI instances in $\pmi$, and complementarity of the entropy, even these latter components of $\Svec$ are correctly reproduced by our algorithm.)}
    \label{fig:graph-construction-eg-chordal}
\end{figure}
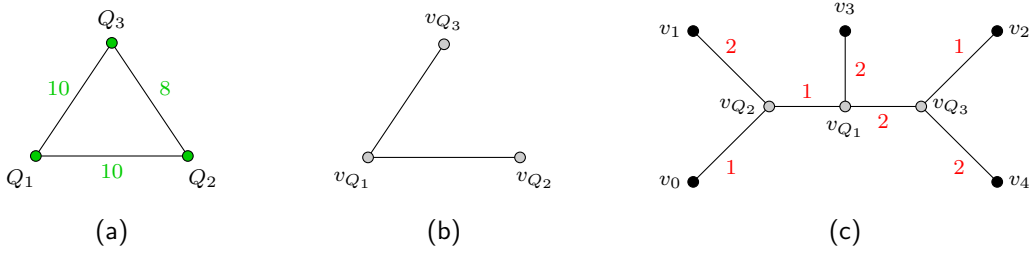

To see this, note that $\gf{T}$ fails to be a host tree of $\ddot{\gf{H}}_\pmi$ only if there exists a hyperedge $\ddot{h}$ of $\ddot{\gf{H}}_\pmi$ such that, for some boundary vertex $v_\ell$ in $\ddot{h}$, the clique $Q^\ell$ of $\lhp$ is contained in a max-clique $Q$ whose corresponding bulk vertex $v_Q$ of $\ddot{\gf{H}}_\pmi$ does not lie in $\ddot{h}$. However, this contradicts the construction of $\ddot{\gf{H}}_\pmi$, which assigns to every hyperedge $h$ of $\hp$ containing $v_\ell$ a bulk vertex $v_Q$ for each max-clique $Q$ of $\lhp$ that contains the vertex $h$ (and hence also the max-clique containing $Q^\ell$).

Having obtained a possible topology for $\gf{T}$, we now need to find the edge weights, but since $\gf{T}$ is a simple tree, it is straightforward to determine what they should be. In particular, since $\gf{T}$ is a tree, each edge $e$ of $\gf{T}$ is the \textit{only} cut edge for the cut $C$ corresponding to the choice of all vertices ``on one side'' of $e$. More precisely, $C$ is the set of vertices of one of the two connected components of the graph obtained from $\gf{T}$ by deleting the edge $e$. Furthermore, since $\gf{T}$ is simple, this cut is an $\X$-cut for $\X=C\cap \partial V$, although in principle $\X$ can be empty.\footnote{\,We believe that this situation actually never occurs. Intuitively, if $\gf{T}$ has a leaf $v_Q$ that is not a boundary vertex, then deleting $v_Q$ results in a new tree $\widetilde{\gf{T}}$ which is still a host tree of $\ddot{\gf{H}}$, contradicting the minimality of its vertex set. Nevertheless, in the interest of generality, we are not excluding this possibility at this stage.}

If $\gf{T}$ is indeed a graph model realization of $\Svec$, and $\X\neq\varnothing$, then $C$ is a min-cut for $\X$. To see this, first notice that if $C$ is not a min-cut, the edge $e$ is not a cut edge for \textit{any} min-cut, and we can simply contract the edge to obtain a new simple tree realizing $\Svec$ \cite{Hernandez-Cuenca:2022pst}. But this tree has one vertex less than $\gf{T}$, and it cannot possibly be a graph model for $\Svec$, because the vertex set of $\gf{T}$, which has been obtained from the minimal Helly extension, contains the smallest possible number of bulk vertices. Therefore, $C$ is a min-cut, and its cost, which is the weight of $e$, is equal to the component $\ent_{\I}$ of $\Svec$, where $\I=\X$ if $\X$ does not contain the purifier and $\I=\X^\complement$ otherwise. 

Based on this argument, for an arbitrary edge $e$ of $\gf{T}$, we denote by $\langle\X,\X^\complement\rangle_e$ the bipartition of $\nsp$ specified by $e$. If $\langle\X,\X^\complement\rangle_e\neq \{\varnothing,\nsp\}$, we assign to $e$ an edge weight $w(e)$ equal to $\ent_{\I}$ (again, with $\I=\X$ or $\I=\X^\complement$ depending on which subsystem contains the purifier). If instead $\langle\X,\X^\complement\rangle_e = \{\varnothing,\nsp\}$, we set $w(e)=0$.\footnote{\,Notice that this is not in tension with the fact that $\Svec$ has, by assumption, no vanishing entries, since we are only setting to 0 the weight of edges which would not be cut edges for any min-cut. In fact, we believe that this situation will never occur in practice, and that for any $\Svec$, each edge of the tree obtained from Algorithm~\ref{alg:reconstruction} will satisfy $\langle\X,\X^\complement\rangle_e\neq \{\varnothing,\nsp\}$.} 
Notice that many components of $\Svec$ have not been utilized, but because of the linear relations corresponding to the elements of $\pmi$, these are determined by the others, cf.~\Cref{fig:graph-construction-eg-chordal} which completes the explicit construction for the entropy vector $\Svec$ in \eqref{eq:Svec-ex-chordal}. This concludes our construction, and we present a summary in Algorithm~\ref{alg:reconstruction}.

\begin{algorithm}[t]
\caption{Construction of a simple tree graph model $\gf{T}$ realizing a chordal irreducible entropy vector $\Svec$.}
\label{alg:reconstruction}
\BlankLine
\Input{an irreducible entropy vector $\Svec$ whose PMI is chordal KC-PMI}
\Output{a simple tree graph model $\gf{T}$ realizing $\Svec$}
\BlankLine
$\pmi \leftarrow \text{the PMI of}\; \Svec$\;
$\pos \leftarrow \text{the set of positive $\bsets$ of}\; \pmi$\;
$\lhp \leftarrow \text{the intersection graph of}\; \pos$\;
$\gf{K}_\pmi \leftarrow \text{the weighted clique graph of}\; \lhp$\;
$\widetilde{\gf{T}} \leftarrow \text{a choice of a maximum spanning tree of}\; \gf{K}_\pmi$\;
$\gf{T} \leftarrow \widetilde{\gf{T}}$\;
\For{each $\ell\in\nsp$}
{
    construct the clique $Q^\ell$ in $\lhp$\;
    find the unique max-clique $Q$ of $\lhp$ that contains $Q^\ell$\; 
    add to $\gf{T}$ a vertex $v_\ell$ for the party $\ell$\;
    add to $\gf{T}$ an edge connecting $v_\ell$ to the vertex $v_Q$ corresponding to $Q$\;
    }
\For{each edge e of $\gf{T}$}
{
    $\I\leftarrow$ the element of $\langle\X,\X^\complement\rangle_e$ that does not contain the purifier\;
    \eIf{$\I\neq \varnothing$}
    {$w(e) \leftarrow \ent_{\I}$\;}
    {$w(e) \leftarrow 0$\;}
    }
\end{algorithm}

\subsection{The non-chordal case}
\label{subsec:non-chordal-case}

Having discussed the chordal case, we now turn to irreducible entropy vectors for which $\lhp$ is \textit{not} chordal. Unlike the chordal case, we do not present a general algorithm here; instead, we discuss a potential strategy for constructing non-simple tree graph models through a combination of fine-grainings and the algorithm proposed for the chordal case. Our focus is on the general conceptual and practical aspects of this strategy, rather than on the details of its implementation, and we will highlight the key challenges that require further investigation, which we leave for future work.

Given an entropy vector $\Svec$ for some number of parties $\N$, suppose that it can be realized by a non-simple tree graph model, and let $\gf{T}$ be one such realization. By relabeling the boundary vertices of $\gf{T}$ such that they all correspond to distinct parties, we obtain a simple tree graph model $\gf{T}'$ realizing a new entropy vector $\Svecp$ for some number of parties $\N' > \N$. The vector $\Svecp$ is then a fine-graining of $\Svec$ with respect to the coarse-graining map $\cgmap$ which relabels the boundary vertices of $\gf{T}'$ by the original parties in $\gf{T}$. Since $\gf{T}'$ is a simple tree graph model, the PMI of $\Svecp$ is a chordal KC-PMI.

This simple observation immediately suggests a strategy to search for a non-simple tree graph model realizing a given irreducible entropy vector $\Svec$ for which $\lhp$ is not chordal. For a choice of CG-map $\cgmap$ (more on this below), we can construct the set of fine-grainings of $\Svec$ with respect to $\cgmap$, as defined in \S\ref{subsec:cg-and-fg}, and search for an entropy vector $\Svecp$ in $\mathscr{A}^{\Sigma}_\cgmap(\Svec)$ such that its PMI is chordal. If we can find $\Svecp$ we can construct $\gf{T}'$ using Algorithm~\ref{alg:reconstruction}, and obtain $\gf{T}$. In the remainder of this subsection, we will ignore the issue of whether Algorithm~\ref{alg:reconstruction} always succeeds, and instead focus on the conceptual and technical challenges of the strategy to find $\Svecp$, or alternatively, to establish unrealizability of $\Svec$ by tree graph models.

We now consider this strategy more closely. Given $\Svec$, let us fix an arbitrary value of $\N' > \N$ and a CG-map $\cgmap$ from the $\N'$-party system to the $\N$-party system. In principle, one could attempt to construct explicitly the full set of fine-grainings $\mathscr{A}^{\Sigma}_\cgmap(\Svec)$ of $\Svec$ with respect to $\cgmap$, but in practice this computation rapidly becomes prohibitively difficult as $\N'$ increases. Moreover, since our goal is to find an entropy vector $\Svecp$ whose PMI is chordal, such a brute-force construction is arguably unnecessarily complicated. A more judicious strategy is instead to search directly for a chordal fine-graining of the PMI of $\Svec$. In fact, even if we were able to compute $\mathscr{A}^{\Sigma}_\cgmap(\Svec)$ explicitly, the absence of any $\Svecp$ in this set whose PMI is chordal would still not allow us to conclude that $\Svec$ is not realizable by a tree graph model, since we might simply have chosen an unsuitable CG-map, or even an unsuitable value of $\N'$. By contrast, once a chordal fine-graining $\pmi'$ of the PMI of $\Svec$ has been found, an appropriate choice of $\Svecp$ can be efficiently identified by means of a linear program. We now turn to an illustrative example that indicates how one might hope to determine the correct value of $\N'$ and CG-map by analyzing the chordless cycles of $\lhp$.

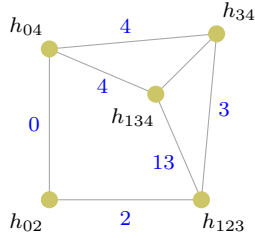
\begin{figure}
    \centering
    \begin{tikzpicture}
    \draw[gray!70, thin] (0,0) -- (2,0);
    \draw[gray!70, thin] (2,0) -- (1.4,1.4);
    \draw[gray!70, thin] (1.4,1.4) -- (0,2);
    \draw[gray!70, thin] (0,2) -- (0,0);
    \draw[gray!70, thin] (0,2) -- (2.2,2.2);
    \draw[gray!70, thin] (2.2,2.2) -- (2,0);
    \draw[gray!70, thin] (2.2,2.2) -- (1.4,1.4);
    
    \filldraw[Mcol!80!black] (0,0) circle (3pt);
    \filldraw[Mcol!80!black] (2,0) circle (3pt);
    \filldraw[EminusMcol!90!black] (1.4,1.4) circle (3pt);
    \filldraw[Mcol!80!black] (0,2) circle (3pt);
    \filldraw[Mcol!80!black] (2.2,2.2) circle (3pt);

    \node[] () at (-0.3,2.3) {{\scriptsize $h_{04}$}};
    \node[] () at (2.5,2.5) {{\scriptsize $h_{34}$}};
    \node[] () at (1.1,1.1) {{\scriptsize $h_{134}$}};
    \node[] () at (-0.3,-0.3) {{\scriptsize $h_{02}$}};
    \node[] () at (2.3,-0.3) {{\scriptsize $h_{123}$}};

    \node[blue] () at (-0.2,1) {{\scriptsize $0$}};
    \node[blue] () at (1,-0.25) {{\scriptsize $2$}};
    \node[blue] () at (1,2.3) {{\scriptsize $4$}};
    \node[blue] () at (2.3,1.2) {{\scriptsize $3$}};
    \node[blue] () at (0.7,1.5) {{\scriptsize $4$}};
    \node[blue] () at (1.5,0.5) {{\scriptsize $13$}};
    
    \end{tikzpicture} 
    \caption{Part of the graph $\lhp$ for the entropy vector in \eqref{eq:N4-example}. To simplify the figure we have shown only the vertices of $\lhp$ which are \textit{not} connected to all other vertices. The omitted vertices are: $h_{013}$, $h_{234}$, $h_{023}$, $h_{024}$, $h_{034}$, and all hyperedges with cardinality four or five. As clear from the figure, $\lhp$ has two chordless cycles of length 4, and is therefore not chordal. For each edge that belongs to at least one of the chordless cycles (i.e., all edges except for $\{h_{34},h_{134}\}$), we have indicated (in blue) the parties corresponding to the vertices in the intersection of the hyperedges of $\hp$ corresponding to its endpoints.}
    \label{fig:N4-example}
\end{figure}

Consider the following $\N=4$ irreducible entropy vector
\begin{equation}
\label{eq:N4-example}
    \Svec=(2, 4, 5, 5, 6, 7, 7, 9, 9, 8, 7, 11, 6, 8, 6),
\end{equation}
for which the non-trivial part of $\lhp$ is shown in \Cref{fig:N4-example}. As clear from the figure, $\lhp$ has two chordless cycles of length four, and it is therefore not chordal. Notice that, for example, party 2 is shared by $h_{02}$ and $h_{123}$, and that the edge connecting these vertices belongs to both chordless cycles. Consider now a \textit{refinement} of party 2 into parties 2 and 5 implemented by the pre-image of the CG-map specified by
\begin{equation}
\label{eq:cg-map-example-1}
    \cgmap:\quad \{\{0\}_0,\{1\}_1,\{2,5\}_2,\{3\}_3,\{4\}_4\}.
\end{equation}
Recall that the correlation hypergraph $\hpp^{(\text{min})}$ of the minimum fine-graining of $\pmi$ with respect to $\cgmap$ (cf., \Cref{thm:min-fg}) has, for each hyperedge $h_\X$ of $\hp$, a hyperedge $h'_{\Y'}$ for each $\Y'\subseteq\cginv(\X)$ that contains at least one party $\ell'\in\cginv(\ell)$ for each $\ell\in\X$. Therefore, for the specific case of $h_{02}$ and $h_{123}$, and CG-map \eqref{eq:cg-map-example-1}, $\hpp^{(\text{min})}$ includes the hyperedges
\begin{equation}
\label{eq:example-1-hyp-fg}
    h_{02}\; \rightarrow\; \{h'_{02} \,,\, h'_{05} \,,\, h'_{025}\} \qquad 
    h_{123}\; \rightarrow\; \{h'_{123} \,,\, h'_{135} \,,\, h'_{1235}\}.
\end{equation}
While $\hpp^{(\text{min})}$ is obviously not chordal, the set of hyperedges of the correlation hypergraph $\hpp$ of any other fine-graining of $\pmi$ (with respect to the same CG-map), is a subset of the set of hyperedges of $\hpp^{(\text{min})}$ (cf., \Cref{thm:partial-order-h}). Therefore, if there is a fine-graining $\pmi'$ whose correlation hypergraph $\hpp$ only has, for example, hyperedges $h'_{05}$ and $h'_{123}$ from \eqref{eq:example-1-hyp-fg}, then we would have effectively ``broken'' the two chordless cycles of $\lhp$, and $\pmi'$ would be a \textit{chordal fine-graining} of $\pmi$, since in the line graph of $\hpp$ there is no longer an edge connecting $h'_{05}$ and $h'_{123}$. An example of this mechanism is illustrated in \Cref{fig:non-chodal-example}, which shows the complement of $\lhpp$ for one such fine-graining $\pmi'$. By solving a simple linear program, we can then verify whether in the face of the 5-party SAC corresponding to $\pmi'$ there is a vector $\Svec'$ such that $\cgsvec(\Svecp)=\Svec$; an example is
\begin{equation}
\label{eq:N4-example-fg}
    \text{{\small $\Svec'=\{2, 2, 5, 5, 2, 4, 7, 7, 4, 7, 7, 4, 8, 7, 7, 5, 9, 6, 6, 9, 9, 6, 9, 9, 10, 4, 7, 11, 8, 8, 6\}$}}.
\end{equation}
Lastly, we can use Algorithm~\ref{alg:reconstruction} to find a simple tree graph model realization $\gf{T}$ of $\Svec'$, and by relabeling the boundary vertices of $\gf{T}$ according to $\cgmap$, we obtain a non-simple tree graph model realizing $\Svec$ (see again \Cref{fig:non-chodal-example}).

\begin{figure}[tbp]
    \centering
    \begin{subfigure}{0.59\textwidth}
    \centering
    \begin{tikzpicture}

    \draw[gray!70, thin] (0,0) -- (-1,-0.5);
    \draw[gray!70, thin] (0,0) -- (-1,0.5);
    \draw[gray!70, thin] (0,0) -- (2,0);
    \draw[gray!70, thin] (2,0) -- (3,1);
    \draw[gray!70, thin] (2,0) -- (3,-1);
    \draw[gray!70, thin] (0,0) -- (0,1);
    \draw[gray!70, thin] (0,0) -- (0,-1);

    \filldraw[Mcol!80!black] (0,0) circle (3pt);
    \filldraw[EminusMcol!90!black] (-1,-0.5) circle (3pt);
    \filldraw[EminusMcol!90!black] (0,1) circle (3pt);
    \filldraw[PminusEcol] (-1,0.5) circle (3pt);
    \filldraw[Mcol!80!black] (0,-1) circle (3pt);
    \filldraw[Mcol!80!black] (2,0) circle (3pt);
    \filldraw[PminusEcol] (3,1) circle (3pt);
    \filldraw[Mcol!80!black] (3,-1) circle (3pt);

    \filldraw[Mcol!80!black] (-2.5,-2.5) circle (3pt);
    \filldraw[Mcol!80!black] (-1.5,-2.5) circle (3pt);
    \filldraw[EminusMcol!90!black] (-0.5,-2.5) circle (3pt);
    \filldraw[PminusEcol] (0.5,-2.5) circle (3pt);
    \filldraw[PminusEcol] (1.5,-2.5) circle (3pt);
    \filldraw[PminusEcol] (2.5,-2.5) circle (3pt);
    \filldraw[PminusEcol] (3.5,-2.5) circle (3pt);
    \filldraw[PminusEcol] (4.5,-2.5) circle (3pt);
    \filldraw[PminusEcol] (-2.5,-3.5) circle (3pt);
    \filldraw[PminusEcol] (-1.5,-3.5) circle (3pt);
    \filldraw[PminusEcol] (-0.5,-3.5) circle (3pt);
    \filldraw[PminusEcol] (0.5,-3.5) circle (3pt);
    \filldraw[PminusEcol] (1.5,-3.5) circle (3pt);
    \filldraw[PminusEcol] (2.5,-3.5) circle (3pt);
    \filldraw[PminusEcol] (3.5,-3.5) circle (3pt);
    \filldraw[PminusEcol] (4.5,-3.5) circle (3pt);

    \node[] () at (0,1.4) {{\scriptsize $h'_{134}$}};
    \node[] () at (-1,0.9) {{\scriptsize $h'_{1234}$}};
    \node[] () at (0,-1.4) {{\scriptsize $h'_{34}$}};
    \node[] () at (-1,-0.9) {{\scriptsize $h'_{234}$}};
    \node[] () at (3,1.4) {{\scriptsize $h'_{045}$}};
    \node[] () at (3,-1.4) {{\scriptsize $h'_{04}$}};
    \node[] () at (0.35,0.3) {{\scriptsize $h'_{05}$}};
    \node[] () at (2,-0.3) {{\scriptsize $h'_{123}$}};

    \node[] () at (-2.5,-2.1) {{\scriptsize $h'_{013}$}};
    \node[] () at (-1.5,-2.1) {{\scriptsize $h'_{023}$}};
    \node[] () at (-0.5,-2.1) {{\scriptsize $h'_{0124}$}};
    \node[] () at (0.5,-2.1) {{\scriptsize $h'_{034}$}};
    \node[] () at (1.5,-2.1) {{\scriptsize $h'_{0123}$}};
    \node[] () at (2.5,-2.1) {{\scriptsize $h'_{0134}$}};
    \node[] () at (3.5,-2.1) {{\scriptsize $h'_{0135}$}};
    \node[] () at (4.5,-2.1) {{\scriptsize $h'_{0234}$}};

    \node[] () at (-2.5,-3.1) {{\scriptsize $h'_{0235}$}};
    \node[] () at (-1.5,-3.9) {{\scriptsize $h'_{0345}$}};
    \node[] () at (-0.5,-3.1) {{\scriptsize $h'_{01234}$}};
    \node[] () at (0.5,-3.9) {{\scriptsize $h'_{01235}$}};
    \node[] () at (1.5,-3.1) {{\scriptsize $h'_{01245}$}};
    \node[] () at (2.5,-3.9) {{\scriptsize $h'_{01345}$}};
    \node[] () at (3.5,-3.1) {{\scriptsize $h'_{02345}$}};
    \node[] () at (4.5,-3.9) {{\scriptsize $h'_{012345}$}};

    \node[] () at (0,-4) {};
    \end{tikzpicture}
    \subcaption[]{}
    \end{subfigure}
    \begin{subfigure}{0.39\textwidth}
    \centering
    \begin{tikzpicture}
    \draw (0,0) -- (-1,-1);
    \draw (0,0) -- (-1,1);
    \draw (0,0) -- (2,0);
    \draw (1,0) -- (1,1.5);
    \draw (2,0) -- (3,1);
    \draw (2,0) -- (3,-1);
    \draw (-1.5,0) -- (0,0);
    
    \filldraw (-1,-1) circle (2pt);
    \filldraw (-1,1) circle (2pt);
    \filldraw (-1.5,0) circle (2pt);
    \filldraw (1,1.5) circle (2pt);
    \filldraw (3,1) circle (2pt);
    \filldraw (3,-1) circle (2pt);
    
    \filldraw[fill=bulkcol] (0,0) circle (2pt);
    \filldraw[fill=bulkcol] (1,0) circle (2pt);
    \filldraw[fill=bulkcol] (2,0) circle (2pt);
    
    \node[] () at (-1.3,1) {{\scriptsize $v_3$}};
    \node[] () at (-1.3,-1) {{\scriptsize $v_1$}};
    \node[] () at (-1.8,0) {{\scriptsize $v_2$}};
    \node[] () at (3.3,1) {{\scriptsize $v_0$}};
    \node[] () at (3.3,-1) {{\scriptsize $v_2$}};
    \node[] () at (1,1.8) {{\scriptsize $v_4$}};
    
    \node[red] () at (-0.5,0.8) {{\scriptsize $5$}};
    \node[red] () at (-0.5,-0.8) {{\scriptsize $2$}};
    \node[red] () at (0.5,0.2) {{\scriptsize $5$}};
    \node[red] () at (1.2,0.75) {{\scriptsize $5$}};
    \node[red] () at (1.5,-0.2) {{\scriptsize $4$}};
    \node[red] () at (2.5,0.8) {{\scriptsize $6$}};
    \node[red] () at (2.5,-0.8) {{\scriptsize $2$}};
    \node[red] () at (-0.8,-0.2) {{\scriptsize $2$}};

    \node[] () at (0,-3) {};

    \end{tikzpicture}
    \subcaption[]{}
    \end{subfigure}
    \caption{An example of a chordal fine-graining $\pmi'$ of the KC-PMI $\pmi$ of $\Svec$ in \eqref{eq:N4-example} with respect to the CG-map specified in \eqref{eq:cg-map-example-1}, and a non simple tree realizing $\Svec$. (a) shows the complement of the line graph of $\hpp$ and (b) the non simple tree realizing $\Svec$. Relabeling the vertex $v_2$ on the right side of (b) gives a simple tree graph model whose entropy vector is $\Svec'$ in \eqref{eq:N4-example-fg} and whose PMI is $\pmi'$.
    }
    \label{fig:non-chodal-example}
\end{figure}
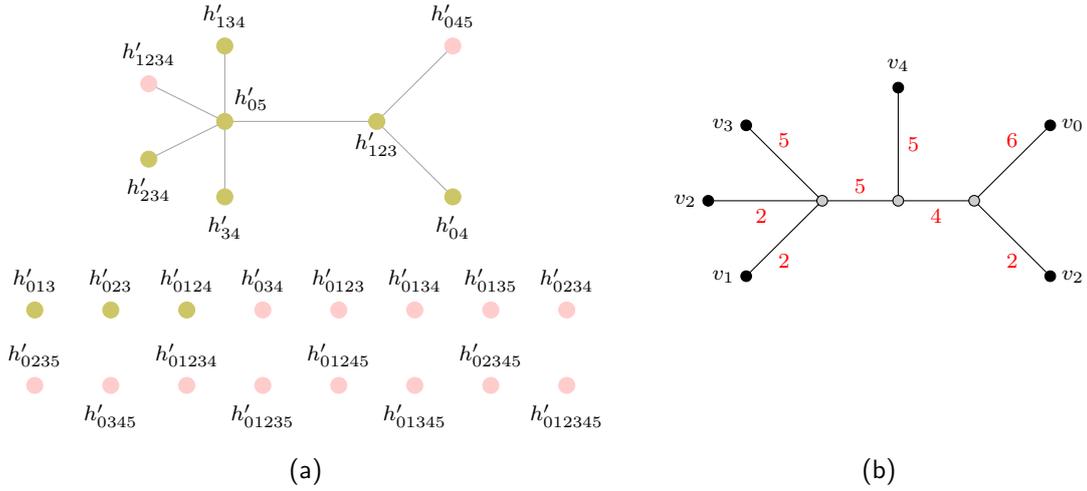

We conclude the analysis of this example with a few more comments about this strategy. First, note that we could have also considered a refinement of party 0, rather than 2, but not, individually, of parties 3 or 1, since any such refinement cannot break the two chordless cycles of $\lhp$ and give a chordal fine-graining of $\pmi$ (although a refinement of both 3 and 1 parties might). Another option would have been to refine party 4, although in this case, since there are two edges in $\lhp$ corresponding to pairs of vertices that share party 4, the necessary transformation of $h_{04}$, $h_{34}$ and $h_{134}$ is more subtle. Secondly, notice that from this type of analysis we might be able to determine which CG-maps are clearly not sufficient to obtain a chordal fine-graining, and which map \textit{might} be sufficient, but we cannot immediately determine the actual fine-grainings, even when they do exist. This is closely related to the fact that not all hypergraphs are correlation hypergraphs, and accordingly, the transformation of $\lhp$ under fine-grainings is in general highly non trivial. We can already see a signature of this complexity in the example that we have presented. Indeed, by examining the positive $\bsets$ of $\pmi'$ in \Cref{fig:non-chodal-example}, one can see that while $h_{02}$ and $h_{123}$ have been mapped (respectively) to $h'_{05}$ and $h'_{123}$ as desired, the transformation of the other vertices involving party 2 is harder to predict. For example, $h_{0234}$ is a hyperedge of $\hp$ and all of $h'_{0234}$, $h'_{0345}$ and $h'_{02345}$ are hyperedges of $\hpp$. But while one might have expected a similar transformation also for $h_{1234}$, only $h'_{1234}$ is a hyperedge of $\hpp$ and there are no hyperedges in $\hpp$ for the subsystems $1345$ and $12345$. Similarly, the hyperedge $h'_{045}$ originates from $h_{024}$, whereas $h'_{024}$ is no longer present. 

The main goal of the example we have just discussed was to elucidate a potential strategy for finding a non-simple tree realization of an entropy vector. However, as we stressed before, one of the central aspects of our program is also to understand whether a given entropy vector \textit{can} be realized by a non-simple tree graph model (and ultimately by any graph model; see \S\ref{sec:discussion}). Consider an entropy vector $\Svec$ with KC-PMI $\pmi$, and let $\pmi'_1$ be a chordal fine-graining of $\pmi$ with respect to some CG-map $\cgmap$. As mentioned above, we can use a linear program to verify whether the face of the $\N'_1$-party SAC corresponding to $\pmi'_1$ contains an entropy vector $\Svec'_1$ such that $\cgsvec(\Svec'_1)=\Svec$. Nevertheless, if this is not the case, we cannot conclude that there is no tree graph model realizing $\Svec$, since, even for the same CG-map, there might be another chordal fine-graining $\pmi'_2$ such that there exists an $\Svec'_2$ with $\cgsvec(\Svec'_2)=\Svec$. An example of this situation is illustrated in \Cref{fig:eg-different-chord-fgs}. Furthermore, even if this is not the case, there might be other chordal fine-grainings of $\pmi$ for different CG-maps, and any one of these might contain the required entropy vector. Therefore, in order to conclude that an entropy vector $\Svec$ cannot be realized by any tree graph model, in principle we need to scan over all possible CG-maps and chordal fine-grainings of $\pmi$. Since there are infinitely many CG-maps (as there is no upper bound on $\N'$), this problem, at least superficially, might seem undecidable. On the other hand, given that the holographic entropy cone is polyhedral for any number of parties, and the question can therefore be settled by knowledge of a finite number of holographic entropy inequalities, it is reasonable to expect that for any $\Svec$ there is effectively an upper bound on $\N'$ and a finite number of CG-maps (and therefore fine-grainings) that one needs to consider. We plan to address the issue of ``finiteness'' in forthcoming work \cite{graphoidal}.

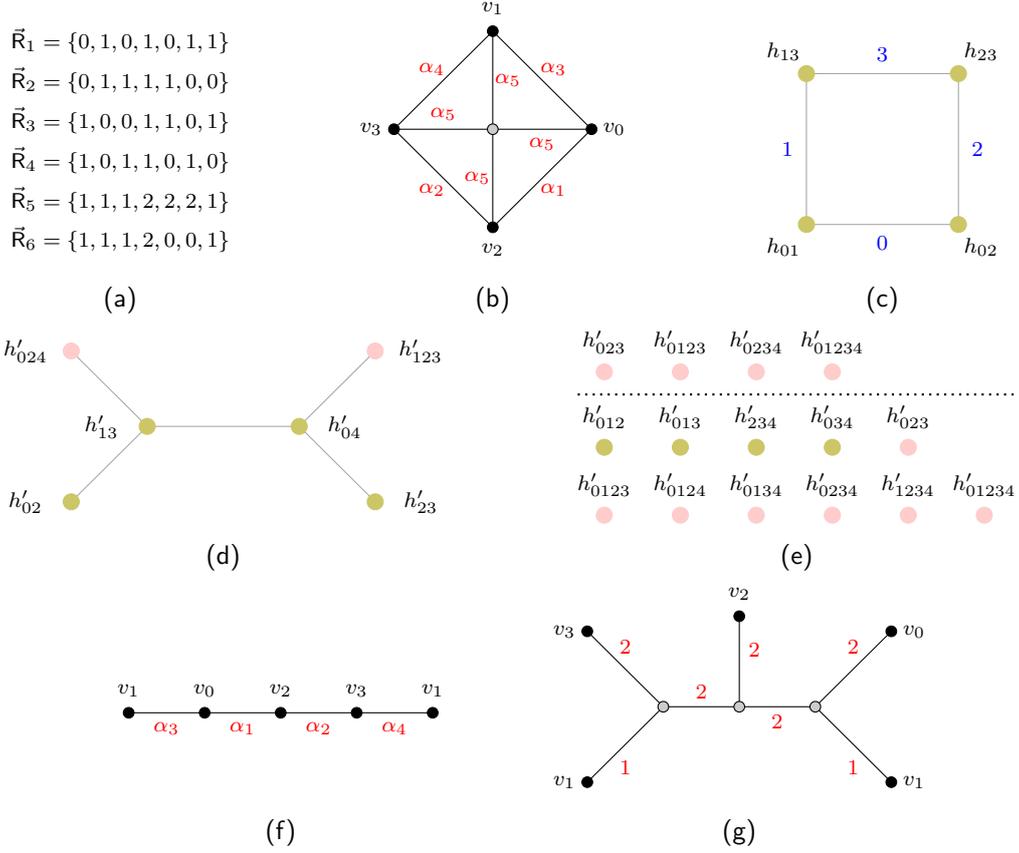
\begin{figure}
    \centering
    \begin{subfigure}{0.3\textwidth}
    \centering
    {\scriptsize
    \begin{align*}
        \vec{{\sf R}}_1=\{0, 1, 0, 1, 0, 1, 1\}\\
        \vec{{\sf R}}_2=\{0, 1, 1, 1, 1, 0, 0\}\\
        \vec{{\sf R}}_3=\{{1, 0, 0, 1, 1, 0, 1}\}\\
        \vec{{\sf R}}_4=\{1, 0, 1, 1, 0, 1, 0\}\\
        \vec{{\sf R}}_5=\{1, 1, 1, 2, 2, 2, 1\}\\
        \vec{{\sf R}}_6=\{1, 1, 1, 2, 0, 0, 1\}\\
    \end{align*}
    }
    \vspace{-1.2cm}
    \subcaption[]{}
    \end{subfigure}
    \begin{subfigure}{0.33\textwidth}
    \centering
    \begin{tikzpicture}
    \draw (0,0) -- (0,1.3);
    \draw (0,0) -- (0,-1.3);
    \draw (0,0) -- (-1.3,0);
    \draw (0,0) -- (1.3,0);
    \draw (0,1.3) -- (1.3,0);
    \draw (0,1.3) -- (-1.3,0);
    \draw (0,-1.3) -- (-1.3,0);
    \draw (0,-1.3) -- (1.3,0);
    
    \filldraw[fill=bulkcol] (0,0) circle (2pt);
    \filldraw (0,1.3) circle (2pt);
    \filldraw (0,-1.3) circle (2pt);
    \filldraw (-1.3,0) circle (2pt);
    \filldraw (1.3,0) circle (2pt);
    
    \node[] () at (0,1.6) {{\scriptsize $v_{1}$}};
    \node[] () at (0,-1.6) {{\scriptsize $v_{2}$}};
    \node[] () at (-1.6,0) {{\scriptsize $v_{3}$}};
    \node[] () at (1.6,0) {{\scriptsize $v_{0}$}};

    \node[red] () at (-0.8,0.8) {{\scriptsize $\alpha_4$}};
    \node[red] () at (0.8,-0.8) {{\scriptsize $\alpha_1$}};
    \node[red] () at (-0.8,-0.8) {{\scriptsize $\alpha_2$}};
    \node[red] () at (0.8,0.8) {{\scriptsize $\alpha_3$}};
    \node[red] () at (-0.65,0.2) {{\scriptsize $\alpha_5$}};
    \node[red] () at (0.65,-0.2) {{\scriptsize $\alpha_5$}};
    \node[red] () at (0.2,0.65) {{\scriptsize $\alpha_5$}};
    \node[red] () at (-0.2,-0.65) {{\scriptsize $\alpha_5$}};
    
    \end{tikzpicture} 
    \subcaption[]{}
    \end{subfigure}
    \begin{subfigure}{0.33\textwidth}
    \centering
    \begin{tikzpicture}
    \draw[gray!70, thin] (0,0) -- (2,0);
    \draw[gray!70, thin] (0,2) -- (0,0);
    \draw[gray!70, thin] (0,2) -- (2,2);
    \draw[gray!70, thin] (2,2) -- (2,0);
    
    \filldraw[Mcol!80!black] (0,0) circle (3pt);
    \filldraw[Mcol!80!black] (2,0) circle (3pt);
    \filldraw[Mcol!80!black] (0,2) circle (3pt);
    \filldraw[Mcol!80!black] (2,2) circle (3pt);

    \node[] () at (-0.3,2.3) {{\scriptsize $h_{13}$}};
    \node[] () at (2.3,2.3) {{\scriptsize $h_{23}$}};
    \node[] () at (-0.3,-0.3) {{\scriptsize $h_{01}$}};
    \node[] () at (2.3,-0.3) {{\scriptsize $h_{02}$}};

    \node[blue] () at (-0.25,1) {{\scriptsize $1$}};
    \node[blue] () at (1,-0.25) {{\scriptsize $0$}};
    \node[blue] () at (1,2.25) {{\scriptsize $3$}};
    \node[blue] () at (2.25,1) {{\scriptsize $2$}};
    
    \end{tikzpicture} 
    \subcaption[]{}
    \end{subfigure}
    \\
    \begin{subfigure}{0.49\textwidth}
    \centering
    \begin{tikzpicture}

    \draw[gray!70, thin] (0,0) -- (-1,-1);
    \draw[gray!70, thin] (0,0) -- (-1,1);
    \draw[gray!70, thin] (0,0) -- (2,0);
    \draw[gray!70, thin] (2,0) -- (3,1);
    \draw[gray!70, thin] (2,0) -- (3,-1);
    
    \filldraw[Mcol!80!black] (0,0) circle (3pt);
    \filldraw[Mcol!80!black] (-1,-1) circle (3pt);
    \filldraw[PminusEcol] (-1,1) circle (3pt);
    \filldraw[Mcol!80!black] (2,0) circle (3pt);
    \filldraw[PminusEcol] (3,1) circle (3pt);
    \filldraw[Mcol!80!black] (3,-1) circle (3pt);
    
    \node[] () at (-1.6,1) {{\scriptsize $h'_{024}$}};
    \node[] () at (-1.6,-1) {{\scriptsize $h'_{02}$}};
    \node[] () at (3.6,1) {{\scriptsize $h'_{123}$}};
    \node[] () at (3.6,-1) {{\scriptsize $h'_{23}$}};
    \node[] () at (-0.6,0) {{\scriptsize $h'_{13}$}};
    \node[] () at (2.6,0) {{\scriptsize $h'_{04}$}};
    
    \end{tikzpicture}
    \subcaption[]{}
    \end{subfigure}
    \begin{subfigure}{0.49\textwidth}
    \centering
    \begin{tikzpicture}

    \draw[dotted,thick] (-2.85,1.6) -- (2.95,1.6);

    \filldraw[PminusEcol] (-2.5,0) circle (3pt);
    \filldraw[PminusEcol] (-1.5,0) circle (3pt);
    \filldraw[PminusEcol] (-0.5,0) circle (3pt);
    \filldraw[PminusEcol] (0.5,0) circle (3pt);
    \filldraw[PminusEcol] (1.5,0) circle (3pt);
    \filldraw[PminusEcol] (2.5,0) circle (3pt);

    \node[] () at (-2.5,0.4) {{\scriptsize $h'_{0123}$}};
    \node[] () at (-1.5,0.4) {{\scriptsize $h'_{0124}$}};
    \node[] () at (-0.5,0.4) {{\scriptsize $h'_{0134}$}};
    \node[] () at (0.5,0.4) {{\scriptsize $h'_{0234}$}};
    \node[] () at (1.5,0.4) {{\scriptsize $h'_{1234}$}};
    \node[] () at (2.5,0.4) {{\scriptsize $h'_{01234}$}};

    \filldraw[EminusMcol!90!black] (-2.5,0.9) circle (3pt);
    \filldraw[EminusMcol!90!black] (-1.5,0.9) circle (3pt);
    \filldraw[EminusMcol!90!black] (-0.5,0.9) circle (3pt);
    \filldraw[EminusMcol!90!black] (0.5,0.9) circle (3pt);
    \filldraw[PminusEcol] (1.5,0.9) circle (3pt);

    \node[] () at (-2.5,1.3) {{\scriptsize $h'_{012}$}};
    \node[] () at (-1.5,1.3) {{\scriptsize $h'_{013}$}};
    \node[] () at (-0.5,1.3) {{\scriptsize $h'_{234}$}};
    \node[] () at (0.5,1.3) {{\scriptsize $h'_{034}$}};
    \node[] () at (1.5,1.3) {{\scriptsize $h'_{023}$}};
    
    \filldraw[PminusEcol] (-2.5,1.9) circle (3pt);
    \filldraw[PminusEcol] (-1.5,1.9) circle (3pt);
    \filldraw[PminusEcol] (-0.5,1.9) circle (3pt);
    \filldraw[PminusEcol] (0.5,1.9) circle (3pt);

    \node[] () at (-2.5,2.3) {{\scriptsize $h'_{023}$}};
    \node[] () at (-1.5,2.3) {{\scriptsize $h'_{0123}$}};
    \node[] () at (-0.5,2.3) {{\scriptsize $h'_{0234}$}};
    \node[] () at (0.5,2.3) {{\scriptsize $h'_{01234}$}};
    
    \end{tikzpicture} 
    \subcaption[]{}
    \end{subfigure}
    \\
    \begin{subfigure}{0.39\textwidth}
    \centering
    \begin{tikzpicture}
    \draw (0,0) -- (4,0);
    
    \filldraw (0,0) circle (2pt);
    \filldraw (1,0) circle (2pt);
    \filldraw (2,0) circle (2pt);
    \filldraw (3,0) circle (2pt);
    \filldraw (4,0) circle (2pt);
    
    \node[] () at (0,0.3) {{\scriptsize $v_1$}};
    \node[] () at (1,0.3) {{\scriptsize $v_0$}};
    \node[] () at (2,0.3) {{\scriptsize $v_2$}};
    \node[] () at (3,0.3) {{\scriptsize $v_3$}};
    \node[] () at (4,0.3) {{\scriptsize $v_1$}};

    \node[red] () at (0.5,-0.2) {{\scriptsize $\alpha_3$}};
    \node[red] () at (1.5,-0.2) {{\scriptsize $\alpha_1$}};
    \node[red] () at (2.5,-0.2) {{\scriptsize $\alpha_2$}};
    \node[red] () at (3.5,-0.2) {{\scriptsize $\alpha_4$}};

    \node[] () at (0,-1) {};

    \end{tikzpicture}
    \subcaption[]{}
    \end{subfigure}
    \begin{subfigure}{0.39\textwidth}
    \centering
    \begin{tikzpicture}
    \draw (0,0) -- (-1,-1);
    \draw (0,0) -- (-1,1);
    \draw (0,0) -- (2,0);
    \draw (1,0) -- (1,1.2);
    \draw (2,0) -- (3,1);
    \draw (2,0) -- (3,-1);
    
    \filldraw (-1,-1) circle (2pt);
    \filldraw (-1,1) circle (2pt);
    \filldraw (1,1.2) circle (2pt);
    \filldraw (3,1) circle (2pt);
    \filldraw (3,-1) circle (2pt);
    
    \filldraw[fill=bulkcol] (0,0) circle (2pt);
    \filldraw[fill=bulkcol] (1,0) circle (2pt);
    \filldraw[fill=bulkcol] (2,0) circle (2pt);
    
    \node[] () at (-1.3,1) {{\scriptsize $v_3$}};
    \node[] () at (-1.3,-1) {{\scriptsize $v_1$}};
    \node[] () at (3.3,1) {{\scriptsize $v_0$}};
    \node[] () at (3.3,-1) {{\scriptsize $v_1$}};
    \node[] () at (1,1.5) {{\scriptsize $v_2$}};
    
    \node[red] () at (-0.5,0.8) {{\scriptsize $2$}};
    \node[red] () at (-0.5,-0.8) {{\scriptsize $1$}};
    \node[red] () at (0.5,0.2) {{\scriptsize $2$}};
    \node[red] () at (1.2,0.75) {{\scriptsize $2$}};
    \node[red] () at (1.5,-0.2) {{\scriptsize $2$}};
    \node[red] () at (2.5,0.8) {{\scriptsize $2$}};
    \node[red] () at (2.5,-0.8) {{\scriptsize $1$}};

    \end{tikzpicture}
    \subcaption[]{}
    \end{subfigure}
    \caption{Up to permutations, the $\N=3$ KC-PMI $\pmi=\{\mi(1:2),\mi(3:0)\}$ is the simplest non-chordal KC-PMI; all $\N=2$ PMIs are chordal KC-PMIs, and any other non-chordal $\N=3$ KC-PMI corresponds to a higher dimensional face of the SAC. The only non-positive $\bsets$ of $\pmi$ are $\bs{12}$ and $\bs{03}$. The extreme rays (ERs) of the 5-dimensional face associated to $\pmi$ are listed in (a): The first four are realized by Bell pairs, the fifth by the perfect tensor, and last one violates MMI. Accordingly,
    there are entropy vectors in the interior of $\face$ which also violate MMI, and are therefore non realizable (e.g., $\Svec=\vec{{\sf R}}_6 + \epsilon\sum_{i=1}^5 \vec{{\sf R}}_i$ with sufficiently small $\epsilon>0$). On the other hand, $\pmi$ is realizable, since any conical combination $\Svec=\sum_{i=1}^5 \alpha_i \vec{{\sf R}}_i$, with $\alpha_i>0$ for all $i$, is in the interior of $\face$, and is realized by the graph in (b), which is simple, but not a tree.
    Indeed $\pmi$ is not chordal, as clear from $\lhp$, which is shown in (c) (to simplify the figure we have omitted the vertices corresponding to hyperedges with cardinality four and five, which are adjacent to all vertices). The structure of $\lhp$ indicates that refining a single party into two might be sufficient to find a chordal fine-graining of $\pmi$, and we choose the CG-map $\cgmap$ specified by $\{\{1,4\}_1,\{2\}_2,\{3\}_3,\{0\}_0\}$. In total there are 258 fine-grainings of $\pmi$ with respect to $\cgmap$, and 8 of them are chordal. The complement of the line graph for two such fine-grainings $\pmi'_1$ and $\pmi'_2$ are shown (jointly) by (d) and (e). Specifically, (d) shows the common non-trivial component of this graph, while (e) shows the isolated vertices for $\pmi'_1$ (top) and $\pmi'_2$ (bottom). The face $\face'_1$ of the SAC$_4$ corresponding to $\pmi'_1$ is simplicial and 4-dimensional. Its ERs are canonical fine-grainings of the first four ERs in (a), and any $\Svec$ from the combination above, where $\alpha_5=0$ and $\alpha_i>0$ for each $i\in[4]$, is in the interior of $\face$ (even if $\face$ is 5-dimensional) and is realized by the non-simple tree graph model in (f). Indeed, notice that (f) can be obtained from (b) by setting $\alpha_5=0$ (and deleting the corresponding edges and bulk vertex) and ``splitting'' the vertex $v_1$ into two vertices labeled by the same party. Setting instead $\alpha_i=1/2$ for all $i\in[4]$ and $\alpha_5=1$, we obtain $\Svec=\{2, 2, 2, 4, 3, 3, 2\}$, which cannot be obtained as a coarse-graining of some $\Svec'_1\in\face'_1$. It is however a coarse-graining of $\Svec'_2=\{1, 2, 2, 1, 3, 2, 2, 3, 3, 3, 2, 4, 3, 3, 2\}$ in the face $\face'_2$, the 6-dimensional face corresponding to $\pmi'_2$, as shown in (g) (the figure shows the non-simple tree realization of $\Svec$, relabeling one vertex $v_1$ as $v_4$ gives the simple tree realizing $\Svec'_2$).
    }
    \label{fig:eg-different-chord-fgs}
\end{figure}

An example of an entropy vector $\Svec$ that cannot be realized by a tree graph model (or in fact, any graph model) even if there exists a chordal fine-graining $\pmi'$ of its KC-PMI $\pmi$ is shown in \Cref{fig:eg-different-chord-fgs}. This example illustrates a typical situation: Given a face $\face$ of the SAC whose corresponding PMI $\pmi$ is a non-chordal KC-PMI, the set of entropy vectors in the interior of $\face$ that can be realized by non simple trees (and therefore via all possible chordal fine-grainings of $\pmi$ with respect to all possible CG-maps) is a proper subset of the interior of $\face$.

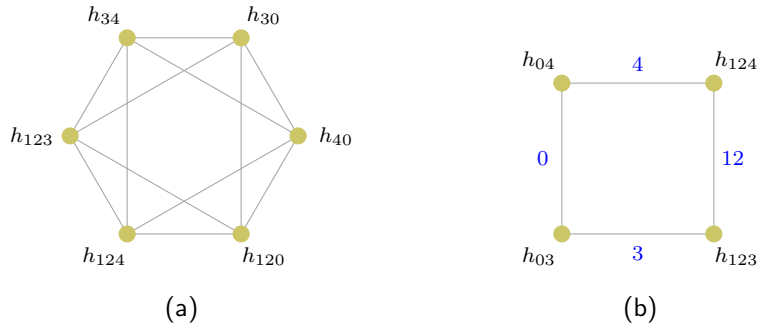
\begin{figure}
    \centering
    \begin{subfigure}{0.39\textwidth}
    \centering
    \begin{tikzpicture}
    \draw[gray!70, thin] (1.50, 0.00) -- (0.75, 1.30);
    \draw[gray!70, thin] (0.75, 1.30) -- (-0.75, 1.30);
    \draw[gray!70, thin] (-0.75, 1.30) -- (-1.50, 0.00);
    \draw[gray!70, thin] (-1.50, 0.00) -- (-0.75, -1.30);
    \draw[gray!70, thin] (-0.75, -1.30) -- (0.75, -1.30);
    \draw[gray!70, thin] (0.75, -1.30) -- (1.50, 0.00);

    \draw[gray!70, thin] (1.50, 0.00) -- (-0.75, 1.30);
    \draw[gray!70, thin] (0.75, 1.30) -- (-1.50, 0.00);
    \draw[gray!70, thin] (-0.75, 1.30) -- (-0.75, -1.30);
    \draw[gray!70, thin] (-1.50, 0.00) -- (0.75, -1.30);
    \draw[gray!70, thin] (-0.75, -1.30) -- (1.50, 0.00);
    \draw[gray!70, thin] (0.75, -1.30) -- (0.75, 1.30);

    \filldraw[Mcol!80!black] (1.50, 0.00) circle (3pt);
    \filldraw[Mcol!80!black] (0.75, 1.30) circle (3pt);
    \filldraw[Mcol!80!black] (-0.75, 1.30) circle (3pt);
    \filldraw[Mcol!80!black] (-1.50, 0.00) circle (3pt);
    \filldraw[Mcol!80!black] (-0.75, -1.30) circle (3pt);
    \filldraw[Mcol!80!black] (0.75, -1.30) circle (3pt);
    
    \node[] () at (2, 0.00) {{\scriptsize $h_{40}$}};
    \node[] () at (1.05, 1.60) {{\scriptsize $h_{30}$}};
    \node[] () at (-1.05, 1.60) {{\scriptsize $h_{34}$}};
    \node[] () at (-2, 0.00) {{\scriptsize $h_{123}$}};
    \node[] () at (-1.05, -1.60) {{\scriptsize $h_{124}$}};
    \node[] () at (1.05, -1.60) {{\scriptsize $h_{120}$}};
    
    \end{tikzpicture} 
    \subcaption[]{}
    \label{fig:N4-example-not-holographic-pmi}
    \end{subfigure}
    \begin{subfigure}{0.39\textwidth}
    \centering
    \begin{tikzpicture}
    \draw[gray!70, thin] (0,0) -- (2,0);
    \draw[gray!70, thin] (0,2) -- (0,0);
    \draw[gray!70, thin] (0,2) -- (2,2);
    \draw[gray!70, thin] (2,2) -- (2,0);
    
    \filldraw[Mcol!80!black] (0,0) circle (3pt);
    \filldraw[Mcol!80!black] (2,0) circle (3pt);
    \filldraw[Mcol!80!black] (0,2) circle (3pt);
    \filldraw[Mcol!80!black] (2,2) circle (3pt);

    \node[] () at (-0.3,2.3) {{\scriptsize $h_{04}$}};
    \node[] () at (2.3,2.3) {{\scriptsize $h_{124}$}};
    \node[] () at (-0.3,-0.3) {{\scriptsize $h_{03}$}};
    \node[] () at (2.3,-0.3) {{\scriptsize $h_{123}$}};

    \node[blue] () at (-0.25,1) {{\scriptsize $0$}};
    \node[blue] () at (1,-0.25) {{\scriptsize $3$}};
    \node[blue] () at (1,2.25) {{\scriptsize $4$}};
    \node[blue] () at (2.25,1) {{\scriptsize $12$}};
    
    \end{tikzpicture} 
    \subcaption[]{}
    \label{fig:N4-example-not-ssa-pmi}
    \end{subfigure}
    \caption{(a) An example of an $\N=4$ SSA-compatible PMI which is not realizable by any holographic graph model. The figure shows the vertices of $\lhp$ which belong to at least one chordless cycle. All the omitted vertices are adjacent to all vertices, and they are: $h_{034}$ and all hyperedges with cardinality four or five. (b) An example of an $\N=4$ KC-PMI which is not SSA-compatible. The figure follows the same conventions as (a), and the omitted vertices are: $h_{012}$, $h_{134}$, $h_{234}$, $h_{034}$ and all hyperedges with cardinality four or five. The edge labels are analogous to the ones in \Cref{fig:N4-example}.}
\end{figure}

We conclude this section with a few comments about yet another possibility. So far we have considered examples of entropy vectors whose PMIs have at least one chordal fine-graining with respect to some CG-map. However, is this always the case? For $\N\geq 4$ it was shown in \cite{He:2022bmi} that there exist faces of the SAC that despite being SSA-compatible do not contain in their interior any entropy vector that can be realized by a graph model (since they all violate at least one instance of MMI). Since these faces are SSA-compatible, the corresponding PMIs are KC-PMIs, and if chordality is a sufficient condition for realizability by a simple tree graph model \cite{sufficiency}, it must then be that each of these PMIs has no chordal fine-graining for any CG-map. 

An example of this situation is shown in \Cref{fig:N4-example-not-holographic-pmi}. Of course, this example was constructed \textit{a posteriori}, using knowledge of holographic entropy inequalities. If instead we are given an SSA-compatible face $\face$ of the SAC for some number of parties $\N$, for which we do not know holographic entropy inequalities, we would like to be able to detect directly from the PMI of the face whether it admits a chordal fine-graining for some CG-map. Again, our goal in this paper is not to answer this difficult question—which is intimately related to the issues already highlighted above—but rather to present an example that suggests how such a test might proceed. To simplify the presentation, rather than focusing on the example in \Cref{fig:N4-example-not-holographic-pmi}, for which $\lhp$ has a rather complicated structure of chordless cycles, we consider instead a face of the $\N=4$ SAC that is not SSA-compatible, i.e., whose interior consists entirely of entropy vectors that violate SSA. These entropy vectors are therefore, obviously, not realizable by graph models, but this choice has the advantage that $\lhp$ contains a single chordless cycle (indicated in \Cref{fig:N4-example-not-ssa-pmi}) and does not alter the logic of the presentation. Importantly, we choose this example so that the corresponding PMI is a KC-PMI, and therefore we can apply the technology of the correlation hypergraph.

For the KC-PMI in \Cref{fig:N4-example-not-ssa-pmi}, an obvious choice of CG-map $\Theta$ with respect to which a chordal fine-graining of $\pmi$ might exist is 
\begin{equation}
    \cgmap:\quad \{\{1\}_1,\{2\}_2,\{3\}_3,\{4\}_4,\{5,0\}_0\}.
\end{equation}
Let $\pmi'$ be the minimum fine-graining of $\pmi$ with respect to $\cgmap$. Its correlation hypergraph $\hpp$ contains in particular the hyperedge $h'_{05}$ and all possible hyperedges that are mapped by the quotient (cf., \Cref{thm:min-fg}) to the hyperedges $h_{03}$ and $h_{04}$ of $\hp$,
\begin{equation}
\label{eq:ssa-incompatible-hyperedges}
    h_{03}\; \rightarrow\; \{h'_{03} \,,\, h'_{35} \,,\, h'_{035}\} \qquad 
    h_{04}\; \rightarrow\; \{h'_{04} \,,\, h'_{45} \,,\, h'_{045}\}.
\end{equation}
We now search for a chordal fine-graining $\widetilde{\pmi}'\succ\pmi'$ by imposing a minimal set of requirements for the structure of $\widetilde{\pmi}'$, which take the form $\mi'=0$ for some $\mi'$ (for clarity, we denote by $\mi'$ the MI instances for $\N'=5$). 

The first necessary requirement for $\widetilde{\pmi}'$ is that the $\bsets$ $\bs{035}$ and $\bs{045}$ are not positive. Furthermore, since in $\pmi$ we have $\mi(0:3)>0$ and $\mi(0:4)>0$, in $\widetilde{\pmi}'$ it must be that $\mi'(05:3)>0$ and $\mi'(05:4)>0$, otherwise $\widetilde{\pmi}'$ is not a fine-graining of $\pmi$ with respect to $\cgmap$. In order for $\bs{035}$ to be non-positive then, we must impose either $\mi'(03:5)=0$ or $\mi'(0:35)=0$, and similarly for $\bs{045}$.\footnote{\,Note that $\mi'(0:35)=0$ is not in contradiction with $\mi(0:3)>0$ because KC only implies $\mi'(0:3)=0$, but in the $\N$-party system, parties 0 and 3 are still correlated if 3 and 5 are correlated in the $\N'$-party system.} Since the minimum fine-graining is obviously invariant under a swap of 0 and 5, we simply make a choice for $\bs{035}$, and we set $\mi'(03:5)=0$ (we search for $\widetilde{\pmi}'$ up to permutations of the parties). This choice also implies (by KC) that $\bs{35}$ is not positive, and therefore that there is no $h'_{35}$ in the correlation hypergraph of $\widetilde{\pmi}'$. On the other hand, we must demand that $\mi'(0:3)>0$, so that there is the hyperedge $h'_{03}$, otherwise the hypergraph of the coarse-graining of $\widetilde{\pmi}'$ does not contain $h_{03}$. Chordality then demands that there is no hyperedge $h'_{04}$, and accordingly we set $\mi'(0:45)=0$. Furthermore, note that we also need $\mi'(0:5)=0$, so that there is no hyperedge $h'_{05}$, otherwise, in the line graph, the corresponding edge ``connects'' $h'_{03}$ to $h'_{45}$, and there is a chordless cycle. This, however, is already guaranteed by KC and our previous requirements.

We now construct the following down-set in the $\N'=5$ MI-poset
\begin{equation}
    \ds' = \pmi'\; \cup\, \downarrow\!\{\mi'(03:5),\mi'(0:45)\}.
\end{equation}
Recall that a down-set is not in general a PMI. In the \textit{lattice of down-sets}\footnote{\,The set of down-sets of any poset is a distributive lattice where meet and join correspond to intersection and union.} of the $\N'$-party system, we can then search for the\footnote{\,The uniqueness of this KC-PMI is not essential here, and was proven in \cite{He:2024xzq}.} minimal KC-PMI $\widetilde{\pmi}'$ such that $\widetilde{\pmi}'\geq \ds'$, which has essential $\bsets$
\begin{equation}
    \ess(\widetilde{\pmi}')=\{\bs{03},\bs{45},\bs{123},\bs{124},\bs{134},\bs{234}\}.
\end{equation}
It is then a simple exercise to verify that $\widetilde{\pmi}'$ is indeed a chordal KC-PMI. However, its coarse-graining under $\cgmap$ is not $\pmi$ but another KC-PMI $\widetilde{\pmi}\succ \pmi$, since, for example, the correlation hypergraph of $\widetilde{\pmi}$ does not have the hyperedge $h_{012}$ (which instead is in $\hp$).

This argument illustrates how one may establish that a given KC-PMI $\pmi$ admits no chordal fine-graining with respect to a specified CG-map, and clarifies the obstruction to its existence. Of course, to conclude that $\pmi$ admits \textit{no} chordal fine-graining whatsoever, one must replicate the same reasoning for every possible CG-map. Since there are infinitely many such maps, this once again encounters the finiteness issue noted above.

\section{Discussion}
\label{sec:discussion}

The main goal of this work has been to take the first steps towards a systematic construction of holographic graph models for a given entropy vector, and to understand under what conditions such a construction is possible independently of any prior knowledge of holographic entropy inequalities. For entropy vectors that satisfy the ``chordality condition'' first introduced in \cite{Hubeny:2024fjn}, we have proposed an efficient algorithm that constructs a simple tree graph model.
Although we have not yet proven that this algorithm always succeeds, and therefore that the chordality condition is sufficient for the realizability of an entropy vector by a holographic graph model, we believe that this is indeed the case and we hope to report on this soon \cite{sufficiency}.

The chordality condition is, however, not necessary for realizability by arbitrary graph models, or even by trees, provided the latter are not simple. While we have not provided a general algorithm for the construction of non-simple tree graph models for a given entropy vector, we have shown by explicit examples that the structure of the set of its possible fine-grainings is crucial for such constructions. Motivated by this, we have expanded the toolkit of the correlation hypergraph introduced in \cite{Hubeny:2024fjn} so that it can naturally accommodate a varying number of parties.

One of the important applications of the technology developed in this work will be a more refined test of the strong form of the conjecture in \cite{Hernandez-Cuenca:2022pst} about the general structure of the holographic entropy cone. This conjecture states that for any number of parties, any extreme ray of the HEC can be realized by a graph model with tree topology. While the conjecture has been confirmed for $\N=5$, a systematic test at $\N=6$—which is particularly interesting because of its very rich structure—is considerably more challenging. Many of the known ERs of the HEC$_6$ are indeed realized by trees, but there also exist extreme rays for which all known graph model realizations contain a ``bulk cycle'', i.e., a cycle of bulk vertices. Importantly, this does not disprove the conjecture, since it does not rule out the existence of alternative tree graph models for the same extreme rays, but initial attempts to construct such tree models have so far failed. One possible explanation is that such models simply do not exist; another, equally plausible, possibility is that they are more difficult to find in practice, owing to the highly nontrivial structure of the relevant fine-graining maps.

While we expect that the techniques developed here can be used to explore this question, it is useful to clarify why the problem remains challenging and how these techniques might be developed further. As we have explained, given an entropy vector $\Svec$ and a CG-map $\cgmap$ it is, in principle, straightforward to construct the set of its fine-grainings with respect to $\cgmap$. However, there exists an infinite set of possible CG-maps, since there is no upper bound on the number of parties that we can introduce. In order to be able to test, and eventually disprove, the conjecture, an important step forward will therefore be to determine whether the problem is actually finite. More precisely, it seems plausible that it is in fact sufficient to verify whether a fine-graining with chordal KC-PMI exists for a finite set of CG-maps, and that, if this is not the case, the introduction of further parties is irrelevant. Moreover, even for a fixed CG-map, as the number of parties grows the explicit construction of the polyhedron of fine-grainings of $\Svec$ with respect to $\cgmap$ becomes fundamentally unapproachable in practice. For these reasons, a more promising approach might be one that is focused on a direct search for ``chordal fine-grainings'' by studying more explicitly the transformation of arbitrary hypergraphs and their line graphs under fine-grainings. Indeed, given any hypergraph, one can simply use a linear program to determine whether it is a correlation hypergraph or not, and thus whether it defines an admissible fine-graining. We plan to explore this avenue in the near future \cite{graphoidal}.

\acknowledgments

V.H. would like to thank Sergio Hern\'andez-Cuenca for useful conversations.  V.H. has been supported in part by the U.S. Department of Energy grant DE-SC0009999 and by funds from the University of California.   M.R. acknowledges support from UK Research and Innovation (UKRI) under the UK government’s Horizon Europe guarantee (EP/Y00468X/1).
V.H. acknowledges the hospitality of the Kavli Institute for Theoretical Physics (KITP) during early stages of this work and of the Aspen Center for Physics (supported by National Science Foundation grant PHY-2210452). The authors acknowledge the hospitality of the Centro de Ciencias de Benasque Pedro Pascual during the workshop ``Gravity - New quantum and string perspectives''.
M.R. would like to thank the hospitality of QMAP and the University of California, Davis, during various stages of this work.

There is no underlying data associated with this work.

For the purpose of open access, the authors have applied a Creative Commons Attribution (CC BY) licence to any Author Accepted Manuscript version arising from this submission.

\bibliography{general_bib}
\bibliographystyle{utphys}

\end{document}